\documentclass[10pt,draftcls,onecolumn,twoside]{IEEEtran}

\ifCLASSINFOpdf

\else

\fi

\usepackage{amsmath}
\usepackage{amsthm}
\usepackage{amssymb}
\usepackage{calc}
\usepackage{graphicx}
\usepackage{threeparttable}
\usepackage{cite}
\usepackage{hyperref}
\usepackage{bm}
\usepackage{color}
\usepackage{mathtools}

\newcount\Comments  
\usepackage{color}
\definecolor{darkgreen}{rgb}{0,0.5,0}
\definecolor{purple}{rgb}{1,0,1}
\newcommand{\kibitz}[2]{\ifnum\Comments=0\textcolor{#1}{#2}\fi}

\theoremstyle{plain}
\newtheorem{thm}{Theorem}[section] 

\theoremstyle{definition}
\newtheorem{lem}[thm]{Lemma}
\newtheorem{prop}[thm]{Proposition}
\newtheorem{remark}[thm]{Remark}

\newtheorem{cor}[thm]{Corollary}
\newtheorem{example}[thm]{Example}

\usepackage{amsmath}
\usepackage{amsthm}
\usepackage{amssymb}
\usepackage{calc}
\usepackage{graphicx}
\usepackage{threeparttable}
\usepackage{cite}
\usepackage{hyperref}
\usepackage{bm}
\usepackage{color}
\usepackage{mathtools}

\begin{document}
%
\title{Scaling the Kalman filter for large-scale\\traffic estimation}
%
%
%

\author{Ye~Sun and~Daniel~B.~Work
\thanks{The authors are with the Department of Civil and Environmental Engineering and Coordinated Science Laboratory, University of Illinois at Urbana-Champaign, Urbana,
IL, 61801 USA (e-mail: yesun@illinois.edu; dbwork@illinois.edu).}
\thanks{}}

%
%

\markboth{}%
{}
%



\maketitle

\begin{abstract}
This work introduces a scalable filtering algorithm for multi-agent traffic estimation. Large-scale networks are spatially partitioned into overlapping road sections. The traffic dynamics of each section is given by the \emph{switching mode model} (SMM) using a conservation principle, and the traffic state in each section is estimated by a local agent. In the proposed filter, a consensus term is applied to promote inter-agent agreement on overlapping sections. The new filter, termed a \emph{(spatially) distributed local Kalman consensus filter} (DLKCF), is shown to maintain \emph{globally asymptotically stable} (GAS) mean error dynamics when all sections switch among observable modes. When a section is unobservable, we show that the mean estimate of each state variable in the section is ultimately bounded, which is achieved by exploring the interaction between the properties of the traffic model and the measurement feedback of the filter. Based on the above results, the boundedness of the mean estimation error of the DLKCF under switching sequences with observable and unobservable modes is established to address the overall performance of the filter. Numerical experiments show the ability of the DLKCF to promote consensus, increase estimation accuracy compared to a local filter, and reduce the computational load compared to a centralized approach. 
\end{abstract}


%
\IEEEpeerreviewmaketitle

\section{Introduction}
%
%
%
%

\IEEEPARstart{D}{espite} important advances in sensing and computation, real-time traffic estimation problems are still open to a number of critical issues including: (\textit{i}) the entire state of the transportation network is too large (usually of order at least $10^5$) for the estimators to scale in real time; (\textit{ii}) few results are available that provide a theoretical analysis of the performance of traffic estimation algorithms, and (\textit{iii}) the non-observability of the traffic model is inevitable due to the existence of shocks and the sparsity of sensor measurements. This work aims at designing a scalable distributed traffic estimation algorithm to address issues (\textit{i}) and (\textit{ii}) with specific care of issue (\textit{iii}). The large-scale network is partitioned into overlapping sections, and the traffic density on each section evolves according to a conservation law traffic model. The density is estimated by a cheap commodity computer (referred hereafter as an agent) associated with each section. However, without coordination among agents, estimates provided by different agents inevitably disagree
on the shared boundaries due to model and measurement errors. This potentially leads to problems where applications computed based on traffic estimates (e.g., navigation, traffic control) produce disparate results depending on which agents provide the estimates. To promote agreement between neighboring agents on their shared states, each agent shares sensor data and estimates with its neighbors, and a consensus term is introduced. The filter trades global optimality in favor of scalability both in terms of communication and computation, thus the proposed filter is suboptimal. Regardless, the proposed filter has performance guarantees when traffic state is observable or unobservable, as well as when the system switches between observable and unobservable modes. Specifically, in unobservable scenarios, the physical properties of the traffic model (i.e., mass conservation and a flow-density relationship) are combined with the measurement feedback in the correction step of the filter to analyse the theoretical performance of the filter.

Research on collaborative information processing is driven by the broad applications of multi-agent systems. A complete communication network with all-to-all links is required in the \emph{decentralized Kalman filter} \cite{Rao1993}, or relaxed in the \emph{channel filter} \cite{grime1994data} for the fixed tree communication topology. Recently, the application of consensus strategies in distributed estimation is widely studied to promote agreement on estimates among agents \cite{olfati2007distributed,Olfati-SaberCDC2009,demetriou2013adaptive}, and/or to reconstruct sensor data not directly accessible through purely sharing measurements with neighbors \cite{bai2011distributed,Kamal2013,battistelli2015consensus}, thus approximating the central estimator. To ensure the stability of the estimators, each local system is assumed to be observable (or detectable) in \cite{olfati2007distributed,Olfati-SaberCDC2009,demetriou2013adaptive}, or the full system observability is only achieved given all the sensor data in the network \cite{bai2011distributed,Kamal2013,battistelli2015consensus}. A common feature of \cite{Rao1993,grime1994data,olfati2007distributed,Olfati-SaberCDC2009,demetriou2013adaptive,bai2011distributed,Kamal2013,battistelli2015consensus} is that all agents estimate the same full state of dimension $n$, which may not scale in large-scale traffic networks since the complexity of the \textit{Kalman filter} (KF) is $O(n^3)$. Moreover, the non-observability of the traffic model cannot be resolved even if all measurements throughout the network are fused.

There are also notable works on scalable distributed estimation algorithms where each agent estimates (or performs computation on) a small subset of the full state. Specifically in \cite{khan2008distributing,stankovic2009consensus}, the large-scale state vector is partitioned into overlapping local states of dimension $n_l\ll n$, and the computation task is distributed across local agents.  In \cite{khan2008distributing}, the cross-correlation of neighboring agents is incorporated in the estimation error covariance at the expense of requiring a $O(n_l^4)$ complexity at each local agent. However, the stability of the proposed estimator is not analyzed. In \cite{stankovic2009consensus}, a consensus term is designed to help each local agent reconstruct the estimates of other local states, and is analyzed only when all local filters are detectable and have achieved a steady state. Other relevant treatments include moving-horizon estimation \cite{farina2010moving} and distributed Kriged Kalman filtering \cite{cortes2009distributed}. However, they either require extensive communication, or rely on the statistics of random fields which are not directly applicable for traffic dynamics. Moreover, the estimators~\cite{farina2010moving,cortes2009distributed} are not analysed when the model is unobservable.

A number of sequential estimation algorithms have been applied for traffic monitoring. 
Due to the non-linearity and non-differentiability \cite{Blandin2012OnSequential} of the nonlinear hyperbolic conservation law used to describe traffic, few results exist which rigorously prove the performance of the proposed estimators. In \cite{Munoz2006TRR}, the discretized conservation law is transformed to a switched linear system known as the \emph{switching mode model} (SMM), and the observability of each mode is analyzed. The properties of the error dynamics of a Luenberger observer in various modes of the SMM is given in \cite{Carlos2011}, which inspires this work and is extended in~\cite{CarlosConstrainedCTM,zeroual2015calibration}. An recent overview of sequential estimation for scalar traffic models is given in \cite{Blandin2012OnSequential}. Another interesting line of work focuses on designing estimators and associated numerical schemes directly for conservation laws, see \cite{banks2012estimation} and references therein.

The main contribution of this article is the design and analysis of a (spatially) \emph{distributed local Kalman consensus filter} (DLKCF) (Section~\ref{sec:DLKCF}) with provable performance and neighbor consistency. The DLKCF is proposed to estimate traffic densities on large freeways, with the system dynamics described by the SMM (Section \ref{sec:MacroModeling}). We analyse the performance of the DLKCF under various observability scenarios, yielding three main results: (\textit{i}) the dynamics of the mean estimation error is \emph{globally asymptotically stable} (GAS) when all sections switch among observable modes of the SMM (Section \ref{sub:ObservableConvergence}); (\textit{ii}) when a section switches among unobservable modes, the mean estimate is ultimately bounded inside a physically meaningful interval (Section \ref{sub:UnobservableBound}); and (\textit{iii}) the mean estimation error is upper bounded for sections that switch among observable and unobservable modes, provided a minimum residence time in the observable mode(s) is satisfied (Section \ref{section:switching}). The above results focus on the mean estimate and are derived based on the stability (or bounded partitions) of the estimation error covariances (given in the lemmas proceeding the propositions). Numerical results (Section \ref{sec: NumericalExperiments}) show the effect of the consensus term on reducing disagreement between estimates given by neighboring agents (with \textasciitilde $50\%$ reduction), and that the DLKCF outperforms a purely local KF on estimation accuracy. 

Compared to our preliminary work \cite{SunWorkCDC14}, the main extension is to prove the overall performance of the DLKCF under switches among observable and unobservable modes. The DLKCF is also modified to be scalable both in the sense of computation (i.e., with cubic computational complexity in the local dimension) and communication (i.e., each agent only communicates with its one-hop neighbors, and the global communication topology is not needed). Specifically in Proposition \ref{Prop:ObservableGUAS}, the upper bound of the scaling factor in the consensus gain is modified to depend only on the information provided by one-hop neighbors. 


\section{Scalar macroscopic traffic modeling}
\label{sec:MacroModeling}

\subsection{Cell transmission model}\label{sec:CTM}

The classical conservation law describing the evolution of traffic density $\rho(t,x)$ on a road at location $x$ and time $t$ is the \emph{Lighthill-Whitham-Richards} \emph{partial differential equation} (LWR PDE) \cite{richards1956swh,Lighthill1955}:
\begin{align}
\partial_t\rho+\partial_x\mathfrak{F}(\rho)=0.\label{eq:LWRPDE}
\end{align}
The function $\mathfrak{F}(\rho)=\rho \mathfrak{v}(\rho)$ is called the flux function, where $\mathfrak{v}(\rho)$ is an empirical velocity function used to close the model. The triangular flux function \cite{DaganzoCell95} used in this work is given by
\begin{align}
\mathfrak{F}(\rho)= \left\{ \begin{array}{ll}
\rho v_{\text{m}}& \textrm{if $\rho \in [0,\varrho_{\text{c}}]$}\\
w(\varrho_{\text{m}}-\rho)& \textrm{if $\rho \in [\varrho_{\text{c}},\varrho_{\text{m}}]$,}\
\end{array} \right.\label{eq:FluxTraingular}
\end{align}
where $w=\frac{\varrho_{\text{c}} v_{\text{m}}}{\varrho_{\text{m}}-\varrho_{\text{c}}}$, $v_{\text{m}}$ denotes the \emph{freeflow speed} and $\varrho_{\text{m}}$ denotes the \emph{maximum density}. The variable $\varrho_{\text{c}}$ is the \emph{critical density} at which the maximum flux is realized. For the triangular fundamental diagram, the flux function has different slopes in \emph{freeflow} ($0<\rho \leq\varrho_{\text{c}}$) and \emph{congestion} ($\varrho_{\text{c}}<\rho \le \varrho_{\text{m}}$). In freeflow, the slope is $v_{\text{m}}$, and in congestion, it is $w$.

The \emph{cell transmission model} (CTM) \cite{DaganzoCell95} is a discretization of \eqref{eq:LWRPDE} and \eqref{eq:FluxTraingular} using a Godunov scheme. Consider a discretization grid defined by a space step $\Delta x >0$ and a time step $\Delta t >0$. Let $l$ index the cell defined by $x\in [l\Delta x, (l+1)\Delta x)$, and denote as $\rho^l_{k}$ the density at time $k\Delta t$ in cell $l$, where $k\in \mathbb{N}$ and $l\in \mathbb{N}^{+}$. The discretized model \eqref{eq:LWRPDE} becomes
\begin{align}
\rho^l_{k+1}=\rho^l_{k}+\frac{\Delta t}{\Delta x}\left(\mathfrak{f}(\rho^{l-1}_{k},\rho^l_{k})-\mathfrak{f}(\rho^l_{k},\rho^{l+1}_{k})\right),\label{eq:DiscretizedLWR}
\end{align}
where $\mathfrak{f}(\rho^{l-1}_{k},\rho^l_{k})$ is the flux between cell $l-1$ and $l$:
\begin{align}
\mathfrak{f}(\rho^{l-1}_{k},\rho^l_{k})=\min \{v_{\text{m}}\rho^{l-1}_{k}, w(\varrho_{\text{m}}-\rho^l_{k}), q_{\text{m}}\},\label{eq:qflow}
\end{align}
where $q_{\text{m}}$ is the maximum flow given by $q_{\text{m}}=v_{\text{m}}\varrho_{\text{c}}$. Note that if the \emph{Courant–-Friedrichs–-Lewy} (CFL) condition is satisfied, the solution of the CTM converges in $L^1$ to the weak solution of the LWR PDE as $\Delta x\rightarrow 0$.

\subsection{Switching mode model}\label{subsec:SMM}

In the SMM \cite{Munoz2006TRR}, \eqref{eq:DiscretizedLWR} is written as a hybrid linear system whose system dynamics switches among different modes depending on the state of the boundary cells.

Consider a freeway section with $n$ cells with the state variable at time step $k\in \mathbb{N}$ defined as $\rho_k=\left(\rho^1_{k},\cdots,\rho^n_{k}\right)^T$. The SMM is derived from~\eqref{eq:LWRPDE} under the main assumption that there is at most one transition between freeflow and congestion in each section. From an estimation point of view, the SMM also assumes the road network is partitioned into sections with sensors located in the first and last cell, such that the densities $\rho^1_{k}$ and $\rho^n_{k}$ are directly measured. Finally, the SMM assumes the boundary density measurements are sufficiently accurate to distinguish between four of the five modes described next, but they cannot determine the precise location or direction of a shock.

Given the assumption of at most one transition in a section, the SMM may switch between the following five modes: (\textit{i}) \textit{freeflow--freeflow} (FF), in which all cells in the section are in freeflow; (\textit{ii}) \textit{congestion--congestion} (CC), in which all cells in the section are in congestion; (\textit{iii}) \textit{congestion--freeflow} (CF), in which the cells in the upstream part of the section (i.e., the cells in the upstream side of the transition between freeflow and congestion based on the direction of travel) are congested, and the cells in the downstream part are in freeflow; (\textit{iv}) \textit{freeflow--congestion 1} (FC1), in which the upstream part of the section is in freeflow, the downstream part is in congestion, and the shock has positive velocity or is stationary; and (\textit{v}) \textit{freeflow--congestion 2} (FC2), in which the upstream part of the section is in freeflow, the downstream part is in congestion, and the shock has negative velocity. Note the boundary sensors cannot distinguish between modes (\textit{iv}) and (\textit{v}).

In each mode stated above, the traffic state $\rho_k$ evolves with linear dynamics, forming a hybrid system:
\begin{align}
\rho_{k+1}=A_{\sigma(k)}^{s(k)}\rho_k+B^{\rho,s(k)}_{\sigma(k)}\bm{1}\varrho_{\text{m}}+B^{q,s(k)}_{\sigma(k)}\bm{1}q_{\text{m}}, \label{eq:SMMeq}
\end{align}
where $\bm{1}$ is the vector of all ones, and $A_{\sigma(k)}^{s(k)}$, $B^{\rho,s(k)}_{\sigma(k)}$, $B^{q,s(k)}_{\sigma(k)}\in \mathbb{R}^{n\times n}$ are to be defined precisely later. The index $\sigma(k)\in\mathcal{S}$ where $\mathcal{S}=\{\text{FF, CC, CF, FC1, FC2}\}$ is the set of the five modes, and $s(k)\in\{1,\cdots,n-1\}$ is the index introduced to precisely locate the transition between freeflow and congestion when it exists. We say $s(k)=l$ when the transition occurs between cell $l$ and $l+1$.

For all $p\in\{1,2,\cdots,n-1\}$, define $\Theta_p\in\mathbb{R}^{p\times p}$ and $\Delta_p\in\mathbb{R}^{p\times p}$ by their $(i,j)^{\textrm{th}}$ entries as
\begin{align}
\Theta_p(i,j)= \left\{ \begin{array}{ll}
1-\frac{v_{\text{m}}\Delta t}{\Delta x}& \textrm{if $i=j$}\\
\frac{v_{\text{m}}\Delta t}{\Delta x}& \textrm{if $i=j+1$}\\
0& \textrm{otherwise,}
\end{array} \right.\notag
\end{align}
\begin{align}
\Delta_p(i,j)= \left\{ \begin{array}{ll}
1-\frac{w\Delta t}{\Delta x}& \textrm{if $i=j$}\\
\frac{w\Delta t}{\Delta x}& \textrm{if $i=j-1$}\\
0& \textrm{otherwise.}
\end{array} \right.\notag
\end{align}

In the FF mode, the mode index $\sigma=\text{FF}$, and the transition does not exist. The explicit forms of $A_{\sigma}^s$, $B^{\rho,s}_{\sigma}$, and $B^{q,s}_{\sigma}$ are:
\begin{align}
A_{\text{FF}}=\left(
  \begin{array}{cc}
  1&\bm{0}_{1,n-1}\\
  \left(
  \begin{array}{c}
  \frac{v_{\text{m}}\Delta t}{\Delta x}\\
  \bm{0}_{n-2,1}
  \end{array}
\right)&\Theta_{n-1}
  \end{array}
\right), B^{\rho}_{\text{FF}}=B^{q}_{\text{FF}}=\bm{0}_{n,n}, \notag
\end{align}
where $\bm{0}_{n,m} \in \mathbb{R}^{n\times m}$ which is zero everywhere. In the CC mode, the transition also does not exist, and
\begin{align}
A_{\text{CC}}=\left(
  \begin{array}{cc}
  \Delta_{n-1}&\left(
  \begin{array}{c}
  \bm{0}_{n-2,1}\\
  \frac{w\Delta t}{\Delta x}
  \end{array}
\right)\\
  \bm{0}_{1,n-1}&1
  \end{array}
\right), B^{\rho}_{\text{CC}}=B^{q}_{\text{CC}}=\bm{0}_{n,n}.  \notag
\end{align}
The FF (resp. CC) mode is observable given density measurement of the downstream (resp. upstream) cell.

In the CF mode, the mode index $\sigma=\text{CF}$, and
\begin{align}
A_{\text{CF}}^{s}=\left(
  \begin{array}{cc}
  \Delta_{s}&\bm{0}_{s,n-s}\\
  \bm{0}_{n-s,s}&\Theta_{n-s}
  \end{array}
\right), B^{\rho,s}_{\text{CF}}=\bm{0}_{n,n}+\frac{w\Delta t}{\Delta x}E_{s,s}, \notag
\end{align}
\begin{align}
B^{q,s}_{\text{CF}}=\bm{0}_{n,n}-\frac{\Delta t}{\Delta x}E_{s,s+1}+\frac{\Delta t}{\Delta x}E_{s+1,s+1}, \notag
\end{align}
where $E_{i,j}$ are matrices that are zero everywhere but the $(i,j)^{\textrm{th}}$ entry, which is one. Note that $s$ may take any value in $\{1,\cdots,n-1\}$, depending on the location of the center of the expansion fan connecting the congested and freeflow states. The CF mode is observable given density measurements of the upstream and downstream cells.

In the two FC modes, define $\hat{\Theta}_p$ and $\hat{\Delta}_p$ as follows:
\begin{align}
\hat{\Theta}_p=\left\{ \begin{array}{ll}
\left(
  \begin{array}{cc}
    1&\bm{0}_{1,p}\\
    \left(
  \begin{array}{c}
  \frac{v_{\text{m}}\Delta t}{\Delta x}\\
  \bm{0}_{p-1,1}
  \end{array}
\right)&\Theta_p
  \end{array}
\right)&\textrm{if $p \in\{1,\cdots,n-1\}$,} \\
1&\textrm{if $p=0$,}
\end{array} \right.\notag
\end{align}
and
\begin{align}
\hat{\Delta}_p=\left\{ \begin{array}{ll}
\left(
  \begin{array}{cc}
    \Delta_p& \left(
  \begin{array}{c}
  \bm{0}_{1,p-1}\\
  \frac{w\Delta t}{\Delta x}
  \end{array}
\right)\\
   \bm{0}_{1,p}&1
  \end{array}
\right)&\textrm{if $p \in\{1,\cdots,n-1\}$,} \\
1&\textrm{if $p=0$.}
\end{array} \right.\notag
\end{align}
When $\sigma=\text{FC1}$ and $s\in\{1,\cdots,n-2\}$, or $\sigma=\text{FC2}$ and $s\in\{2,\cdots,n-1\}$, the matrices $A_{\sigma}^s$, $B^{\rho,s}_{\sigma}$, and $B^{q,s}_{\sigma}$ read:
\begin{align}
A_{\sigma}^s=\left(
  \begin{array}{ccc}
\hat{\Theta}_{\tilde{s}-1}&\bm{0}_{\tilde{s},1}&\bm{0}_{\tilde{s},\bar{s}}\\
\left(
  \begin{array}{cc}
  \bm{0}_{1,\tilde{s}-1}&\frac{v_{\text{m}}\Delta t}{\Delta x}
  \end{array}
\right)&1&\left(
  \begin{array}{cc}
  \frac{w\Delta t}{\Delta x}&\bm{0}_{1,\bar{s}-1}
  \end{array}
\right)\\
\bm{0}_{\bar{s},\tilde{s}}&\bm{0}_{\bar{s},1}&\hat{\Delta}_{\bar{s}-1}
\end{array}
\right),  \notag
\end{align}
\begin{align}
B^{\rho,s}_{\sigma}=\left(
  \begin{array}{cc}
  \bm{0}_{\tilde{s}+1,\tilde{s}+1}&\left(
  \begin{array}{cc}
  \bm{0}_{\tilde{s},1}&\bm{0}_{\tilde{s},\bar{s}-1}\\
  -\frac{w\Delta t}{\Delta x}&\bm{0}_{1,\bar{s}-1}
  \end{array}
\right)\\
  \bm{0}_{\bar{s},\tilde{s}+1}&\bm{0}_{\bar{s},\bar{s}}
  \end{array}
\right),  B^{q,s}_{\sigma}=\bm{0}, \notag
\end{align}
where for $\sigma=\text{FC1}$ we have $\tilde{s}=s$ and $\bar{s}=n-s-1$, and for $\sigma=\text{FC2}$ we have $\tilde{s}=s-1$ and $\bar{s}=n-s$. When $\sigma=\text{FC1}$ and $s=n-1$, we have $A_{\sigma}^s=\textrm{diag}(\hat{\Theta}_{n-2},1)$ (i.e., with $\hat{\Theta}_{n-2}$ and 1 on the diagonal), and $A_{\sigma}^s=\textrm{diag}(1,\hat{\Delta}_{n-2})$ when $\sigma=\text{FC2}$ and $s=1$. For both cases, $B^{\rho,s}_{\sigma}=B^{q,s}_{\sigma}=\bm{0}_{n,n}$. The two FC modes are not observable unless density measurements of all the cells are available, which does not occur in practical discretizations of road networks. We classify the state transition matrices according to the observability of the SMM. Define the matrix set with state transition matrices associated with the observable and unobservable modes as $\mathcal{A}_{\text{O}}=\left\{A_{\text{FF}},A_{\text{CC}},A_{\text{CF}}^s\left|s\in\left\{1,2,\cdots,n-1\right\}\right.\right\}$ and $\mathcal{A}_{\text{U}}=\left\{A_{\text{FC1}}^s,A_{\text{FC2}}^s\left|s\in\left\{1,2,\cdots,n-1
\right\}\right.\right\}$, respectively. The set of all state transition matrices is thus defined as $\mathcal{A}=\mathcal{A}_{\text{O}}\cup \mathcal{A}_{\text{U}}$. 

For consistency with the shock dynamics in~\eqref{eq:LWRPDE}, the allowed mode transitions are enumerated in the graph constrained-SMM \cite{CarlosConstrainedCTM}. The results in this article hold for the graph constrained and more general switching sequences.

The observability results of the SMM for individual modes can be derived directly from computing the rank of the observability matrix for each mode given \eqref{eq:SMMeq} and the observation equation $z_k=H_k\rho_k$, where $z_k$ is the measurement, and $H_k$ is the appropriate output matrix. From a physical viewpoint, the non-observability of the SMM is due to the irreversibility of the LWR PDE given the available sensor measurements in the presence of shocks, and is not due to the discretization.


\section{Distributed local Kalman consensus filter}\label{sec:DLKCF}
\subsection{Kalman filter}\label{subsec:KF}

In this subsection, we briefly review the KF and introduce notations needed later in the proposed filter. Consider the linear time-varying system
\begin{align*}
\rho_{k+1}&=A_k\rho_k+\omega_k \textrm{, } \rho_k\in \mathbb{R}^n,\\
z_k&=H_k\rho_k+v_k \textrm{, }z_k\in \mathbb{R}^{m},
\end{align*}
where $\omega_k\sim \mathcal{N}(0,Q_k)$ and $v_k\sim \mathcal{N}(0,R_k)$ are the white Gaussian model and measurement noise. Given the sensor data up to time $k$ denoted by $\mathcal{Z}_k=\{z_0,\cdots,z_k\}$, the \emph{prior estimate} and \emph{posterior estimate} of the state can be expressed as $\rho_{k|k-1}=\mathbb{E}[\rho_k|\mathcal{Z}_{k-1}]$ and $\rho_{k|k}=\mathbb{E}[\rho_k|\mathcal{Z}_k]$, respectively. Let $\eta_{k|k-1}=\rho_{k|k-1}-\rho_k$ and $\eta_{k|k}=\rho_{k|k}-\rho_k$ denote the prior and posterior estimation errors. The estimation error covariance matrices associated with $\rho_{k|k-1}$ and $\rho_{k|k}$ are given by $\Gamma_{k|k-1}=\mathbb{E}[\eta_{k|k-1}\eta_{k|k-1}^T|\mathcal{Z}_{k-1}]$ and $\Gamma_{k|k}=\mathbb{E}[\eta_{k|k}\eta_{k|k}^T|\mathcal{Z}_k]$. The KF sequentially computes $\rho_{k|k}$ from $\rho_{k-1|k-1}$ as follows:
\begin{align}
&\textrm{Prediction:} \left\{ \begin{array}{lc}
\rho_{k|k-1}=A_{k-1}\rho_{k-1|k-1}\\
\Gamma_{k|k-1}=A_{k-1}\Gamma_{k-1|k-1}A_{k-1}^{T}+Q_{k-1},
\end{array}\right.\notag
\end{align}
\begin{align}
&\textrm{Correction:} \left\{ \begin{array}{lc}
\rho_{k|k}=\rho_{k|k-1}+K_{k}(z_{k}-H_{k}\rho_{k|k-1})\\
\Gamma_{k|k}=\Gamma_{k|k-1}-K_kH_k\Gamma_{k|k-1}\\
K_k=\Gamma_{k|k-1}H_{k}^{T}(R_{k}+H_{k}\Gamma_{k|k-1}H_{k}^{T})^{-1}.
\end{array}\right.\notag
\end{align}

\subsection{Distributed local Kalman consensus filter}\label{subsec:DLKCF}
In the DLKCF, the discretized freeway network is spatially partitioned into $N$ overlapping sections, with each section estimated by its own agent. Neighboring agents are allowed to exchange measurements and state estimates to reduce disagreement on shared cells. For the one-dimensional freeway, the set of neighbors of section $i$ is given by
\begin{align}
\mathcal{N}_{i}= \left\{ \begin{array}{ll}
\{i+1\}& \textrm{if $i=1$}\\
\{i-1,i+1\}& \textrm{if $i\neq1$, and $i\neq N$}\\
\{i-1\}& \textrm{if $i=N$}.
\end{array} \right.\label{eq:DLKCFneighbor}
\end{align}
Hence, the Laplacian associated with the communication topology is a tridiagonal matrix. The reader is referred to Figure \ref{Fig:DLKCF_DLKCF_network}d for an illustration of the partitioning of a roadway into overlapping sections. In Figure \ref{Fig:DLKCF_DLKCF_network}d, the freeway is partitioned into seven sections with 28 cells and four sensors in each section, and there are 10 cells in each overlapping region between neighboring sections. Except for the agents associated with the first and last sections, each agent obtains direct measurements from the two boundary sensors in the section. For the other two sensors, their measurements are collected by the neighbors and sent to the agent.

Given the SMM, the system dynamics of section $i$ reads
\begin{align}
\rho_{i,k+1}=A_{i,k}\rho_{i,k}+B^{\rho}_{i,k}\bm{1}\varrho_{\text{m}}+B^{q}_{i,k}\bm{1}q_{\text{m}}+\omega_{i,k}, \label{eq:DLKCF_Dynamics}
\end{align}
where $A_{i,k}\in \mathcal{A}$, $\rho_{i,k}\in \mathbb{R}^{n_i}$ and $\omega_{i,k}\sim \mathcal{N}(0,Q_{i,k})$ is the white Gaussian model noise. Note that in \eqref{eq:DLKCF_Dynamics} and for the remainder of the article subscript $k$ for $A$, $B^{\rho}$, and $B^{q}$ combines the effect of $\sigma(k)$ and $s(k)$, and subscript $i\in \{1,2,\cdots,N\}$ is the section index. The sensors are spatially distributed in the road network and measure the traffic density at their locations. Hence, if the $p^{\textrm{th}}$ sensor directly connected to agent $i$ is located at the $l^{\textrm{th}}$ cell in section $i$, the $p^{\textrm{th}}$ row of $H^{i}_{i,k}$ is given by $(0,\cdots,0,1,0,\cdots,0)$ where the $l^{\textrm{th}}$ element is 1. The observation equation modeled at agent $i$ that corresponds to the sensor data obtained by the sensors directly connected to agent $j$ is given by:
\begin{align}
z^{i}_{j,k}=H^{i}_{j,k}\rho_{i,k}+v^{i}_{j,k}, \text{ } z^{i}_{j,k}\in \mathbb{R}^{m_j^i}, \text{} j\in\mathcal{J}_i=\mathcal{N}_i\bigcup\{i\}, \label{eq:sensorModel}
\end{align}
where $v^{i}_{j,k}\sim \mathcal{N}(0,R^{i}_{j,k})$. Note that the sensor data $z_{j,k}^i$ for $j\in \mathcal{N}_i$ is obtained through receiving measurements from agent $j$. Consequently, through communication each agent possesses columnized sensor data $z_{i,k}=\textrm{col}_{j\in \mathcal{J}_i}(z^{i}_{j,k})$ with noise $v_{i,k}=\textrm{col}_{j\in \mathcal{J}_i}(v^{i}_{j,k})$ and a corresponding columnized output matrix $H_{i,k}=\textrm{col}_{j\in \mathcal{J}_i}(H^{i}_{j,k})$, as well as a block diagonal measurement error covariance $R_{i,k}=\textrm{diag}_{j\in \mathcal{J}_i}(R^{i}_{j,k})$.

For $j\in \mathcal{N}_i$, denote the dimension of the overlap between section $i$ and $j$ as $n_{i,j}$, and define the projection $\hat{I}_{i,j}$ as
\begin{align}
\hat{I}_{i,j}= \left\{ \begin{array}{ll}
\left(
  \begin{array}{cc}
    I_{n_{i,j}}&\bm{0}_{n_{i,j},n_i-n_{i,j}}
  \end{array}
\right)& \textrm{if $j=i-1$}\\
\left(
  \begin{array}{cc}
    \bm{0}_{n_{i,j},n_i-n_{i,j}}&I_{n_{i,j}}
  \end{array}
\right)& \textrm{if $j=i+1$,}
\end{array} \right.\label{eq:I_hat_ij}
\end{align}
where $I_{n_{i,j}} \in \mathbb{R}^{n_{i,j}}$ is the identity matrix. The quantity $\hat{I}_{i,j}\rho_{i,k}$ selects the state of section $i$ that overlaps with section $j$. A consensus term is added to the correction step of the DLKCF to promote agreement on estimates among neighboring agents on their shared overlapping regions. The prediction and correction steps of the DLKCF for agent $i$ reads
\begin{flalign}
&\left\{ \begin{array}{l}
\rho_{i, k|k-1}=A_{i,k-1}\rho_{i,k-1|k-1}\\
\Gamma_{i,k|k-1}=A_{i,k-1}\Gamma_{i,k-1|k-1}A_{i,k-1}^{T}+Q_{i,k-1}
\end{array}\right.\label{eq:DLKCFforecast}\\
&\left\{ \begin{array}{l}
\rho_{i,k|k}=\rho_{i,k|k-1}+K_{i,k}\left(z_{i,k}-H_{i,k}\rho_{i,k|k-1}\right)\\
\quad+\sum_{j\in \mathcal{N}_i}C^{j}_{i,k}\left(\hat{I}_{j,i}\rho_{j,k|k-1}-\hat{I}_{i,j}\rho_{i,k|k-1}\right)\\
\Gamma_{i,k|k}=\Gamma_{i,k|k-1}-K_{i,k}H_{i,k}\Gamma_{i,k|k-1}\\
K_{i,k}=\Gamma_{i,k|k-1}H_{i,k}^{T}(R_{i,k}+H_{i,k}\Gamma_{i,k|k-1}H_{i,k}^{T})^{-1},
\end{array}\right.\label{eq:DLKCFAnalysis}
\end{flalign}
where $C^{j}_{i,k}$ is the \emph{consensus gain} of agent $i$ associated with neighbor $j$ at time step $k$, and for simplicity we drop the middle two terms in \eqref{eq:DLKCF_Dynamics} which are deterministic. Our choice of the consensus gain is given by:
\begin{align}
C^{j}_{i,k}= \left\{ \begin{array}{ll}
\gamma^j_{i,k}\Gamma_{i,k|k-1}\hat{I}_{i,j}^T & \sigma(k)\in \{\text{FF, CC, CF}\}\\
\bm{0}_{n_i,n_{i,j}}& \sigma(k)\in\{\text{FC1, FC2}\},
\end{array} \right.\label{eq:DLKCFconsensusgain}
\end{align}
where $\gamma^j_{i,k}=\gamma^i_{j,k}$ is a sufficiently small scaling factor, with $\gamma^j_{i,k}<\gamma_{i,k}^{*}$ for all $i$, $j\in\mathcal{N}_i$ and $k$. The explicit form of $\gamma_{i,k}^{*}$ will be given in Proposition \ref{Prop:ObservableGUAS} to ensure the unbiasedness of the DLKCF. Under unobservable modes, the consensus term is turned off. According to \eqref{eq:DLKCFconsensusgain}, the consensus term is designed based on the belief of the current estimation accuracy and the disparity among neighbors on the prior estimate, thus promoting agreement on the state estimates. Although an arbitrary convex combination of the estimates between neighboring agents may considerably reduce disagreement, it may largely increase the estimation error. Hence, the scaling factor needs to be carefully designed to ensure stability of the DLKCF.

\begin{remark}\label{RM:Optimality}
Given the consensus gain \eqref{eq:DLKCFconsensusgain}, one may derive the optimal Kalman gain $K_{i,k}$ through minimizing $\textrm{tr}(\Gamma_{i,k|k})$ in a similar way as Theorem 1 in \cite{Olfati-SaberCDC2009}, thus yielding an optimal DLKCF which incorporates the cross-correlations among different agents in the estimation error covariance. However, the optimal DLKCF has large communication requirements (i.e., the cross covariance $\Gamma^j_{i,k|k}$ between section $i$ and $j$ needs to be computed by agent $i$ for all $j\in\{1,\cdots,N\}$) that conflicts the goal of designing a scalable traffic estimation algorithm. Moreover, when cross-correlation terms are included, a section which is always observable can have an unbounded error covariance if the neighboring section is unobservable, as detailed in Appendix \ref{ap:crosscov_unstb}. Instead, the Kalman gain $K_{i,k}$ in the DLKCF is non-interacting, resulting in a suboptimal filter. Nevertheless, it is shown in Proposition 1 that the GAS property of the error dynamics is not affected by neglecting the cross-correlation terms. The consistency of the DLKCF is validated through exploring the average normalized (state) estimation error squared (NEES) measure \cite{BarShalom2001} in Section \ref{sec: NumericalExperiments}.
\end{remark}
Before proving the properties of the estimator, the following assumptions are made for the DLKCF: (\textit{i}) the state dimension $n_i \ge 2$ for all $i$ since at least two boundary cells exist in each freeway section; (\textit{ii}) the noise models satisfy $q_1 I <Q_{i,k}<q_2 I$ and $r_1 I<R_{i,k}<r_2 I$ for all $i$ and $k$, where $q_1$, $q_2$, $r_1$ and $r_2$ are positive constants; and (\textit{iii}) the scaling factor satisfies $\gamma^j_{i,k}\le \hat{\gamma}_{i,k}^{j}=\hat{c}\left|\mathcal{N}_i\right|^{-1}\|\Gamma_{i,k|k-1}\hat{I}_{i,j}^Tu_{i,k}^j\|^{-1}$ in addition to $\gamma^j_{i,k}<\gamma_{i,k}^{*}$ for all $i$, $j\in\mathcal{N}_i$ and $k$. Here $|\mathcal{N}_i|$ is the number of neighbors of agent $i$, and $\hat{c}>0$ is a constant predefined to set an upper bound for the magnitude of the consensus term. Also note that the upper bound $\hat{\gamma}_{i,k}^{j}$ can be computed locally and online by each agent. In this case, the 2-norm\footnote{For the remainder of this article, we denote as $\|\cdot\|$ the 2-norm of a matrix or a vector.} of the consensus term is upper bounded as follows:
\begin{align}\label{eq:bound_consensus}
\begin{array}{l}
\left\|\sum_{j\in \mathcal{N}_i}\gamma^j_{i,k}\Gamma_{i,k|k-1}\hat{I}_{i,j}^Tu_{i,k}^j\right\|\le \hat{c}, \quad\text{for all $i$ and $k$.}
\end{array}
\end{align}

In practice, to run the DLKCF each agent needs to use $\hat{A}_{i,k}$ (i.e., the estimated $A_{i,k}$ obtained based on the state estimate and sensor data) instead of $A_{i,k}$ in \eqref{eq:DLKCFforecast}. In observable modes, the matrix $A_{i,k}$ can be correctly reconstructed by the local agent. However, in the FC modes $\hat{A}_{i,k}$ and $A_{i,k}$ are unlikely to be the same since the agent also needs to estimate the location and direction of the shock. As a related note, using the constrained-CTM \cite{CarlosConstrainedCTM} can improve the estimation accuracy of $A_{i,k}$. Also note that all the theoretical performance analysis of the DLKCF regarding the unobservable scenarios in the next section hold even if $\hat{A}_{i,k}$ and $A_{i,k}$ differ.

\section{Stability and performance analysis of the DLKCF for traffic estimation}\label{sec:StabilityPerformance}


\subsection{Asymptotic stability of mean error in observable modes}\label{sub:ObservableConvergence}

Define the prior and posterior estimation error for section $i$ as $\eta_{i,k|k-1}=\rho_{i,k|k-1}-\rho_{i,k}$ and $\eta_{i,k|k}=\rho_{i,k|k}-\rho_{i,k}$, and define the neighbor disagreement on the shared estimates as:
\begin{align}
u_{i,k}^j=\hat{I}_{j,i}\eta_{j,k|k-1}-\hat{I}_{i,j}\eta_{i,k|k-1}.\label{eq:DLCKFu}
\end{align}
Note that this is a different notion of disagreement from Corollary 1 of \cite{Olfati-SaberCDC2009}, which measures the disagreement of an agent's estimate with respect to the mean estimate over all the agents.
The global estimation error $\eta_{1:N,k|k}$ is constructed by $\eta_{1:N,k|k}=\textrm{col}(\eta_{1,k|k},\cdots,\eta_{N,k|k})$. Let the bold font $\bm{x}$ denote the mean of random vector $x$ (i.e., $\bm{x}=\mathbb{E}[x]$). The mean of the estimation error in section $i$ evolves as follows:
\begin{align}
\bm{\eta}_{i,k|k}=F_{i,k}A_{i,k-1}\bm{\eta}_{i,k-1|k-1}+\sum_{j\in \mathcal{N}_i}C^{j}_{i,k}\bm{u}_{i,k}^j,\label{eq:DLKCFerrorDynamics}
\end{align}
where $F_{i,k}=I-K_{i,k}H_{i,k}$. We choose a \emph{common Lyapunov function} candidate which reads
\begin{align}
V_k=\sum_{i=1}^{N}\bm{\eta}_{i,k|k}^T\Gamma_{i,k|k}^{-1}\bm{\eta}_{i,k|k}, \label{eq:DLCKFKLyap}
\end{align}
and compute its one-step change $\Delta V_k=V_{k}-V_{k-1}$ by applying \eqref{eq:DLKCFerrorDynamics} as follows:
\begin{align}\label{eq:Lyap_one_step}
\arraycolsep=1.5pt\def\arraystretch{1.5}
\begin{array}{rl}
\Delta V_k=&\sum_{i=1}^{N}\bm{\eta}_{i,k-1|k-1}^T\left(A_{i,k-1}^TF_{i,k}^T\Gamma_{i,k|k}^{-1}F_{i,k}A_{i,k-1}-\Gamma_{i,k-1|k-1}^{-1}\right)\bm{\eta}_{i,k-1|k-1}\\
&+2\sum_{i=1}^{N}\left(\bm{\eta}_{i,k|k-1}^TF_{i,k}^T\Gamma_{i,k|k}^{-1}\sum_{j\in \mathcal{N}_i}C^{j}_{i,k}\bm{u}_{i,k}^j\right)\\
&+\sum_{i=1}^{N}\left(\sum_{j\in \mathcal{N}_i}C^{j}_{i,k}\bm{u}_{i,k}^j\right)^T\Gamma_{i,k|k}^{-1}\left(\sum_{j\in \mathcal{N}_i}C^{j}_{i,k}\bm{u}_{i,k}^j\right).
\end{array}
\end{align}
\subsubsection{Radial unboundedness of the common Lyapunov function}
In order to ensure that the common Lyapunov function \eqref{eq:DLCKFKLyap} is \emph{radially unbounded}, we need to show that $\Gamma_{i,k|k}^{-1}$ is upper and lower bounded for all $i$ and $k$ when all freeway sections switch among the observable modes of the SMM. The derivation of the bounds is divided into the following three parts. For notational simplicity, we drop the section index $i$ in the first two parts.
\begin{enumerate}
 \item In Lemma \ref{lem:GammaBound}, the upper and lower bounds of $\Gamma_{k|k}^{-1}$ for a freeway section are derived for $k\ge \max\left\{1,n-2\right\}$, which are independent of the initial condition $\Gamma_{0|0}$.
 \item In Lemma \ref{lem:GammaBound1}, the upper and lower bounds of $\Gamma_{k|k}^{-1}$ for the freeway section are derived for $0\le k < \max\left\{1,n-2\right\}$, which are functions of the initial condition $\Gamma_{0|0}$.
 \item The results in Lemma \ref{lem:GammaBound} and Lemma \ref{lem:GammaBound1} are combined together in Lemma \ref{prop:GammaBound} to express the upper and lower bounds for $\Gamma_{i,k|k}^{-1}$ for each freeway section (indexed by $i$), and the derived bounds are uniform across all time steps $k\ge 0$.
\end{enumerate}
The next lemma shows that the upper and lower bounds of $\Gamma_{k|k}^{-1}$ for a freeway section are independent of the initial condition $\Gamma_{0|0}$ when $k\ge \max\left\{1,n-2\right\}$.

\begin{lem}\label{lem:GammaBound}
Consider a freeway section of dimension $n\ge 2$ that switches among observable modes. Let $H_\text{b}$ be the output matrix associated with the boundary measurements:
\begin{equation}\label{eq:h_b}
\begin{split}
H_{\text{b}}=&\left(
  \begin{array}{ccccc}
    1 & 0 & \cdots & 0 & 0\\
    0 & 0 & \cdots & 0 & 1
  \end{array}
\right)\in \mathbb{R}^{2\times n}.
\end{split}
\end{equation}
Define $T_1=\max\left\{1,n-2\right\}$, and\footnote{When $T_1=1$, the definitions of $a_{\mathcal{C}}$ and $b_{\mathcal{C}}$ are given by $a_{\mathcal{C}}=q_1$ and $b_{\mathcal{C}}=q_2$, respectively.}
\begin{equation}\label{eq:a_b_I}
\begin{split}
&a_{\mathcal{I}}=r_2^{-1}\min_{M_{\kappa} \in \mathcal{A}_{\text{O}}}\left\{\lambda_{\min}\left(\left(H_{\text{b}}^TH_{\text{b}}+\sum_{\iota=1}^{T_1}\left(\prod_{\kappa=\iota}^{T_1}M_{\kappa}^{-1}\right)^TH_{\text{b}}^TH_{\text{b}}\left(\prod_{\kappa=\iota}^{T_1}M_{\kappa}^{-1}\right)\right)\right)\right\},\\
&b_{\mathcal{I}}=r_1^{-1}\max_{M_{\kappa} \in \mathcal{A}_{\text{O}}}\left\{\lambda_{\max}\left(\left(I+\sum_{\iota=1}^{T_1}\left(\prod_{\kappa=\iota}^{T_1}M_{\kappa}^{-1}\right)^T\left(\prod_{\kappa=\iota}^{T_1}M_{\kappa}^{-1}\right)\right)\right)\right\},
\end{split}
\end{equation}
\begin{equation}\label{eq:a_b_C}
\begin{split}
&a_{\mathcal{C}}=q_1\min_{M_{\kappa} \in \mathcal{A}_{\text{O}}}\left\{\lambda_{\min}\left(I+\sum_{\iota=1}^{T_1-1}\left(\prod_{\kappa=\iota}^{T_1-1}M_{\kappa}\right)\left(\prod_{\kappa=\iota}^{T_1-1}M_{\kappa}\right)^T\right)\right\},\\
&b_{\mathcal{C}}=q_2\max_{M_{\kappa} \in \mathcal{A}_{\text{O}}}\left\{\lambda_{\max}\left(I+\sum_{\iota=1}^{T_1-1}\left(\prod_{\kappa=\iota}^{T_1-1}M_{\kappa}\right)\left(\prod_{\kappa=\iota}^{T_1-1}M_{\kappa}\right)^T\right)\right\}.
\end{split}
\end{equation}
If $\Gamma_{0|0}>\bm{0}$, the inverse of the error covariance computed by the DLKCF  \eqref{eq:DLKCFforecast}-\eqref{eq:DLKCFAnalysis} satisfies
\begin{align*}
\bm{0}<c_1 I < \Gamma_{k|k}^{-1} < c_2 I, \quad \text{for all $k \ge \max\left\{1,n-2\right\}$,}
\end{align*}
where $\bm{0}$ is the matrix of appropriate dimensions which is zero everywhere, and
\begin{align}\label{eq:c_1_c_2}
c_1=\frac{a}{1+ab}>0,\quad c_2=\frac{1+ab}{a}> c_1>0,
\end{align}
with $a=\min\{a_{\mathcal{I}}, a_{\mathcal{C}}\}>0$ and $b=\max\{b_{\mathcal{I}}, b_{\mathcal{C}}\}> a >0$.

\end{lem}
\begin{proof}
The proof is reported in Appendix \ref{ap:proof_lemma_GammaBound}.
\end{proof}
In fact, the values of $a_{\mathcal{I}}$, $b_{\mathcal{I}}$, $a_{\mathcal{C}}$, $b_{\mathcal{C}}$ in \eqref{eq:a_b_I} and \eqref{eq:a_b_C} can be determined (i.e., numerically) since all of the matrices in $\mathcal{A}_{\text{O}}$ are known. Hence, the upper and lower bounds for $c_1$ and $c_2$ in \eqref{eq:c_1_c_2} can be computed offline before the filter is implemented. The next corollary derives a lower bound for $c_1$ and an upper bound for $c_2$ that can be calculated analytically.
\begin{cor}\label{cor:bound_gamma_analytical}
Define $\underline{\theta}=\min\{v_{\text{m}}\frac{\Delta t}{\Delta x}, 1-v_{\text{m}}\frac{\Delta t}{\Delta x}, w\frac{\Delta t}{\Delta x}, 1-w\frac{\Delta t}{\Delta x}\}$, $\tilde{\theta}=\min\{1-v_{\text{m}}\frac{\Delta t}{\Delta x}, 1-w\frac{\Delta t}{\Delta x}\}$, and
\begin{align*}
&\underline{a}=\min\left\{q_1\left(1+\frac{\tilde{\theta}^{2n}\left(1-\tilde{\theta}^{2n\left(T_1-1\right)}\right)}{\left(1-\tilde{\theta}^{2n}\right)\sqrt{4n\left(2T_1-1\right)^22^{n-2}}}\right),\quad \frac{r_2^{-1}\underline{\theta}^{T_1\left(T_1+1\right)}}{2\left(2T_1+1\right)\sqrt{4n\left(T_1+1\right)^22^{n-2}}}\right\},\\
&\bar{b}=\max\left\{q_2\left(2T_1^2-1\right),\quad r_1^{-1}\left(1+ \frac{T_1\left(T_1+2\right)\sqrt{4n2^{n-2}}}{\tilde{\theta}^{2T_1n}}\right)\right\},
\end{align*}
where $T_1=\max\left\{1,n-2\right\}$. The upper and lower bounds $c_1$ and $c_2$ in \eqref{eq:c_1_c_2} satisfy
\begin{align*}
c_1>\frac{\underline{a}}{1+\underline{a}\bar{b}},\quad c_2<\frac{1+\underline{a}\bar{b}}{\underline{a}}.
\end{align*}
\end{cor}
\begin{proof}
The proof is reported in Appendix \ref{ap:cor_bound}.
\end{proof}
The next lemma derives the upper and lower bounds of $\Gamma_{k|k}^{-1}$ for a freeway section when $0\le k < \max\left\{1,n-2\right\}$, which are functions of the initial condition $\Gamma_{0|0}$.

\begin{lem}\label{lem:GammaBound1}
Consider a freeway section of dimension $n\ge 2$ that switches among observable modes. If $\Gamma_{0|0}>\bm{0}$, the inverse of the error covariance computed by the DLKCF \eqref{eq:DLKCFforecast}-\eqref{eq:DLKCFAnalysis} satisfies
\begin{equation*}
\bm{0}<\tilde{\mathfrak{c}}_1\left(\Gamma_{0|0}\right)I\le \Gamma^{-1}_{k|k} \le \tilde{\mathfrak{c}}_2\left(\Gamma_{0|0}\right)I,\quad \text{for all $0\le k < \max\left\{1,n-2\right\}$,}
\end{equation*}
where $\tilde{\mathfrak{c}}_1\left(\cdot\right)$ and $\tilde{\mathfrak{c}}_2\left(\cdot\right)$ are functions of $M \in\mathbb{R}^{n\times n}$, and are defined as follows:
\begin{equation}\label{eq:mathfrak_c1}
\tilde{\mathfrak{c}}_1\left(M \right)= \left\{ \begin{array}{ll}
\left(2\left(2n-5\right)\left\|M\right\|+\left(2\left(n-3\right)^2-1\right)q_2\right)^{-1}& \textrm{if $n\ge 4$}\\
\left\|M \right\|^{-1} & \textrm{if $2\le n<4$,}
\end{array} \right.
\end{equation}
and
\begin{equation}\label{eq:mathfrak_c2}
\tilde{\mathfrak{c}}_2\left(M \right)= \max\left\{\lambda^{-1}_{\min}\left(M \right),q_1^{-1}+r_1^{-1}\right\}.
\end{equation}
\end{lem}
\begin{proof}
The proof is reported in Appendix \ref{ap:GammaBound1_proof}.
\end{proof}
Combining Lemma \ref{lem:GammaBound} and \ref{lem:GammaBound1}, the upper and lower bounds for the inverse of the error covariance is obtained when a freeway section switches among the observable modes of the SMM, as stated next in Lemma \ref{prop:GammaBound}. We add back the section index $i$ in the statement of Lemma \ref{prop:GammaBound} to emphasize that the result holds for each individual freeway section.

\begin{lem}[Lemma 1 in \cite{SunWorkCONES2016}]\label{prop:GammaBound}
Consider a freeway section (indexed by $i$) that switches among observable modes for all $k\ge0$. If $\Gamma_{i,0|0}>\bm{0}$, then $\Gamma^{-1}_{i,k|k}$ given in the DLKCF \eqref{eq:DLKCFforecast}-\eqref{eq:DLKCFAnalysis} satisfies
\begin{equation}\label{eq:def_frak_c}
\bm{0}<\mathfrak{c}_1\left(\Gamma_{i, 0|0}\right)I\le \Gamma^{-1}_{i, k|k} \le \mathfrak{c}_2\left(\Gamma_{i, 0|0}\right)I,\quad \text{for $k\ge 0$,}
\end{equation}
independent of the switching sequence, where $\mathfrak{c}_1\left(\cdot\right)$ and $\mathfrak{c}_2\left(\cdot\right)$ are functions of $M \in\mathbb{R}^{n\times n}$ defined as follows:
\begin{equation*}
\begin{split}
\mathfrak{c}_1\left(M \right)&=\min\left\{\tilde{\mathfrak{c}}_1\left(M\right), c_1\right\}\\
\mathfrak{c}_2\left(M \right)&=\max\left\{\tilde{\mathfrak{c}}_2\left(M\right),c_2\right\},
\end{split}
\end{equation*}
with $\tilde{\mathfrak{c}}_1\left(\cdot \right)$, $\tilde{\mathfrak{c}}_2\left(\cdot \right)$ defined in \eqref{eq:mathfrak_c1}-\eqref{eq:mathfrak_c2} (cf. Lemma \ref{lem:GammaBound1}), and $c_1$, $c_2$ defined in \eqref{eq:c_1_c_2} (cf. Lemma \ref{lem:GammaBound}).
\end{lem}

\subsubsection{Global asymptotic stability of the mean error dynamics in observable modes}
When all sections switch among the observable modes, $V_k$ is radially unbounded since \eqref{eq:def_frak_c} holds for all $i$. Now we are ready to show the GAS of the mean error dynamics when all freeway sections switch among the observable modes.

\begin{prop}[Proposition 1 in \cite{SunWorkCONES2016}]\label{Prop:ObservableGUAS}
Consider the DLKCF in \eqref{eq:DLKCFforecast} and \eqref{eq:DLKCFAnalysis} with the consensus gain in \eqref{eq:DLKCFconsensusgain}. Suppose all sections switch among the observable modes of the SMM. Then, the mean estimation error $\bm{\eta}_{1:N,k|k}=\mathbb{E}[\eta_{1:N,k|k}]$ is GAS for sufficiently small $\gamma^j_{i,k}$.
\end{prop}
\begin{proof}
We show $\Delta V_{k}$ in \eqref{eq:Lyap_one_step} is negative definite when $\bm{\eta}_{1:N, k-1|k-1}\neq 0$.

\noindent\textbf{Step 1.}  Negative definiteness of the first term in \eqref{eq:Lyap_one_step}.

The proof for the first term follows closely from \cite{Olfati-SaberCDC2009} with minor changes. Here we only show the result and introduce the matrices needed in this article. Note that $A_{i,k}$ is invertible for all $i$ and $k$ in the SMM. Each element in the first term in \eqref{eq:Lyap_one_step} can be equivalently written as:
\begin{equation*}
\arraycolsep=1.5pt\def\arraystretch{1.5}
\begin{array}{rl}
&\bm{\eta}_{i,k-1|k-1}^T\left(A_{i,k-1}^TF_{i,k}^T\Gamma_{i,k|k}^{-1}F_{i,k}A_{i,k-1}-\Gamma_{i,k-1|k-1}^{-1}\right)\bm{\eta}_{i,k-1|k-1}\notag\\
=&-\bm{\eta}_{i,k|k-1}^T\left(\left(A_{i,k-1}\Gamma_{i,k-1|k-1}A_{i,k-1}^T\right)^{-1}-F_{i,k}^T\Gamma_{i,k|k}^{-1}F_{i,k}\right)\bm{\eta}_{i,k|k-1}\\
=&-\bm{\eta}_{i,k|k-1}^T\left(\left(A_{i,k-1}\Gamma_{i,k-1|k-1}A_{i,k-1}^T\right)^{-1}-\left(A_{i,k-1}\Gamma_{i,k-1|k-1}A_{i,k-1}^T+W_{i,k-1}\right)^{-1}\right)\bm{\eta}_{i,k|k-1},\\
=&-\bm{\eta}_{i,k|k-1}^T\Lambda_{i,k-1}\bm{\eta}_{i,k|k-1},
\end{array}
\end{equation*}
where the second equation is due to Lemma 2 in \cite{Olfati-SaberCDC2009}, with $\Lambda_{i,k}$ defined as
\begin{align}
\Lambda_{i,k}=&\left(A_{i,k}\Gamma_{i,k|k}A_{i,k}^T\right)^{-1}-\left(A_{i,k}\Gamma_{i,k|k}A_{i,k}^T+W_{i,k}\right)^{-1},\notag
\end{align}
where
\begin{align}
W_{i,k}=&Q_{i,k}+\Gamma_{i,k+1|k}S_{i,k+1}\Gamma_{i,k+1|k}>0\notag,
\end{align}
and $S_{i,k}=H_{i,k}^TR_{i,k}^{-1}H_{i,k}$. Due to the matrix inversion lemma,
\begin{align}
&\Gamma_{i,k|k}A_{i,k}^T\Lambda_{i,k}A_{i,k}\Gamma_{i,k|k}\notag\\
=&\Gamma_{i,k|k}-\Gamma_{i,k|k}A_{i,k}^T\left(A_{i,k}\Gamma_{i,k|k}A_{i,k}^T+W_{i,k}\right)^{-1}A_{i,k}\Gamma_{i,k|k},\notag\\
=&\left(\Gamma_{i,k|k}^{-1}+A_{i,k}^TW_{i,k}^{-1}A_{i,k}\right)^{-1}>0,\notag
\end{align}
hence $\Lambda_{i,k}>0$. Consequently, the first term in \eqref{eq:Lyap_one_step} is negative definite.

\noindent\textbf{Step 2.} Negative semidefiniteness of the second term in \eqref{eq:Lyap_one_step}.

Due to Lemma 2(i) in \cite{Olfati-SaberCDC2009} we have $F_{i,k}=\Gamma_{i,k|k}\Gamma_{i,k|k-1}^{-1}$, hence the consensus gain is equivalent to
\begin{align}
C^{j}_{i,k}=\gamma^j_{i,k}\Gamma_{i,k|k-1}\hat{I}_{i,j}^T=\gamma^j_{i,k}\Gamma_{i,k|k}\left(F_{i,k}^T\right)^{-1}\hat{I}_{i,j}^T\notag.
\end{align}
Let $\hat{i}\in\{1,\cdots,N-1\}$ be the index of the overlapping regions, and define
\begin{align*}
\bm{\hat{\eta}}_{\hat{i},k|k-1}=(\bm{\eta}^T_{\hat{i},k|k-1}\hat{I}_{\hat{i},\hat{i}+1}^T, \bm{\eta}^T_{\hat{i}+1,k|k-1}\hat{I}_{\hat{i}+1,\hat{i}}^T)^T.   
\end{align*}
The second term in \eqref{eq:Lyap_one_step} can be written as
\begin{align}
\arraycolsep=1.5pt\def\arraystretch{1.5}
\begin{array}{rl}
2\sum_{i=1}^{N}\left(\bm{\eta}_{i,k|k-1}^TF_{i,k}^T\Gamma_{i,k|k}^{-1}\sum_{j\in \mathcal{N}_i}C^{j}_{i,k}\bm{u}_{i,k}^j\right)&=2\sum_{\hat{i}=1}^{N-1}\gamma_{\hat{i},k}^{\hat{i}+1}\left(\bm{\eta}_{\hat{i},k|k-1}^T\hat{I}_{\hat{i},\hat{i}+1}^T\bm{u}_{\hat{i},k}^{\hat{i}+1}+\bm{\eta}_{\hat{i}+1,k|k-1}^T\hat{I}_{\hat{i}+1,\hat{i}}^T\bm{u}_{\hat{i}+1,k}^{\hat{i}}\right)\\
&=-2\sum_{\hat{i}=1}^{N-1}\gamma_{\hat{i},k}^{\hat{i}+1}\bm{\hat{\eta}}^T_{\hat{i},k|k-1}\hat{L}_{\hat{i}}\bm{\hat{\eta}}_{\hat{i},k|k-1}\leq0,
\end{array}\notag
\end{align}
where
\begin{align}
\hat{L}_{\hat{i}}=\left(
  \begin{array}{cc}
    1&-1\\
    -1&1
  \end{array}
\right)\otimes I_{n_{\hat{i},\hat{i}+1}},\notag
\end{align}
and the last inequality holds due to the quadratic property of the Laplacian matrix \cite{Godsil2001Graph}.

\noindent \textbf{Step 3.} Upper bound of the third term in \eqref{eq:Lyap_one_step}.

Given the choice of consensus gain in \eqref{eq:DLKCFconsensusgain}, the third term in \eqref{eq:Lyap_one_step} can be written as
\begin{align}
\begin{array}{l}
\quad \sum_{i=1}^{N}\left(\sum_{j\in \mathcal{N}_i}C^{j}_{i,k}\bm{u}_{i,k}^j\right)^T\Gamma_{i,k|k}^{-1}\left(\sum_{j\in \mathcal{N}_i}C^{j}_{i,k}\bm{u}_{i,k}^j\right)\\
=\sum_{i=1}^{N}\left(\sum_{j\in \mathcal{N}_i}\hat{I}_{i,j}^T\gamma^j_{i,k}\bm{u}_{i,k}^j\right)^TG_{i,k}\left(\sum_{j\in \mathcal{N}_i}\hat{I}_{i,j}^T\gamma^j_{i,k}\bm{u}_{i,k}^j\right),
\end{array}\notag
\end{align}
where we define $G_{i,k}=A_{i,k-1}\Gamma_{i,k-1|k-1}A_{i,k-1}^{T}+Q_{i,k-1}+\Gamma_{i,k|k-1}S_{i,k}\Gamma_{i,k|k-1}$.
Recall that $\mathcal{J}_i=\mathcal{N}_i\bigcup\{i\}$, and define $\bm{\eta}_{\mathcal{J}_i,k|k-1}=\textrm{col}_{j\in\mathcal{J}_i}\left(\bm{\eta}_{j,k|k-1}\right)$ where $j$ are sorted in ascending order. Columnizing $\bm{u}_{i,k}^j$ over all neighbors $j\in \mathcal{N}_i$ within section $i$ yields
\begin{align}
\bm{u}_{\mathcal{N}_i,k}=\textrm{col}_{j\in \mathcal{N}_i}\left(\gamma_{i,k}^j\bm{u}_{i,k}^j\right)=\tilde{L}_i\tilde{I}_i\bm{\eta}_{\mathcal{J}_i,k|k-1},
\end{align}
where $j$ are sorted in ascending order, $\tilde{L}_i$ is defined as
\begin{align}
\tilde{L}_{i}= \left\{ \begin{array}{ll}
\left(
  \begin{array}{cc}
   -\hat{I}_{i,i+1}&\hat{I}_{i+1,i}
  \end{array}
\right)& \textrm{if $i=1$}\\
\left(
  \begin{array}{cc}
   \hat{I}_{i-1,i}&-\hat{I}_{i,i-1}
  \end{array}
\right)& \textrm{if $i=n$,}\\
\left(
  \begin{array}{ccc}
   \hat{I}_{i-1,i}&-\hat{I}_{i,i-1}&\bm{0}_{n_{i+1,i+1}}\\
   \bm{0}_{n_{i-1,i-1}}& -\hat{I}_{i,i+1}&\hat{I}_{i+1,i}
  \end{array}
\right)& \textrm{otherwise,}\\
\end{array} \right.\notag
\end{align}
and $\tilde{I}_{i}=\textrm{diag}(\gamma^{i-1}_{i,k}I_{n_{i-1}+\lfloor0.5n_{i}\rfloor},\gamma^{i+1}_{i,k}I_{n_{i}-\lfloor0.5n_{i}\rfloor+n_{i+1}})$. Further define
\begin{align}
\tilde{H}_i= \left\{ \begin{array}{ll}
\hat{I}_{i,i+1}& \textrm{if $i=1$}\\
\hat{I}_{i,i-1}& \textrm{if $i=n$}\\
\left(\hat{I}_{i,i-1}^T\textrm{ }\hat{I}_{i,i+1}^T\right) & \textrm{otherwise.}
\end{array} \right.\notag
\end{align}
The third term in \eqref{eq:Lyap_one_step} is equivalent to
\begin{align}
\begin{array}{l}
\quad \sum_{i=1}^{N}\left(\sum_{j\in \mathcal{N}_i}\hat{I}_{i,j}^T\gamma^j_{i,k}\bm{u}_{i,k}^j\right)^TG_{i,k}\left(\sum_{j\in \mathcal{N}_i}\hat{I}_{i,j}^T\gamma^j_{i,k}\bm{u}_{i,k}^j\right)\\
=\sum_{i=1}^{N}\bm{\eta}_{\mathcal{J}_i,k|k-1}^T\tilde{I}_i\tilde{L}_i^T\tilde{H}_i^TG_{i,k}\tilde{H}_i\tilde{L}_i\tilde{I}_i\bm{\eta}_{\mathcal{J}_i,k|k-1}\\
\leq\sum_{i=1}^{N}\left(\gamma_{i,k}^{\max}\right)^2\lambda_{\max}\left(\tilde{L}_i^T\tilde{H}_i^TG_{i,k}\tilde{H}_i\tilde{L}_i\right)\|\bm{\eta}_{\mathcal{J}_i,k|k-1}\|^2,
\end{array}\notag
\end{align}
where $\gamma_{i,k}^{\max}=\max_{j\in\mathcal{N}_i}\gamma_{i,k}^j$ and $\lambda_{\max}$ (resp. $\lambda_{\min}$) is the maximum (resp. minimum) eigenvalue of a matrix.

\noindent\textbf{Step 4.} The negative definiteness of \eqref{eq:Lyap_one_step}.

Note that given Step 1, the first term of \eqref{eq:Lyap_one_step} can be equivalently written as
\begin{align*}
\begin{array}{l}
\quad\sum_{i=1}^{N}-\bm{\eta}_{i,k|k-1}^T\Lambda_{i,k-1}\bm{\eta}_{i,k|k-1}=\sum_{i=1}^{N}-\bm{\eta}_{\mathcal{J}_i,k|k-1}^T\Lambda_{\mathcal{J}_i,k-1}\bm{\eta}_{\mathcal{J}_i,k|k-1},\notag
\end{array}
\end{align*}
where $\Lambda_{\mathcal{J}_i,k}=\textrm{diag}_{j\in\mathcal{J}_i}(\mu^j_{i}\Lambda_{j,k})$ with the indexes $j$ sorted by ascending order, and the scaling factors are pre-defined and satisfy $\sum_{j\in\mathcal{J}_i}\mu_j^i=1$ for all $i$. Given Steps 1-3, $\Delta V_k$ satisfies
\begin{align}
\begin{array}{l}
\Delta V_k\leq-2\sum_{\hat{i}=1}^{N-1}\gamma_{\hat{i},k}^{\hat{i}+1}\bm{\hat{\eta}}^T_{\hat{i},k|k-1}\hat{L}_{\hat{i}}\bm{\hat{\eta}}_{\hat{i},k|k-1}\\
\quad\quad\quad+\sum_{i=1}^{N}\left(\left(\gamma_{i,k}^{\max}\right)^2\lambda_{\max}\left(\tilde{L}_i^T\tilde{H}_i^TG_{i,k}\tilde{H}_i\tilde{L}_i\right)-\lambda_{\min}\left(\Lambda_{\mathcal{J}_i,k-1}\right)\right)\|\bm{\eta}_{\mathcal{J}_i,k|k-1}\|^2.
\end{array}\label{eq:oneStepChange}
\end{align}
Therefore by choosing $\gamma^j_{i,k}$ sufficiently small we can render $\Delta V_k<0$ for all $k\geq0$ and for all $\bm{\eta}_{1:N,k-1|k-1}\neq0$. Precisely, we need $\gamma^j_{i,k}<\gamma_{i,k}^*$ where $\gamma_{i,k}^*$ is defined by
\begin{align*}
\begin{array}{l}
\gamma_{i,k}^*=\left(\frac{\lambda_{\min}\left(\Lambda_{\mathcal{J}_i,k-1}\right)}{\lambda_{\max}\left(\tilde{L}_i^T\tilde{H}_i^TG_{i,k}\tilde{H}_i\tilde{L}_i\right)}\right)^{\frac{1}{2}}.
\end{array}
\end{align*}
Note that to compute $\gamma_{i,k}^*$, only information from one-hop neighbors is needed, and global communication topology is not required compared to \cite{Olfati-SaberCDC2009}.
Hence, $\Delta V_k<0$ for all $k\geq0$ and $\bm{\eta}_{1:N,k-1|k-1}\neq0$, and therefore $\bm{\eta}_{1:N,k|k}=0$ is GAS for the mean error dynamics of the DLKCF. Consequently, all estimators reach consensus on the shared states.
\end{proof}
When the consensus gain is zero, the mean error dynamics of each local agent is also GAS under observable modes. However, due to different model errors and innovation sequences, the estimates provided by neighboring agents on their shared overlapping regions inevitably disagree in any realization of the filter. Hence, the consensus term is designed to promote agreement without destabilizing the filter, which is further verified in Section \ref{sub:Individual_DLKCF}. Moreover, when $\gamma^j_{i,k}<\gamma_{i,k}^*$, it can be deduced from \eqref{eq:oneStepChange} that $\Delta V_k<-\sqrt{2}\sum_{\hat{i}=1}^{N-1}\gamma_{\hat{i},k}^{\hat{i}+1}\|\bm{u}^{\hat{i}+1}_{\hat{i},k}\|^2$ (derived in Appendix \ref{ap:deltaV}). This indicates that $V_k$ strictly decreases at the rate proportional to the total disagreement until the neighboring disagreements on all the overlapping regions converge to zero, which is a property cannot be achieved without the consensus term.

\subsection{Ultimately bounded mean estimates in unobservable modes}\label{sub:UnobservableBound}

Challenges for estimating an unobservable section stem from the dependence of the system dynamics of the SMM on the shock velocity and location, which are functions of the state variables to be estimated. Hence, non-observability of the system will lead to unknown system dynamics. Moreover, the unobservable modes are also undetectable since the density of the cells in the unobservable subsystem does not dissipate. In this subsection we show that the mean estimates of all the cells in an unobservable section are ultimately bounded inside $[-\epsilon,\varrho_{\text{m}}+\epsilon]$ for all $\epsilon>0$, provided that the upstream and downstream measurements are available. This ensures that the mean estimates of the DLKCF for unobservable modes are always physically meaningful to within $\epsilon$. Since this subsection studies the properties of the filter for an individual unobservable freeway section, the section index $i$ is dropped for notational simplicity.

First we present a lemma stating the boundedness of the Kalman gain, which is necessary for the boundedness of the state estimate, and is obtained based on the boundedness of the cross-covariance of the observable and unobservable subsystems in the \emph{Kalman observability canonical form}. 

\begin{lem}[Lemma 2 in \cite{SunWorkCONES2016}]\label{lem_u_bound_k}
Consider a freeway section with dimension $n\ge 2$. Let $(\underline{k}_{\text{U}},\bar{k}_{\text{U}}]$ be the time interval\footnote{Throughout this article, the time instant $k\in \mathbb{N}$. Hence $k\in(\underline{k}_{\text{U}},\bar{k}_{\text{U}}]$ means $k\in\{\underline{k}_{\text{U}}+1,\cdots, \bar{k}_{\text{U}}\}$.} while the section stays inside the unobservable modes, i.e., $\sigma(k)\in\{\text{FC1, FC2}\}$ for $k\in(\underline{k}_{\text{U}}, \overline{k}_{\text{U}}]$, and $\sigma(k)\in\{\text{FF, CC, CF}\}$ for $k=\underline{k}_{\text{U}}$ and $k=\bar{k}_{\text{U}}+1$, where $0\le \underline{k}_{\text{U}}<\bar{k}_{\text{U}}\le+\infty$. Define
\begin{align}\label{eq:check_a_b_c}
\check{a}=\min\left\{2r_2^{-1},q_1\right\},\quad \check{b}=\max\left\{2r_1^{-1},q_2\right\},\quad \check{c}_1=\frac{\check{a}}{1+\check{a}\check{b}},\quad \check{c}_2=\frac{1+\check{a}\check{b}}{\check{a}},
\end{align}
and let
\begin{equation}\notag
\begin{split}
\check{c}_3&=\check{c}_1^{-1}+q_1^{-1}\check{c}_1^{-2},\quad\bar{t}=\sqrt{2}\left(n-2\right) \left(\frac{\check{c}_2}{\check{c}_1}\right)^{\frac{1}{2}},\\
\bar{q}&=\left(1-\check{c}_3\check{c}_2^{-1}\right)^{\frac{1}{2}},\quad
\bar{p}=2r_2q_1^{-1}\frac{\Delta t}{\Delta x}\max\{v_{\text{m}},w\}\left(r_2+q_2\right)+q_2,\\
\bar{\gamma}&=n\sqrt{n}\left\|\Gamma_{\underline{k}_{\text{U}}|\underline{k}_{\text{U}}}\right\|\left(1+\frac{\Delta t}{\Delta x}\left(w+v_{\text{m}}\right)\right)^2+n\sqrt{n}q_2\left(1+\frac{\Delta t}{\Delta x}\left(w+v_{\text{m}}\right)\right)+\sqrt{n}q_2.
\end{split}
\end{equation}
Given density measurements of the boundary cells, the Kalman gain satisfies $\left\|K_k\right\|_{\infty}\le \mathfrak{k}\left(\Gamma_{\underline{k}_{\text{U}}|\underline{k}_{\text{U}}}\right)$ for all $k\in (\underline{k}_{\text{U}}, \bar{k}_{\text{U}}]$, where $\mathfrak{k}\left(\cdot\right)$ is a function of $M\in\mathbb{R}^{n\times n}$ given by
\begin{equation}\label{eq:mathfrak_k}
\begin{split}
\mathfrak{k}\left(M\right)=\sqrt{2}r_1^{-1}\max\left\{\sqrt{n}\left(\left\|M\right\|\left(1+\frac{\Delta t \left(w+v_{\text{m}}\right)}{\Delta x}\right)+q_2\right), \sqrt{2}\left(r_2+q_2\right),  2\bar{\gamma}, 2\left(\bar{t}\bar{q}\bar{\gamma}+\bar{p}\right), 2\left(\frac{\bar{t}\bar{p}\bar{q}}{1-\bar{q}}+\bar{p}\right)\right\}.
\end{split}
\end{equation}
\end{lem}
\begin{proof}
The proof is reported in Appendix \ref{ap:lem_u_bound_k}.
\end{proof}

\begin{prop}[Proposition 2 in \cite{SunWorkCONES2016}]\label{Prop:UltimateBoundedness}
Consider an unobservable section in a road network with dimension $n$. For all $\epsilon>0$, a finite time $T(\epsilon)$ exists such that $\bm{\rho}^l_{k|k}\in[-\epsilon,\varrho_{{\text{m}}}+\epsilon]$ for all $k>T(\epsilon)$ and for all $l\in\{1,\cdots,n\}$, independent of the initial estimate.
\end{prop}
\begin{proof}
The proof is reported in \cite[Proposition 2]{SunWorkCDC14} (also given in Appendix \ref{ap:prop_ub_proof}).
\end{proof}

Proposition \ref{Prop:UltimateBoundedness} indicates that when the estimation error of the boundary cells converges to zero, it will drive the state estimate of the interior cells inside $[0, \varrho_{\text{m}}]$ due to the conservation law and the flow-density relationship embedded in the traffic model. Hence, it is necessary to ensure the error dynamics of the boundary cells is asymptotically stable.

\subsection{Boundedness of the mean error under switches among observable and unobservable modes}\label{section:switching}
This subsection derives the upper bound for the 2-norm of the mean estimation error when a freeway section switches among observable and unobservable modes. We first analyse the upper bound of the mean error when the section switches among the unobservable modes, which quantifies the increase of the mean error while the section is unobservable. Next, the convergence rate of the mean error dynamics while the section switches among the observable modes is studied. Finally, we derive the minimum number of time steps (i.e., the residence time) required in observable modes to ensure the boundedness of the mean error. All results in this subsection hold for every individual freeway section. In the analysis below, we drop the section index $i$ when it can be omitted for notational simplicity.

\subsubsection{Upper bound of the mean error in unobservable modes}
Let $(\underline{k}_{\text{U}}, \bar{k}_{\text{U}}]$ be the time interval inside which a section switches among unobservable modes, i.e., the mode index $\sigma(k)\in\{\text{FC1, FC2}\}$ for $k\in (\underline{k}_{\text{U}}, \bar{k}_{\text{U}}]$, and $\sigma(k)\in\{\text{FF, CC, CF}\}$ for $k=\underline{k}_{\text{U}}$ and $k=\bar{k}_{\text{U}}+1$. Based on Lemma \ref{lem_u_bound_k}, the next proposition derives an upper bound for $\|\bm{\eta}_{k|k}\|$ which is uniform across all $k\in(\underline{k}_{\text{U}}, \bar{k}_{\text{U}}]$. The derived bound is a function of $\epsilon$ and $\Gamma_{\underline{k}_{\text{U}}|\underline{k}_{\text{U}}}$, and is larger than $\epsilon$ (where $\epsilon$ is defined as the upper bound for $\|\bm{\eta}_{\underline{k}_{\text{U}}|\underline{k}_{\text{U}}}\|$). Moreover, the derived bound does not depend on the length of the time interval $(\underline{k}_{\text{U}}, \bar{k}_{\text{U}}]$.
\begin{prop}[Proposition 3 in \cite{SunWorkCONES2016}]\label{prop_u_bound}
Consider a freeway section which switches among the unobservable modes while $k\in (\underline{k}_{\text{U}},\bar{k}_{\text{U}}]$, where $0\le \underline{k}_{\text{U}}<\bar{k}_{\text{U}}\le+\infty$. Let
\begin{align}\label{eq:gamma_c}
\arraycolsep=1.5pt\def\arraystretch{1.5}
\begin{array}{l}
c_0=\max\left\{1, \sqrt{\check{c}_2\check{c}_1^{-1}}r_2q_1^{-1}\right\},\\
\mathfrak{c}\left(\Gamma_{\underline{k}_{\text{U}}|\underline{k}_{\text{U}}}\right)=c_0\Delta x\mathfrak{k}\left(\Gamma_{\underline{k}_{\text{U}}|\underline{k}_{\text{U}}}\right)\left(\Delta t \min\left\{v_{{\text{m}}}, w\right\}\right)^{-1},
\end{array}
\end{align}
where $\check{c}_1=\check{a}(1+\check{a}\check{b})^{-1}$ and $\check{c}_2=\check{a}^{-1}(1+\check{a}\check{b})$, with $\check{a}=\min\left\{2r_2^{-1},q_1\right\}$ and $\check{b}=\max\left\{2r_1^{-1},q_2\right\}$, and $\mathfrak{k}\left(\cdot\right)$ is given in \eqref{eq:mathfrak_k}. For all $\epsilon>0$, if $\left\|\bm{\eta}_{\underline{k}_{\text{U}}|\underline{k}_{\text{U}}}\right\|<\epsilon$, then $\|\bm{\eta}_{k|k}\|<\mathfrak{h}(\epsilon,\Gamma_{\underline{k}_{\text{U}}|\underline{k}_{\text{U}}})$ for all $k\in(\underline{k}_{\text{U}},\bar{k}_{\text{U}}]$, where $\mathfrak{h}(\epsilon,\Gamma_{\underline{k}_{\text{U}}|\underline{k}_{\text{U}}})=\sqrt{n}\left(\varrho_{{\text{m}}}+\epsilon\left(c_0+(n-2)\mathfrak{c}\left(\Gamma_{\underline{k}_{\text{U}}|\underline{k}_{\text{U}}}\right)\right)\right)$.
\end{prop}
\begin{proof}
The proof is by induction.

\noindent \textbf{Step 1}: Denote as $\check{\bm{\eta}}^{(1)}_{k|k}=(\bm{\eta}^{1}_{k|k},\bm{\eta}^{n}_{k|k})^T$ the mean error of the observable subsystem\footnote{A detailed description of the observable and unobservable subsystems is given in Appendix \ref{ap:subsystems}.} (i.e., the boundary cells). The error covariance of the observable subsystem $\check{\Gamma}_{k|k}^{(1)}$ satisfies
\begin{equation*}
\check{\Gamma}_{k|k}^{(1)}<r_2I,\quad\text{and}\quad\check{\Gamma}_{k|k-1}^{(1)}>q_1 I, \quad \text{for $k\in(\underline{k}_{\text{U}},\bar{k}_{\text{U}}]$.}
\end{equation*}
Let $\check{A}^{(1)}=I$ be the state transition matrix associated with the observable subsystem, it follows that
\begin{equation*}
\arraycolsep=1.5pt\def\arraystretch{1.5}
\begin{array}{ll}
&\text{ }\left\|\check{\bm{\eta}}^{(1)}_{\underline{k}_{\text{U}}+1|\underline{k}_{\text{U}}+1}\right\|\le \left\|\left(I-\check{K}^{(1)}_{\underline{k}_{\text{U}}+1}\check{H}^{(1)}\right)\check{A}^{(1)}\right\|\left\|\check{\bm{\eta}}^{(1)}_{\underline{k}_{\text{U}}|\underline{k}_{\text{U}}}\right\|\\
=&\left\|\check{\Gamma}^{(1)}_{\underline{k}_{\text{U}}+1|\underline{k}_{\text{U}}+1}\left(\check{\Gamma}_{\underline{k}_{\text{U}}+1|\underline{k}_{\text{U}}}^{(1)}\right)^{-1}\right\|\left\|\check{\bm{\eta}}^{(1)}_{\underline{k}_{\text{U}}|\underline{k}_{\text{U}}}\right\|<r_2q_1^{-1}\left\|\check{\bm{\eta}}^{(1)}_{\underline{k}_{\text{U}}|\underline{k}_{\text{U}}}\right\|.
\end{array}
\end{equation*}
Denote as $\check{\mathcal{I}}^{(1)}_{\cdot,\cdot}$ and $\check{\mathcal{C}}^{(1)}_{\cdot,\cdot}$ the information and controllability matrix of the observable subsystem, we have $2 r_2^{-1}I< \check{\mathcal{I}}^{(1)}_{k,k-1}=R_{k-1}^{-1}+R_k^{-1}< 2r_1^{-1}I$ and $q_1 I< \check{\mathcal{C}}^{(1)}_{k,k-1}=\check{Q}_{k}^{(1)}< q_2I$ for all $k\in(\underline{k}_{\text{U}}, \bar{k}_{\text{U}}]$, where $\check{Q}_{k}^{(1)}$ is the model error covariance for the observable subsystem. Hence $\check{c}_1 I< (\check{\Gamma}_{k|k}^{(1)})^{-1}
< \check{c}_2 I$ for all $k\in(\underline{k}_{\text{U}}, \bar{k}_{\text{U}}]$ according to Lemma 7.1 and 7.2 in \cite{Jazwinski1970}. Define the Lyapunov function of the observable subsystem as $\check{V}_k=(\check{\bm{\eta}}^{(1)}_{k|k})^{T}(\check{\Gamma}^{(1)}_{k|k})^{-1}\check{\bm{\eta}}^{(1)}_{k|k}$, then $\check{V}_{k+1}<\check{V}_k$ for all $k\in(\underline{k}_{\text{U}}, \bar{k}_{\text{U}})$ due to \cite[Lemma 3]{Olfati-SaberCDC2009}. Consequently,
\begin{equation*}
\begin{array}{l}
\left\|\check{\bm{\eta}}^{(1)}_{k|k}\right\|<\left(\frac{\check{V}_k}{\check{c}_1}\right)^{\frac{1}{2}}<\left(\frac{\check{V}_{\underline{k}_{\text{U}}+1}}{\check{c}_1}\right)^{\frac{1}{2}}< \sqrt{\check{c}_2\check{c}_1^{-1}}\left\|\check{\bm{\eta}}^{(1)}_{\underline{k}_{\text{U}}+1|\underline{k}_{\text{U}}+1}\right\|,
\end{array}
\end{equation*}
for all $k\in(\underline{k}_{\text{U}}+1,\bar{k}_{\text{U}}]$. It follows that for all $k\in(\underline{k}_{\text{U}},\bar{k}_{\text{U}}]$,
\begin{equation*}
\begin{array}{l}
\left\|\check{\bm{\eta}}^{(1)}_{k|k}\right\|<\sqrt{\check{c}_2\check{c}_1^{-1}}r_2q_1^{-1}\left\|\check{\bm{\eta}}^{(1)}_{\underline{k}_{\text{U}}|\underline{k}_{\text{U}}}\right\|<\sqrt{\check{c}_2\check{c}_1^{-1}}r_2q_1^{-1}\epsilon\le c_0\epsilon.
\end{array}
\end{equation*}

\noindent \textbf{Step 2}: We use induction to show that $\bm{\rho}^{l}_{k|k}>-c_0\epsilon-(l-1)\mathfrak{c}\left(\Gamma_{\underline{k}_{\text{U}}|\underline{k}_{\text{U}}}\right)\epsilon \ge -\epsilon(c_0+(n-2)\mathfrak{c}(\Gamma_{\underline{k}_{\text{U}}|\underline{k}_{\text{U}}}))$ for all $k\in(\underline{k}_{\text{U}},\bar{k}_{\text{U}}]$ and $l\in\{2,\cdots,n-1\}$. Since $|\bm{\eta}^1_{k|k}|<c_0\epsilon$ for all $k\in(\underline{k}_{\text{U}},\bar{k}_{\text{U}}]$, it holds that $\bm{\rho}^1_{k|k}>-c_0\epsilon=-c_0\epsilon-(1-1)\mathfrak{c}\left(\Gamma_{\underline{k}_{\text{U}}|\underline{k}_{\text{U}}}\right)\epsilon$.
Hence when $l=1$, $\bm{\rho}^1_{k|k}>-c_0\epsilon-(l-1)\mathfrak{c}\left(\Gamma_{\underline{k}_{\text{U}}|\underline{k}_{\text{U}}}\right)\epsilon$ holds for all $k\in(\underline{k}_{\text{U}},\bar{k}_{\text{U}}]$.

For $l\in\{1,2,\cdots,n-2\}$, suppose $\bm{\rho}^l_{k|k}>-c_0\epsilon-(l-1)\mathfrak{c}\left(\Gamma_{\underline{k}_{\text{U}}|\underline{k}_{\text{U}}}\right)\epsilon$ for all $k\in(\underline{k}_{\text{U}},\bar{k}_{\text{U}}]$. If $\bm{\rho}^{l+1}_{k|k}<-c_0\epsilon-l\mathfrak{c}\left(\Gamma_{\underline{k}_{\text{U}}|\underline{k}_{\text{U}}}\right)\epsilon$, we obtain from \eqref{eq:qflow} that
\begin{equation*}
\begin{array}{l}
\mathfrak{f}\left(\bm{\rho}^{l}_{k|k},\bm{\rho}^{l+1}_{k|k}\right)=v_{\text{m}}\bm{\rho}^{l}_{k|k}>v_{\text{m}}\left(-c_0\epsilon-(l-1)\mathfrak{c}\left(\Gamma_{\underline{k}_{\text{U}}|\underline{k}_{\text{U}}}\right)\epsilon\right),\\
\mathfrak{f}\left(\bm{\rho}^{l+1}_{k|k},\bm{\rho}^{l+2}_{k|k}\right)\leq v_{\text{m}}\bm{\rho}^{l+1}_{k|k}.
\end{array}
\end{equation*}
It follows that the estimate of cell $l+1$ satisfies
\begin{equation}\label{eq:induction}
\arraycolsep=1.5pt\def\arraystretch{1.5}
\begin{array}{rl}
\bm{\rho}^{l+1}_{k+1|k+1}&=\bm{\rho}^{1+1}_{k|k}+\frac{\Delta t}{\Delta x}\left(\mathfrak{f}\left(\bm{\rho}^{l}_{k|k},\bm{\rho}^{l+1}_{k|k}\right)-\mathfrak{f}\left(\bm{\rho}^{l+1}_{k|k},\bm{\rho}^{l+2}_{k|k}\right)\right)-K_{k+1}(l+1,1)\bm{\eta}^1_{k+1|k}-K_{k+1}(l+1,2)\bm{\eta}^2_{k+1|k}\\
&>\bm{\rho}^{l+1}_{k|k}+\frac{v_{\text{m}} \Delta t}{\Delta x}\left|\bm{\rho}^{l+1}_{k|k}+c_0\epsilon+(l-1)\mathfrak{c}\left(\Gamma_{\underline{k}_{\text{U}}|\underline{k}_{\text{U}}}\right)\epsilon\right|-\mathfrak{k}\left(\Gamma_{\underline{k}_{\text{U}}|\underline{k}_{\text{U}}}\right)c_0\epsilon\\
&=\bm{\rho}^{l+1}_{k|k}+\frac{v_{\text{m}} \Delta t}{\Delta x}\left|\bm{\rho}^{l+1}_{k|k}+c_0\epsilon+l\mathfrak{c}\left(\Gamma_{\underline{k}_{\text{U}}|\underline{k}_{\text{U}}}\right)\epsilon\right|+\frac{v_{\text{m}} \Delta t}{\Delta x}\mathfrak{c}\left(\Gamma_{\underline{k}_{\text{U}}|\underline{k}_{\text{U}}}\right)\epsilon-\mathfrak{k}\left(\Gamma_{\underline{k}_{\text{U}}|\underline{k}_{\text{U}}}\right)c_0\epsilon\\
&\ge \bm{\rho}^{l+1}_{k|k}+\frac{v_{\text{m}} \Delta t}{\Delta x}\left|\bm{\rho}^{l+1}_{k|k}+c_0\epsilon+l \mathfrak{c}\left(\Gamma_{\underline{k}_{\text{U}}|\underline{k}_{\text{U}}}\right)\epsilon\right|,
\end{array}
\end{equation}
where the first inequality is due to $\|K_k\|_{\infty}\le \mathfrak{k}(\Gamma_{\underline{k}_{\text{U}}|\underline{k}_{\text{U}}})$ given in Lemma \ref{lem:Lyapunov} and the fact that
$\|\check{\bm{\eta}}^{(1)}_{k+1|k}\|=\|\check{A}^{(1)}\check{\bm{\eta}}^{(1)}_{k|k}\|=\|\check{\bm{\eta}}^{(1)}_{k|k}\|< c_0\epsilon$ for all $k\in(\underline{k}_{\text{U}},\bar{k}_{\text{U}}]$, and the last inequality is obtained by $\frac{v_{\text{m}} \Delta t}{\Delta x}\mathfrak{c}(\Gamma_{\underline{k}_{\text{U}}|\underline{k}_{\text{U}}})\epsilon-\mathfrak{k}(\Gamma_{\underline{k}_{\text{U}}|\underline{k}_{\text{U}}})c_0\epsilon=\frac{v_{\text{m}}}{\min\left\{v_{\text{m}}, w\right\}}\mathfrak{k}(\Gamma_{\underline{k}_{\text{U}}|\underline{k}_{\text{U}}})c_0\epsilon-\mathfrak{k}(\Gamma_{\underline{k}_{\text{U}}|\underline{k}_{\text{U}}})c_0\epsilon\ge 0$. Also since $\bm{\rho}^{l+1}_{\underline{k}_{\text{U}}|\underline{k}_{\text{U}}}>-\epsilon\ge-c_0\epsilon>-c_0\epsilon-l \mathfrak{c}(\Gamma_{\underline{k}_{\text{U}}|\underline{k}_{\text{U}}})\epsilon$, it is concluded that $\bm{\rho}^{l+1}_{k|k}>-c_0\epsilon-l \mathfrak{c}(\Gamma_{\underline{k}_{\text{U}}|\underline{k}_{\text{U}}})\epsilon$ for all $k\in(\underline{k}_{\text{U}},\bar{k}_{\text{U}}]$. Continuing the induction along the cells, we obtain $\bm{\rho}^{n-1}_{k|k}>-c_0\epsilon-(n-2) \mathfrak{c}(\Gamma_{\underline{k}_{\text{U}}|\underline{k}_{\text{U}}})\epsilon \le $ for all $k\in(\underline{k}_{\text{U}},\bar{k}_{\text{U}}]$.

We can use a similar induction to show $\bm{\rho}^l_{k|k}<\varrho_{\text{m}}+c_0\epsilon+(n-l)\mathfrak{c}(\Gamma_{\underline{k}_{\text{U}}|\underline{k}_{\text{U}}})\epsilon\le \varrho_{\text{m}}+\epsilon(c_0+(n-2)\mathfrak{c}(\Gamma_{\underline{k}_{\text{U}}|\underline{k}_{\text{U}}}))$ for all $k\in(\underline{k}_{\text{U}},\bar{k}_{\text{U}}]$ and $l\in\{2,\cdots,n-1\}$.

Since $|\bm{\eta}^n_{k|k}|<c_0\epsilon$ for all $k\in(\underline{k}_{\text{U}},\bar{k}_{\text{U}}]$, we have $\bm{\rho}^n_{k|k}<\varrho_{\text{m}}+c_0\epsilon=\varrho_{\text{m}}+c_0\epsilon+(n-n)\mathfrak{c}\left(\Gamma_{\underline{k}_{\text{U}}|\underline{k}_{\text{U}}}\right)\epsilon$.
Hence when $l=n$, $\bm{\rho}^l_{k|k}<\varrho_{\text{m}}+c_0\epsilon+(n-l)\mathfrak{c}\left(\Gamma_{\underline{k}_{\text{U}}|\underline{k}_{\text{U}}}\right)\epsilon$ holds for all $k\in(\underline{k}_{\text{U}},\bar{k}_{\text{U}}]$.

For $l\in\{n-1,n-2,\cdots,2\}$, suppose $\bm{\rho}^l_{k|k}<\varrho_{\text{m}}+c_0\epsilon+(n-l)\mathfrak{c}\left(\Gamma_{\underline{k}_{\text{U}}|\underline{k}_{\text{U}}}\right)\epsilon$ for all $k\in(\underline{k}_{\text{U}},\bar{k}_{\text{U}}]$. If $\bm{\rho}^{l-1}_{k|k}>\varrho_{\text{m}}+c_0\epsilon+(n-l+1)\mathfrak{c}\left(\Gamma_{\underline{k}_{\text{U}}|\underline{k}_{\text{U}}}\right)\epsilon$, following the similar argument as in \eqref{eq:induction} yields
\begin{align*}
\bm{\rho}^{l-1}_{k+1|k+1}&<\bm{\rho}^{l-1}_{k|k}-\frac{w \Delta t}{\Delta x}\left|\bm{\rho}^{l-1}_{k|k}-\varrho_{\text{m}}-c_0\epsilon-(n-l)\mathfrak{c}\left(\Gamma_{\underline{k}_{\text{U}}|\underline{k}_{\text{U}}}\right)\epsilon\right|+\mathfrak{k}\left(\Gamma_{\underline{k}_{\text{U}}|\underline{k}_{\text{U}}}\right)c_0\epsilon\\
&=\bm{\rho}^{l-1}_{k|k}-\frac{w \Delta t}{\Delta x}\left|\bm{\rho}^{l-1}_{k|k}-\varrho_{\text{m}}-c_0\epsilon-(n-l+1)\mathfrak{c}\left(\Gamma_{\underline{k}_{\text{U}}|\underline{k}_{\text{U}}}\right)\epsilon\right|-\frac{w \Delta t}{\Delta x}\mathfrak{c}\left(\Gamma_{\underline{k}_{\text{U}}|\underline{k}_{\text{U}}}\right)\epsilon+\mathfrak{k}\left(\Gamma_{\underline{k}_{\text{U}}|\underline{k}_{\text{U}}}\right)c_0\epsilon\\
&\le \bm{\rho}^{l-1}_{k|k}-\frac{w \Delta t}{\Delta x}\left|\bm{\rho}^{l-1}_{k|k}-\varrho_{\text{m}}-c_0\epsilon-(n-l+1)\mathfrak{c}\left(\Gamma_{\underline{k}_{\text{U}}|\underline{k}_{\text{U}}}\right)\epsilon\right|.
\end{align*}
Also since $\bm{\rho}^{l-1}_{\underline{k}_{\text{U}}|\underline{k}_{\text{U}}}<\varrho_{\text{m}}+\epsilon\le \varrho_{\text{m}}+c_0\epsilon<\varrho_{\text{m}}+c_0\epsilon+(n-l+1) \mathfrak{c}\left(\Gamma_{\underline{k}_{\text{U}}|\underline{k}_{\text{U}}}\right)\epsilon$, it is concluded that $\bm{\rho}^{l-1}_{k|k}<\varrho_{\text{m}}+c_0\epsilon+(n-l+1) \mathfrak{c}\left(\Gamma_{\underline{k}_{\text{U}}|\underline{k}_{\text{U}}}\right)\epsilon$ for all $k\in(\underline{k}_{\text{U}},\bar{k}_{\text{U}}]$. Continuing the induction, we obtain $\bm{\rho}^{2}_{k|k}<\varrho_{\text{m}}+c_0\epsilon+(n-2) \mathfrak{c}\left(\Gamma_{\underline{k}_{\text{U}}|\underline{k}_{\text{U}}}\right)\epsilon$ for all $k\in(\underline{k}_{\text{U}},\bar{k}_{\text{U}}]$.

\noindent \textbf{Step 3}: Combining Steps 1 and 2, we obtain $\bm{\rho}^{l}_{k|k}\in(-\epsilon(c_0+(n-2)\mathfrak{c}(\Gamma_{\underline{k}_{\text{U}}|\underline{k}_{\text{U}}})),\varrho_{\text{m}}+\epsilon(c_0+(n-2)\mathfrak{c}(\Gamma_{\underline{k}_{\text{U}}|\underline{k}_{\text{U}}})))$ for all $l\in\{1,\cdots,n\}$ and $k\in(\underline{k}_{\text{U}},\bar{k}_{\text{U}}]$. Consequently, $\|\bm{\eta}_{k|k}\|<\sqrt{n}(\varrho_{\text{m}}+\epsilon(c_0+(n-2)\mathfrak{c}(\Gamma_{\underline{k}_{\text{U}}|\underline{k}_{\text{U}}})))=\mathfrak{h}(\epsilon, \Gamma_{\underline{k}_{\text{U}}|\underline{k}_{\text{U}}})$ for all $k\in(\underline{k}_{\text{U}},\bar{k}_{\text{U}}]$.
\end{proof}
\subsubsection{Convergence rate of the mean error in observable modes}
Let $(\underline{k}_{\text{O}}, \bar{k}_{\text{O}}]$ be the time interval inside which a section switches among observable modes, i.e., the mode index $\sigma(k)\in\{\text{FF, CC, CF}\}$ for $k\in (\underline{k}_{\text{O}}, \bar{k}_{\text{O}}]$, and $\sigma(k)\in\{\text{FC1, FC2}\}$ for $k=\underline{k}_{\text{O}}$ and $k=\bar{k}_{\text{O}}+1$. Due to the boundedness of the consensus term described in \eqref{eq:bound_consensus}, the mean error satisfies
\begin{equation}\label{eq:error_dynamics_o_c}
\arraycolsep=1.5pt\def\arraystretch{1.5}
\begin{array}{l}
\left\|\bm{\eta}_{k|k}\right\|\le \left\|\prod_{\kappa=k-1}^{\underline{k}_{\text{O}}}F_{\kappa+1}A_{\kappa}\right\|\left\|\bm{\eta}_{\underline{k}_{\text{O}}|\underline{k}_{\text{O}}}\right\|+\hat{c}\left(1+\sum_{\iota=1}^{k-\underline{k}_{\text{O}}-1}\left\|\prod_{\kappa=k-1}^{\underline{k}_{\text{O}}+\iota}F_{\kappa+1}A_{\kappa}\right\|\right),
\end{array}
\end{equation}
for $k\in (\underline{k}_{\text{O}}, \bar{k}_{\text{O}}]$, where $F_{k}=I-K_{k}H_{k}$. According to \eqref{eq:error_dynamics_o_c}, we need to analyse the magnitude of $\left\|\prod_{\kappa=k-1}^{\underline{k}_{\text{O}}}F_{\kappa+1}A_{\kappa}\right\|$ in order to study the convergence rate of the mean estimation error, which is detailed in the next lemma.

\begin{lem}[Lemma 3 in \cite{SunWorkCONES2016}]\label{lem:Lyapunov}
Consider a freeway section that switches among the observable modes while $k\in (\underline{k}_{\text{O}},\bar{k}_{\text{O}}]$, where $0\le \underline{k}_{\text{O}}<\bar{k}_{\text{O}}\le+\infty$. If the error covariance satisfies $\bm{0}<d_1I\le\Gamma^{-1}_{k|k}\le d_2 I$ for all $\underline{k}_{\text{O}}<k\le \bar{k}_{\text{O}}$, where $d_1, d_2 \in \mathbb{R}^{+}$, then
\begin{equation}\label{eq:norm_exp}
\begin{array}{l}
\left\|\prod_{\kappa=k-1}^{\underline{k}_{\text{O}}}F_{\kappa+1}A_{\kappa}\right\|\le \hat{a}\hat{q}^{k-\underline{k}_{\text{O}}},\quad\text{for $k \in (\underline{k}_{\text{O}}, \bar{k}_{\text{O}}$}],
\end{array}
\end{equation}
where $\hat{a}=\left(d_2d_1^{-1}\right)^{\frac{1}{2}}\ge 1$, $0<\hat{q}=\left(1-\mathfrak{d}\left(d_1,d_2\right)d_2^{-1}\right)^{\frac{1}{2}}<1$, and $\mathfrak{d}\left(\cdot,\cdot\right)$ is a function of $d_1, d_2$ defined by
\begin{equation*}
\mathfrak{d}\left(d_1,d_2\right)=\left(d_1^{-1}+q_1^{-1}d_1^{-2}\max_{M\in\mathcal{A}_{\text{O}}}\sigma^2_{\max}\left(M\right)\right)^{-1},
\end{equation*}
where $\mathcal{A}_{\text{O}}=\left\{A_{\text{FF}},A_{\text{CC}},A_{\text{CF}}^s\left|s\in\left\{1,2,\cdots,n-1\right\}\right.\right\}$ and $\sigma_{\max}(M)$ is the maximum singular value of matrix $M$.
\end{lem}

\begin{proof}
The proof is reported in Appendix \ref{ap:Proof_Lyap}.
\end{proof}

\subsubsection{Residence time in observable modes}

When a freeway section switches from an unobservable mode at time $\underline{k}_{\text{O}}$ to an observable mode at $\underline{k}_{\text{O}}+1$, the next proposition derives the residence time the section must remain in the set of observable modes in order to reduce the mean estimation error below a given threshold. The residence time is a function of the mean error and error covariance of the section at time $\underline{k}_{\text{O}}$, and also depends on the magnitude of the mean error to be satisfied.

\begin{prop}[Proposition 4 in \cite{SunWorkCONES2016}]\label{prop_o_bound}
Consider a freeway section which switches among the observable modes while $k\in (\underline{k}_{\text{O}},\bar{k}_{\text{O}}]$, where $0\le \underline{k}_{\text{O}}<\bar{k}_{\text{O}}\le+\infty$. Define
\begin{equation}\label{eq:tilde_a_q1}
\begin{array}{l}
\mathfrak{a}\left(\Gamma_{\underline{k}_{\text{O}}|\underline{k}_{\text{O}}}\right)=\left(\mathfrak{c}_2\left(\Gamma_{\underline{k}_{\text{O}}|\underline{k}_{\text{O}}}\right)\left(\mathfrak{c}_1\left(\Gamma_{\underline{k}_{\text{O}}|\underline{k}_{\text{O}}}\right)\right)^{-1}\right)^{\frac{1}{2}},\\ \mathfrak{q}\left(\Gamma_{\underline{k}_{\text{O}}|\underline{k}_{\text{O}}}\right)=\left(1-\mathfrak{c}_3\left(\Gamma_{\underline{k}_{\text{O}}|\underline{k}_{\text{O}}}\right)\left(\mathfrak{c}_2\left(\Gamma_{\underline{k}_{\text{O}}|\underline{k}_{\text{O}}}\right)\right)^{-1}\right)^{\frac{1}{2}},
\end{array}
\end{equation}
where $\mathfrak{c}_1\left(\cdot\right)$, $\mathfrak{c}_2\left(\cdot \right)$ are the bounds from \eqref{eq:def_frak_c}, and $\mathfrak{c}_3\left(\Gamma_{\underline{k}_{\text{O}}|\underline{k}_{\text{O}}}\right)$ is given by $\mathfrak{c}_3\left(\Gamma_{\underline{k}_{\text{O}}|\underline{k}_{\text{O}}}\right)=\mathfrak{d}\left(\mathfrak{c}_1\left(\Gamma_{\underline{k}_{\text{O}}|\underline{k}_{\text{O}}}\right),\mathfrak{c}_2\left(\Gamma_{\underline{k}_{\text{O}}|\underline{k}_{\text{O}}}\right)\right)$ with $\mathfrak{d}\left(\cdot,\cdot\right)$ defined in Lemma \ref{lem:Lyapunov}.

For all $\epsilon>0$, there exists $\mathfrak{t}\left(\epsilon,\left\|\bm{\eta}_{\underline{k}_{\text{O}}|\underline{k}_{\text{O}}}\right\|,\Gamma_{\underline{k}_{\text{O}}|\underline{k}_{\text{O}}}\right)$ such that if $\bar{k}_{\text{O}}-\underline{k}_{\text{O}}>\mathfrak{t}\left(\epsilon,\left\|\bm{\eta}_{\underline{k}_{\text{O}}|\underline{k}_{\text{O}}}\right\|,\Gamma_{\underline{k}_{\text{O}}|\underline{k}_{\text{O}}}\right)$, the mean error at time $\bar{k}_{\text{O}}$ satisfies $\|\bm{\eta}_{\bar{k}_{\text{O}}|\bar{k}_{\text{O}}}\|<\epsilon+\hat{c}+\frac{\hat{c}\mathfrak{a}\left(\Gamma_{\underline{k}_{\text{O}}|\underline{k}_{\text{O}}}\right)\mathfrak{q}\left(\Gamma_{\underline{k}_{\text{O}}|\underline{k}_{\text{O}}}\right)}{1-\mathfrak{q}\left(\Gamma_{\underline{k}_{\text{O}}|\underline{k}_{\text{O}}}\right)}$. Explicitly,
\begin{equation}\label{eq:tilde_T}
\begin{split}
&\mathfrak{t}\left(\epsilon,\left\|\bm{\eta}_{\underline{k}_{\text{O}}|\underline{k}_{\text{O}}}\right\|,\Gamma_{\underline{k}_{\text{O}}|\underline{k}_{\text{O}}}\right)= \left\{
\arraycolsep=1.5pt\def\arraystretch{1.5}
\begin{array}{l}
0,\quad \textrm{if $\mathfrak{a}\left(\Gamma_{\underline{k}_{\text{O}}|\underline{k}_{\text{O}}}\right)\mathfrak{q}\left(\Gamma_{\underline{k}_{\text{O}}|\underline{k}_{\text{O}}}\right) \left\|\bm{\eta}_{\underline{k}_{\text{O}}|\underline{k}_{\text{O}}}\right\|\le \frac{\hat{c}\mathfrak{a}\left(\Gamma_{\underline{k}_{\text{O}}|\underline{k}_{\text{O}}}\right)\mathfrak{q}\left(\Gamma_{\underline{k}_{\text{O}}|\underline{k}_{\text{O}}}\right)}{1-\mathfrak{q}\left(\Gamma_{\underline{k}_{\text{O}}|\underline{k}_{\text{O}}}\right)}$,}\\
\log_{\mathfrak{q}\left(\Gamma_{\underline{k}_{\text{O}}|\underline{k}_{\text{O}}}\right)}\left(\epsilon\left(\mathfrak{a}\left(\Gamma_{\underline{k}_{\text{O}}|\underline{k}_{\text{O}}}\right)\left\|\bm{\eta}_{\underline{k}_{\text{O}}|\underline{k}_{\text{O}}}\right\|-\frac{\hat{c}\mathfrak{a}\left(\Gamma_{\underline{k}_{\text{O}}|\underline{k}_{\text{O}}}\right)}{1-\mathfrak{q}\left(\Gamma_{\underline{k}_{\text{O}}|\underline{k}_{\text{O}}}\right)}\right)^{-1}\right),\quad \text{otherwise.}
\end{array} \right.
\end{split}
\end{equation}
Furthermore, for all $k\in(\underline{k}_{\text{O}},\bar{k}_{\text{O}}]$,
\begin{equation*}
\begin{split}
\left\|\bm{\eta}_{k|k}\right\|\le& \max\left\{\hat{c}+\mathfrak{a}\left(\Gamma_{\underline{k}_{\text{O}}|\underline{k}_{\text{O}}}\right)\mathfrak{q}\left(\Gamma_{\underline{k}_{\text{O}}|\underline{k}_{\text{O}}}\right) \left\|\bm{\eta}_{\underline{k}_{\text{O}}|\underline{k}_{\text{O}}}\right\|,\quad\hat{c}+\hat{c}\mathfrak{a}\left(\Gamma_{\underline{k}_{\text{O}}|\underline{k}_{\text{O}}}\right)\mathfrak{q}\left(\Gamma_{\underline{k}_{\text{O}}|\underline{k}_{\text{O}}}\right)\left(1-\mathfrak{q}\left(\Gamma_{\underline{k}_{\text{O}}|\underline{k}_{\text{O}}}\right)\right)^{-1} \right\}.
\end{split}
\end{equation*}

\end{prop}
\begin{proof}

According to Lemma \ref{prop:GammaBound}, when $\underline{k}_{\text{O}}< k \le \bar{k}_{\text{O}}$ the error covariance satisfies $\mathfrak{c}_1(\Gamma_{\underline{k}_{\text{O}}|\underline{k}_{\text{O}}})I\le \Gamma^{-1}_{k|k} \le \mathfrak{c}_2(\Gamma_{\underline{k}_{\text{O}}|\underline{k}_{\text{O}}})I$.
Given Lemma \ref{lem:Lyapunov}, it follows that for $\underline{k}_{\text{O}}< k \le \bar{k}_{\text{O}}$,
\begin{equation*}
\begin{array}{l}
\left\|\prod_{\kappa=k-1}^{\underline{k}_{\text{O}}}F_{\kappa+1}A_{\kappa}\right\|\le \mathfrak{a}\left(\Gamma_{\underline{k}_{\text{O}}|\underline{k}_{\text{O}}}\right)\mathfrak{q}\left(\Gamma_{\underline{k}_{\text{O}}|\underline{k}_{\text{O}}}\right)^{k-\underline{k}_{\text{O}}},
\end{array}
\end{equation*}
where $\mathfrak{a}\left(\Gamma_{\underline{k}_{\text{O}}|\underline{k}_{\text{O}}}\right)\ge 1$ provides an upper bound for the increase of the mean estimation error when the section first switches to an observable mode at time $\underline{k}_{\text{O}}+1$, and $0<\mathfrak{q}\left(\Gamma_{\underline{k}_{\text{O}}|\underline{k}_{\text{O}}}\right)<1$ describes the convergence rate of the mean estimation error in observable modes. Hence when $\underline{k}_{\text{O}}< k \le \bar{k}_{\text{O}}$, the 2-norm of $\bm{\eta}_{k|k}$ satisfies
\begin{equation*}
\arraycolsep=1.5pt\def\arraystretch{1.5}
\begin{array}{rl}
\left\|\bm{\eta}_{k|k}\right\|&\le \left\|\prod_{\kappa=k-1}^{\underline{k}_{\text{O}}}F_{\kappa+1}A_{\kappa}\right\|\left\|\bm{\eta}_{\underline{k}_{\text{O}}|\underline{k}_{\text{O}}}\right\|+\hat{c}\left(1+\sum_{\iota=1}^{k-\underline{k}_{\text{O}}-1}\left\|\prod_{\kappa=k-1}^{\underline{k}_{\text{O}}+\iota}F_{\kappa+1}A_{\kappa}\right\|\right)\\
&\le \hat{c}+\left\|\bm{\eta}_{\underline{k}_{\text{O}}|\underline{k}_{\text{O}}}\right\|\mathfrak{a}\left(\Gamma_{\underline{k}_{\text{O}}|\underline{k}_{\text{O}}}\right)\mathfrak{q}\left(\Gamma_{\underline{k}_{\text{O}}|\underline{k}_{\text{O}}}\right)^{k-\underline{k}_{\text{O}}}+\sum_{\iota=1}^{k-\underline{k}_{\text{O}}-1}\hat{c}\mathfrak{a}\left(\Gamma_{\underline{k}_{\text{O}}|\underline{k}_{\text{O}}}\right)\mathfrak{q}\left(\Gamma_{\underline{k}_{\text{O}}|\underline{k}_{\text{O}}}\right)^{k-\underline{k}_{\text{O}}-\iota}\\
&= \hat{c}+\left\|\bm{\eta}_{\underline{k}_{\text{O}}|\underline{k}_{\text{O}}}\right\|\mathfrak{a}\left(\Gamma_{\underline{k}_{\text{O}}|\underline{k}_{\text{O}}}\right)\mathfrak{q}\left(\Gamma_{\underline{k}_{\text{O}}|\underline{k}_{\text{O}}}\right)^{k-\underline{k}_{\text{O}}}+\frac{\hat{c}\mathfrak{a}\left(\Gamma_{\underline{k}_{\text{O}}|\underline{k}_{\text{O}}}\right)\mathfrak{q}\left(\Gamma_{\underline{k}_{\text{O}}|\underline{k}_{\text{O}}}\right)}{1-\mathfrak{q}\left(\Gamma_{\underline{k}_{\text{O}}|\underline{k}_{\text{O}}}\right)}\left(1-\mathfrak{q}\left(\Gamma_{\underline{k}_{\text{O}}|\underline{k}_{\text{O}}}\right)^{k-\underline{k}_{\text{O}}-1}\right)\\
&\triangleq \mathfrak{u}\left(\Gamma_{\underline{k}_{\text{O}}|\underline{k}_{\text{O}}}, k\right),
\end{array}
\end{equation*}
where for a fixed $\Gamma_{\underline{k}_{\text{O}}|\underline{k}_{\text{O}}}$, the function $\mathfrak{u}\left(\Gamma_{\underline{k}_{\text{O}}|\underline{k}_{\text{O}}}, k\right)$ is either non-increasing or non-decreasing with respect to $k$.
As a consequence, for all $\epsilon>0$, there exists $\mathfrak{t}(\epsilon,\|\bm{\eta}_{\underline{k}_{\text{O}}|\underline{k}_{\text{O}}}\|,\Gamma_{\underline{k}_{\text{O}}|\underline{k}_{\text{O}}})\ge 0$ such that for all $k-\underline{k}_{\text{O}}>\mathfrak{t}(\epsilon,\|\bm{\eta}_{\underline{k}_{\text{O}}|\underline{k}_{\text{O}}}\|,\Gamma_{\underline{k}_{\text{O}}|\underline{k}_{\text{O}}})$,
\begin{align*}
\left\|\bm{\eta}_{k|k}\right\|<\epsilon+\hat{c}+\frac{\hat{c}\mathfrak{a}\left(\Gamma_{\underline{k}_{\text{O}}|\underline{k}_{\text{O}}}\right)\mathfrak{q}\left(\Gamma_{\underline{k}_{\text{O}}|\underline{k}_{\text{O}}}\right)}{1-\mathfrak{q}\left(\Gamma_{\underline{k}_{\text{O}}|\underline{k}_{\text{O}}}\right)}.
\end{align*}
When $\mathfrak{a}\left(\Gamma_{\underline{k}_{\text{O}}|\underline{k}_{\text{O}}}\right)\mathfrak{q}\left(\Gamma_{\underline{k}_{\text{O}}|\underline{k}_{\text{O}}}\right) \left\|\bm{\eta}_{\underline{k}_{\text{O}}|\underline{k}_{\text{O}}}\right\|\le \frac{\hat{c}\mathfrak{a}\left(\Gamma_{\underline{k}_{\text{O}}|\underline{k}_{\text{O}}}\right)\mathfrak{q}\left(\Gamma_{\underline{k}_{\text{O}}|\underline{k}_{\text{O}}}\right)}{1-\mathfrak{q}\left(\Gamma_{\underline{k}_{\text{O}}|\underline{k}_{\text{O}}}\right)}$, we have $\mathfrak{u}(\Gamma_{\underline{k}_{\text{O}}|\underline{k}_{\text{O}}}, \underline{k}_{\text{O}}+1)\le \lim_{k \rightarrow \infty}\mathfrak{u}(\Gamma_{\underline{k}_{\text{O}}|\underline{k}_{\text{O}}}, k)$, and $\mathfrak{u}\left(\Gamma_{\underline{k}_{\text{O}}|\underline{k}_{\text{O}}}, k\right)\le  \lim_{k \rightarrow \infty}\mathfrak{u}(\Gamma_{\underline{k}_{\text{O}}|\underline{k}_{\text{O}}}, k)$ non-decreasing with respect to $k\in(\underline{k}_{\text{O}},\bar{k}_{\text{O}}]$, thus $\mathfrak{t}(\epsilon,\|\bm{\eta}_{\underline{k}_{\text{O}}|\underline{k}_{\text{O}}}\|,\Gamma_{\underline{k}_{\text{O}}|\underline{k}_{\text{O}}})=0$. On the other hand, $\mathfrak{u}\left(\Gamma_{\underline{k}_{\text{O}}|\underline{k}_{\text{O}}}, k\right)$ is decreasing with respect to $k$ when $\mathfrak{u}(\Gamma_{\underline{k}_{\text{O}}|\underline{k}_{\text{O}}}, \underline{k}_{\text{O}}+1)> \lim_{k \rightarrow \infty}\mathfrak{u}(\Gamma_{\underline{k}_{\text{O}}|\underline{k}_{\text{O}}}, k)$. In this case,
\begin{equation*}
\arraycolsep=1pt\def\arraystretch{1}
\begin{array}{l}
\mathfrak{t}(\epsilon,\|\bm{\eta}_{\underline{k}_{\text{O}}|\underline{k}_{\text{O}}}\|,\Gamma_{\underline{k}_{\text{O}}|\underline{k}_{\text{O}}})=\log_{\mathfrak{q}\left(\Gamma_{\underline{k}_{\text{O}}|\underline{k}_{\text{O}}}\right)}\left(\epsilon\left(\mathfrak{a}\left(\Gamma_{\underline{k}_{\text{O}}|\underline{k}_{\text{O}}}\right)\left\|\bm{\eta}_{\underline{k}_{\text{O}}|\underline{k}_{\text{O}}}\right\|-\frac{\hat{c}\mathfrak{a}\left(\Gamma_{\underline{k}_{\text{O}}|\underline{k}_{\text{O}}}\right)}{1-\mathfrak{q}\left(\Gamma_{\underline{k}_{\text{O}}|\underline{k}_{\text{O}}}\right)}\right)^{-1}\right).
\end{array}
\end{equation*}


Furthermore, the upper bound of $\left\|\bm{\eta}_{k|k}\right\|$ is given as follows:
\begin{equation*}
\arraycolsep=1.5pt\def\arraystretch{1.5}
\begin{array}{rl}
\left\|\bm{\eta}_{k|k}\right\|&\le \max\left\{\mathfrak{u}\left(\Gamma_{\underline{k}_{\text{O}}|\underline{k}_{\text{O}}}, \underline{k}_{\text{O}}+1\right), \lim_{k\rightarrow \infty}\mathfrak{u}\left(\Gamma_{\underline{k}_{\text{O}}|\underline{k}_{\text{O}}}, k\right) \right\}\\
&=\max\left\{\hat{c}+\mathfrak{a}\left(\Gamma_{\underline{k}_{\text{O}}|\underline{k}_{\text{O}}}\right)\mathfrak{q}\left(\Gamma_{\underline{k}_{\text{O}}|\underline{k}_{\text{O}}}\right) \left\|\bm{\eta}_{\underline{k}_{\text{O}}|\underline{k}_{\text{O}}}\right\|,\quad\hat{c}+\hat{c}\mathfrak{a}\left(\Gamma_{\underline{k}_{\text{O}}|\underline{k}_{\text{O}}}\right)\mathfrak{q}\left(\Gamma_{\underline{k}_{\text{O}}|\underline{k}_{\text{O}}}\right)\left(1-\mathfrak{q}\left(\Gamma_{\underline{k}_{\text{O}}|\underline{k}_{\text{O}}}\right)\right)^{-1} \right\}.
\end{array}
\end{equation*}
for all $k\in(\underline{k}_{\text{O}},\bar{k}_{\text{O}}]$, which concludes the proof.
\end{proof}

\subsubsection{Boundedness of the mean estimation error under switches among observable and unobservable modes}

Based on Proposition \ref{prop_u_bound} and Proposition \ref{prop_o_bound}, the boundedness of the mean estimation error when the SMM switches among observable and unobservable modes is summarized in Proposition \ref{prop:switch}.

The main concept of Proposition \ref{prop:switch} is given as follows. For a freeway section, denote as $(\underline{k}_{\text{U}}^r,\bar{k}_{\text{U}}^r]$ and $(\underline{k}_{\text{O}}^r,\bar{k}_{\text{O}}^r]$ the $r^{\text{th}}$ unobservable and observable time intervals, respectively. Consider a freeway section that switches from an observable mode at $\bar{k}_{\text{O}}^{r-1}=\underline{k}_{\text{U}}^r$ to an unobservable mode at $\underline{k}_{\text{U}}^r+1$, and remains unobservable through $\bar{k}_{\text{U}}^r$. An upper bound for the 2-norm of the mean estimation error, which is uniform over $(\underline{k}_{\text{U}}^r,\bar{k}_{\text{U}}^r]$, can be obtained through Proposition \ref{prop_u_bound} based on the error covariance and the upper bound of the mean error at time $\underline{k}_{\text{U}}^r$. When the section switches back to the set of observable modes at time $\bar{k}_{\text{U}}^r+1=\underline{k}_{\text{O}}^r+1$ and remains observable through $\bar{k}_{\text{O}}^r$, the mean estimation error has been increased during the unobservable time interval, and may continue to increase initially before decreasing while the section is observable. Based on Proposition \ref{prop_o_bound}, the minimum residence time $\bar{k}_{\text{O}}^r-\underline{k}_{\text{O}}^r$ the section must remain observable to offset the increase of the mean estimation error, as well as the upper bound of the mean error during the observable interval $(\underline{k}_{\text{O}}^r, \bar{k}_{\text{O}}^r]$ are derived. The minimum residence time ensures that when the section switches back to an unobservable mode, the mean estimation error is smaller than a given upper bound. Based on this upper bound and the error covariance at time $\bar{k}_{\text{O}}^r=\underline{k}_{\text{U}}^{r+1}$, we can apply Proposition \ref{prop_u_bound} again and obtain the upper bound for the 2-norm of the mean estimation error during the unobservable time interval starting at time $\underline{k}_{\text{U}}^{r+1}+1$. We continue the induction and derive the minimum residence time for each observable time interval, as well as the upper bounds of the 2-norm of the mean estimation error for all the observable and unobservable time intervals.

\begin{prop}[Proposition 5 in \cite{SunWorkCONES2016}]\label{prop:switch}
For a freeway section, denote as $(\underline{k}_{{\text{U}}}^r,\bar{k}_{{\text{U}}}^r]$ the $r^{\text{th}}$ time interval while the section switches among unobservable modes, and $(\underline{k}_{{\text{O}}}^r,\bar{k}_{{\text{O}}}^r]$ the $r^{\text{th}}$ time interval while the section switches among observable modes. Hence $\underline{k}_{{\text{U}}}^1=0$ (resp. $\underline{k}_{{\text{O}}}^1=0$) when the section is unobservable (resp. observable) at time $0$. Let $\delta >0$ be an arbitrary positive constant, and suppose the following condition on the residence time for the observable time intervals holds:
\begin{flalign}\label{eq:time_condition}
&\bar{k}_{{\text{O}}}^r-\underline{k}_{{\text{O}}}^r>\left\{
\arraycolsep=1.5pt\def\arraystretch{1}
\begin{array}{ll}
\mathfrak{t}\left(\delta,\mathfrak{e}\left(\delta,\Gamma_{\underline{k}_{{\text{O}}}^{r-1}|\underline{k}_{{\text{O}}}^{r-1}},\Gamma_{\bar{k}_{{\text{O}}}^{r-1}|\bar{k}_{{\text{O}}}^{r-1}}\right),\Gamma_{\underline{k}_{{\text{O}}}^{r}|\underline{k}_{{\text{O}}}^{r}}\right)  &\quad r\ge 2\\
\mathfrak{t}\left(\delta,\mathfrak{e}_{0}\left(\Gamma_{0|0}\right),\Gamma_{\underline{k}_{{\text{O}}}^1|\underline{k}_{{\text{O}}}^1}\right) &\quad  r=1 \text{ and } \underline{k}_{{\text{U}}}^{1}=0\\
\mathfrak{t}\left(\delta,\sqrt{n}\varrho_{{\text{m}}},\Gamma_{0|0}\right) &\quad r=1 \text{ and } \underline{k}_{{\text{O}}}^{1}=0,
\end{array} \right.
\end{flalign}
where $\mathfrak{e}_0\left(M\right)=\sqrt{n}\left(\sqrt{n}\varrho_{{\text{m}}}\left(c_0+\left(n-2\right)\mathfrak{c}\left(M\right)\right)+\varrho_{{\text{m}}} \right)$ for $M\in\mathbb{R}^{n\times n}$, and
\begin{align*}
\arraycolsep=1.5pt\def\arraystretch{1}
\begin{array}{rl}
\mathfrak{e}\left(\delta, M_1,M_2\right)=&
\sqrt{n}\left(\varrho_{{\text{m}}}+\left(\delta+\hat{c}+\frac{\hat{c}\mathfrak{a}\left(M_1\right)\mathfrak{q}\left(M_1\right)}{1-\mathfrak{q}\left(M_1\right)}\right)\left(c_0+\left(n-2\right)\mathfrak{c}\left(M_2\right)\right)\right),
\end{array}
\end{align*}
for $M_1, M_2\in\mathbb{R}^{n\times n}$, with $\hat{c}$ given in \eqref{eq:bound_consensus}, $c_0$ and $\mathfrak{c}(\cdot)$ defined in \eqref{eq:gamma_c}, $\mathfrak{a}\left(\cdot\right)$ and $\mathfrak{q}\left(\cdot\right)$ defined in \eqref{eq:tilde_a_q1}.

When $r\ge 2$, the mean error is upper bounded a follows:
\begin{flalign*}
&\left\|\bm{\eta}_{k|k}\right\|\le \left\{
\arraycolsep=1.5pt\def\arraystretch{1}
\begin{array}{rl}
\text{(a) for $k\in (\underline{k}_{\text{U}}^{r},\bar{k}_{\text{U}}^{r}]$: }&\mathfrak{e}\left(\delta,\Gamma_{\bar{k}_{\text{U}}^{r-1}|\bar{k}_{\text{U}}^{r-1}}, \Gamma_{\underline{k}_{\text{U}}^{r}|\underline{k}_{\text{U}}^{r}}\right),\\
\text{(b) for $k\in (\underline{k}_{\text{O}}^{r},\bar{k}_{\text{O}}^{r}]$: } &\max\left\{\hat{c}+\mathfrak{a}\left(\Gamma_{\underline{k}_{\text{O}}^{r}|\underline{k}_{\text{O}}^{r}}\right)\mathfrak{q}\left(\Gamma_{\underline{k}_{\text{O}}^{r}|\underline{k}_{\text{O}}^{r}}\right)\mathfrak{e}\left(\delta,\Gamma_{\underline{k}_{\text{O}}^{r-1}|\underline{k}_{\text{O}}^{r-1}},\Gamma_{\bar{k}_{\text{O}}^{r-1}|\bar{k}_{\text{O}}^{r-1}}\right),\right.\\
&\left.\enskip \hat{c}+\hat{c}\mathfrak{a}\left(\Gamma_{\underline{k}_{\text{O}}^{r}|\underline{k}_{\text{O}}^{r}}\right)\mathfrak{q}\left(\Gamma_{\underline{k}_{\text{O}}^{r}|\underline{k}_{\text{O}}^{r}}\right)\left(1-\mathfrak{q}\left(\Gamma_{\underline{k}_{\text{O}}^{r}|\underline{k}_{\text{O}}^{r}}\right)\right)^{-1}\right\}.
\end{array} \right.
\end{flalign*}

When $r=1$ and $\underline{k}_{{\text{U}}}^1=0$, the mean estimation error satisfies
\begin{flalign*}
&\left\|\bm{\eta}_{k|k}\right\|\le \left\{
\arraycolsep=1.5pt\def\arraystretch{1}
\begin{array}{rl}
\text{(a) for $k\in (\underline{k}_{\text{U}}^{1},\bar{k}_{\text{U}}^{1}]$: }&\mathfrak{e}_0\left(\Gamma_{0|0}\right),\\
\text{(b) for $k\in (\underline{k}_{\text{O}}^{1},\bar{k}_{\text{O}}^{1}]$: }&\max\left\{\hat{c}+\mathfrak{a}\left(\Gamma_{\underline{k}_{\text{O}}^{1}|\underline{k}_{\text{O}}^{1}}\right)\mathfrak{q}\left(\Gamma_{\underline{k}_{\text{O}}^{1}|\underline{k}_{\text{O}}^{1}}\right)\mathfrak{e}_{0}\left(\Gamma_{0|0}\right),\enskip \right.\\
&\left.\enskip \hat{c}+\hat{c}\mathfrak{a}\left(\Gamma_{\underline{k}_{\text{O}}^{1}|\underline{k}_{\text{O}}^{1}}\right) \mathfrak{q}\left(\Gamma_{\underline{k}_{\text{O}}^{1}|\underline{k}_{\text{O}}^{1}}\right)\left(1-\mathfrak{q}\left(\Gamma_{\underline{k}_{\text{O}}^{1}|\underline{k}_{\text{O}}^{1}}\right)\right)^{-1} \right\}.\\
\end{array} \right.
\end{flalign*}

When $r=1$ and $\underline{k}_{{\text{O}}}^1=0$, the mean estimation error satisfies
\begin{flalign*}
&\left\|\bm{\eta}_{k|k}\right\|\le \left\{
\arraycolsep=1.5pt\def\arraystretch{1}
\begin{array}{rl}
\text{(a) for $k\in (\underline{k}_{\text{U}}^{1},\bar{k}_{\text{U}}^{1}]$: }&\mathfrak{e}\left(\delta,\Gamma_{0|0}, \Gamma_{\underline{k}_{\text{U}}^{1}|\underline{k}_{\text{U}}^{1}}\right),\\
\text{(b) for $k\in (\underline{k}_{\text{O}}^{1},\bar{k}_{\text{O}}^{1}]$: } &\max\left\{\hat{c}+\mathfrak{a}\left(\Gamma_{0|0}\right)\mathfrak{q}\left(\Gamma_{0|0}\right) \sqrt{n}\varrho_{\text{m}},\enskip\right.\\
&\left.\enskip \hat{c}+\hat{c}\mathfrak{a}\left(\Gamma_{0|0}\right)\mathfrak{q}\left(\Gamma_{0|0}\right)\left(1-\mathfrak{q}\left(\Gamma_{0|0}\right)\right)^{-1} \right\}.\\
\end{array} \right.
\end{flalign*}
\end{prop}
\begin{proof}
The proof is done by combining Proposition \ref{prop_u_bound} and Proposition \ref{prop_o_bound}, which is detailed in Appendix \ref{ap:prop_switch}.
\end{proof}

\begin{remark}\label{RM:dwelltime}
The minimum residence time in Proposition \ref{prop:switch} shares a similar concept with the definition of (average) dwell time (e.g. \cite{liberzon2012switching,xie2013exponential}), in the sense that both impose conditions on sufficiently long time spent in modes that are globally asymptotically stable (or observable in our case). However, several main differences between the two exist. For example, there is no condition imposed in this work regarding the ratio between the total time spent in observable and unobservable modes, while the analysis using an average dwell time (e.g., \cite{xie2013exponential}) requires a sufficient large ratio between the total time spent in stable and unstable modes. Moreover, since this work derives switching conditions to ensure bounded estimation error provided by an online filter, the minimum residence times are also computed online, which depend on the estimation error covariances at the beginning of the observable time intervals. This also differs from the stability analysis based on the (average) dwell time where the timing conditions on the switching sequences are computed offline.
\end{remark}

\section{Numerical experiments}\label{sec: NumericalExperiments}

\subsection{Effect of inter-agent communication and filter consistency}\label{sub:Individual_DLKCF}
In this section, we show the critical role the consensus term plays in reducing the disagreement between agents, and validate the consistency of the DLKCF using the NEES measure. The network is a stretch of highway divided into 136 cells and 7 sections.
We apply normalized parameters for the triangular fundamental diagram. The true solution is set to be a combination of an expansion fan and a shock propagating upstream, with a sinusoidal upstream boundary condition (Figure \ref{Fig:DLKCF_DLKCF_network}a), which is computed based on the CTM. Parameter values and elements of the experimental setup not detailed here can be found in the README documentation for the supplementary source code https://github.com/yesun/DLKCF.
\begin{figure}
\begin{center}
\includegraphics[width=240pt]{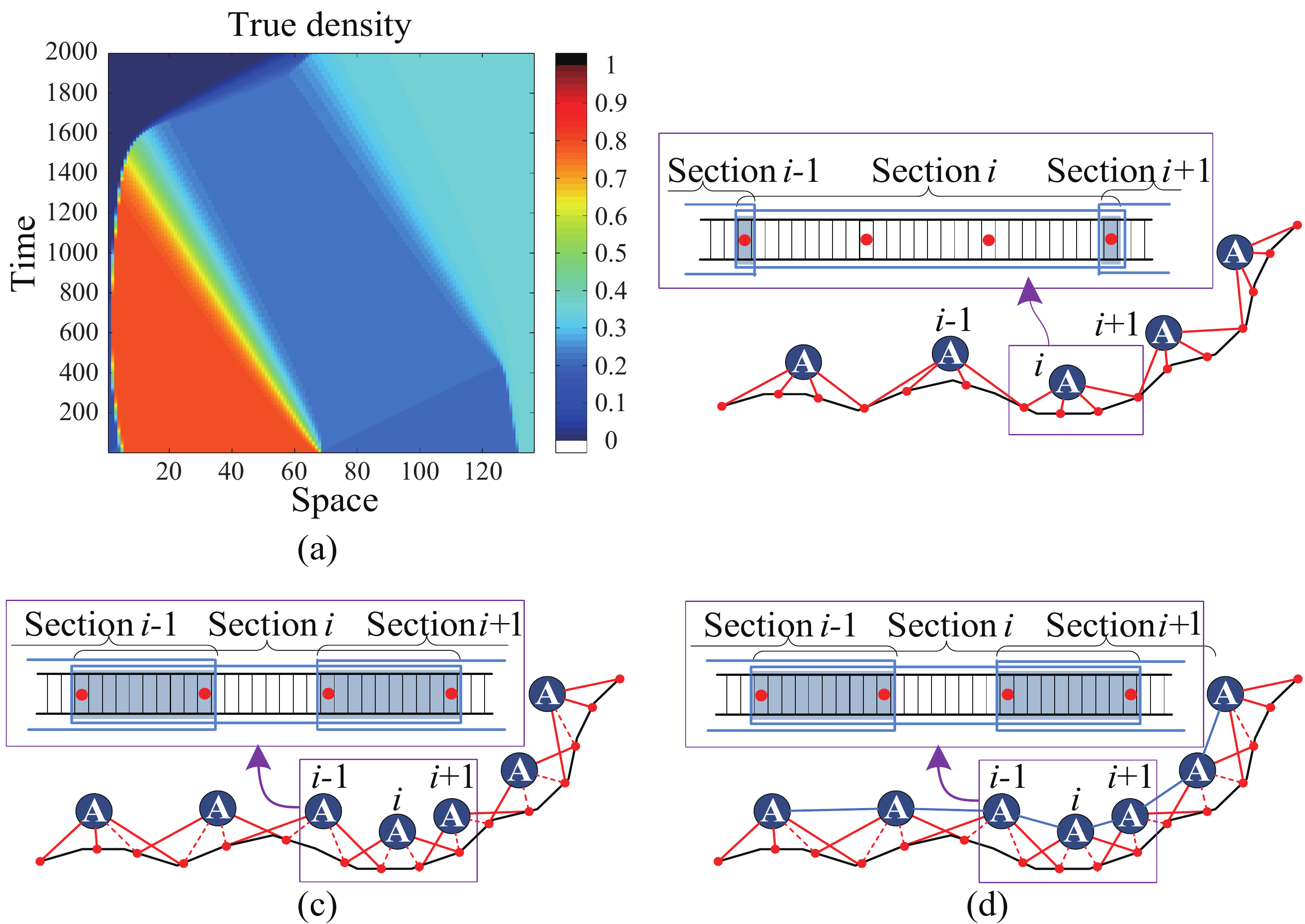}
\caption{(a) True Solution; (b-d) Freeway network setup and communication topology for: (b) the LKF; (c) the DLKCF-0; and (d) the DLKCF. The red solid lines represent the direct connection between agents (labeled A in circles) and sensors (red dots), and the red dashed lines represent connection between agents and sensors obtained through receiving shared measurements and sensor models from neighbors. The blue lines stand for the existence of consensus terms between agents. In the zoomed-in parts, the freeway is discretized by cells (small rectangles) and localized by sections (blocks). Overlapping regions are represented by the blue shaded cells, and sensor locations are represented by red dots in the cells.}\label{Fig:DLKCF_DLKCF_network}
\end{center}
\vspace{-0.2in}
\end{figure}

Disagreement and error on state estimates can be generated for various reasons, here we consider the combining effects of the following two causes: (\textit{i}) \textit{heterogeneous sensors} (HS), with some of the sensors having large measurement errors; (\textit{ii}) \textit{inconsistent agents} (IA), with some agents assuming incorrect (too small) noise models for the low quality sensors. In this experiment, we put a large-error sensor (with the measurement error standard deviation of 0.3, compared to 0.03 for all other sensors) once every three sensors starting from the downstream sensor of the first section. Moreover, agents associated with sections indexed by even numbers are unable to recognize the large-error sensors they are directly connected to (thus still applying 0.03 as the measurement error standard deviation for these sensors). We also apply perturbations of 10-20\% on the model parameters (i.e., $\varrho_{\text{m}}$, $\varrho_{\text{c}}$, and $v_{\text{m}}$) on different sections.

We explore the effects of the above two causes on the disagreement and error of estimates for (\textit{i}) the \textit{local KF} (LKF), where each local agent runs the KF described in Section \ref{subsec:KF} independently based on measurements from the sensors it is directly connected to (e.g., $z_{i,k}^i$ for agent $i$), without sharing measurements or estimates; (\textit{ii}) the \textit{DLKCF with zero consensus gain} (DLKCF-0), where the prediction and correction steps are given by \eqref{eq:DLKCFforecast}-\eqref{eq:DLKCFAnalysis} (i.e., neighboring agents share sensor data and sensor models) with consensus gains set to zero (i.e., $C_{i,k}^j=\bm{0}_{n_i,n_{i,j}}$ for all $i$, $j\in\mathcal{N}_i$ and $k$); and (\textit{iii}) the DLKCF with consensus gain as given in \eqref{eq:DLKCFconsensusgain} (where $\gamma^j_{i,k}=0.99\min\{\gamma^*_{i,k},\gamma^*_{j,k},\hat{\gamma}_{i,k}^j\}$ with $\hat{c}=0.01$). Figure \ref{Fig:DLKCF_DLKCF_network} shows the network setup and the communication topology for the LKF (Figure \ref{Fig:DLKCF_DLKCF_network}b), the DLKCF-0 (Figure \ref{Fig:DLKCF_DLKCF_network}c) and DLKCF (Figure \ref{Fig:DLKCF_DLKCF_network}d). At time $k$, the average disagreement $\tilde{u}_k$ of the posterior estimate is computed by $\tilde{u}_k=\frac{1}{N-1}\sum_{i=1}^{N-1}\frac{\parallel \tilde{u}^{i+1}_{i,k}\parallel_2^2}{n_{i,i+1}}$ with $\tilde{u}^{i+1}_{i,k}=\hat{I}_{j,i}\eta_{j,k|k}-\hat{I}_{i,j}\eta_{i,k|k}$, and the average estimation error is given by $\eta_k=\frac{1}{N}\sum_{i=1}^{N}\frac{\parallel \eta_{i,k|k}\parallel_2^2}{n_i}$.

\begin{table}
\centering
\caption{Disagreement and error of estimate$^{1}$}
\label{tb:LOD_check}
\vspace{-0.1in}
\begin{IEEEeqnarraybox}[\IEEEeqnarraystrutmode\IEEEeqnarraystrutsizeadd{1pt}{1pt}]{c/c/c/c/c/c/c/c}
\hline
\IEEEeqnarraymulticol{2}{t}{Causes}&\IEEEeqnarraymulticol{3}{t}{Disagreement $\tilde{u}$ ($\times 10^{-2}$)}&\IEEEeqnarraymulticol{3}{t}{Error $\eta$ ($\times 10^{-2}$)}\\
\textrm{HS}&\textrm{IA}&\textrm{LKF}&\textrm{DLKCF-0}&\textrm{DLKCF}&\textrm{LKF}&\textrm{DLKCF-0}&\textrm{DLKCF}\\
\hline
\textrm{False}&--&--&0.294&0.119&0.423&0.349&0.308\\
\textrm{True}&\textrm{False}&--&0.336&0.119&0.562&0.503&0.468\\
\textrm{True}&\textrm{True}&--&7.361&4.664&2.941&2.670&2.633\\
\hline
\end{IEEEeqnarraybox}
\vspace{-0.25in}
\end{table}

Table \ref{tb:LOD_check} reports the disagreement and estimation error of the three filters, where $\tilde{u}=\sum_{k=1}^{k_{\max}}\tilde{u}_k$ and $\eta=\sum_{k=1}^{k_{\max}}\eta_k$ with $k_{\max}$ denoting the total number of time steps. Since the neighboring sections in the LKF have no overlapping cells except the shared boundary cells with sensor measurements, the neighbor disagreement for the LKF is not considered. It is shown that the estimation accuracy of the LKF is vulnerable to inconsistent error models, since the inconsistent agents can never identify the high-error sensors they are connected to, while in the DLKCF-0 and DLKCF some of the inconsistent agents apply the correct measurement error covariance matrices when they share sensor data and sensor models with neighbors. Moreover, compared to the DLKCF-0, adding the consensus term in the DLKCF considerably reduces the neighbor disagreement (regardless of the existence of heterogeneous sensors or inconsistent agents). Hence, the DLKCF outperforms the other two filters with respect to agreement and accuracy on estimates (further illustrations of the performance of the filters is given in \cite{DLKCFthesis}).

As stated in Remark \ref{RM:Optimality}, we remove the existence of inconsistent agents and perform an NEES check \cite{BarShalom2001} of the DLKCF across 50 Monte Carlo runs, thus accessing the validity of dropping the cross-correlations among different agents in the estimation error covariance. The average percentage of time steps across all the sections that the NEES measure surpasses the two-sided 95\% probability concentration region $[1.484, 2.6]$ is 1.98\%. Among all the sections, the maximum (resp. minimum) percentage of time steps that the NEES is greater than the upper limit (resp. smaller than the lower limit) is 2.45\% (resp. 1.8\%). This indicates that the filter-calculated error covariance matches the mean square error of the DLKCF.

\subsection{Computational complexity}\label{sub:complexity}
\begin{table}
\centering
\caption{Runtime comparison of the central KF and DLKCF (per agent)}
\label{Runtime}
\begin{IEEEeqnarraybox}[\IEEEeqnarraystrutmode\IEEEeqnarraystrutsizeadd{1pt}{1pt}]{c/c/c/c/c/c/c/c}
\hline
\IEEEeqnarraymulticol{2}{t}{Central KF}&\IEEEeqnarraymulticol{1}{t}{$\quad$}&\IEEEeqnarraymulticol{5}{t}{DLKCF}\\
n&\textrm{runtime $t_{\text{c}}$ (sec)}&\quad&n&n_l&\hat{n}&N&\textrm{runtime $t_{\text{d}}$ (sec)}\\
\hline
100&104&\quad&100&28&10&5&6.2\\
210&512&\quad&210&50&10&5&24.6\\
210&512&\quad&210&58&20&5&42.2\\
\hline
\end{IEEEeqnarraybox}
\vspace{-0.25in}
\end{table}

For simplicity, let $n_i=n_l$ for all $i$, and denote as $\hat{n}$ the uniform size of the overlapping regions. The computational complexity of the DLKCF for the $i^{\textrm{th}}$ local agent is dominated by $O(n_l^3+(|\mathcal{N}_i|\hat{n})^3)$ at each time step, where $|\mathcal{N}_i|$ is the number of neighbors of agent $i$. This implies that we need $\hat{n}<n_l|\mathcal{N}_i|^{-1}$ to have a consensus term with computational complexity less than the local KF. Table \ref{Runtime} reports the runtime per agent of the DLKCF and the central KF to complete 2000 estimation steps tracking a shockwave on a stretch of freeway, which we denote as $t_{\text{d}}$ and $t_{\text{c}}$, respectively. It is evident that compared to the central KF, the runtime of the DLKCF is considerably reduced. Moreover, given a fixed network dimension and a fixed number of agents, the computation load increases with the size of the overlapping regions.

\section{Conclusion and future work}\label{sec:Conclusion}
In this article a distributed local Kalman consensus filter is designed for large-scale multi-agent traffic estimation. The DLKCF is applied to the SMM to monitor traffic on a road network partitioned into local sections, with overlapping regions between neighbors introduced to allow for information exchange on measurements and estimates. We prove that the mean error dynamics of the DLKCF is GAS when all sections switch among observable modes of the SMM. For an unobservable section, we show that the mean estimates are ultimately bounded inside the physically meaningful interval. We also prove that the 2-norm of the mean error for any given section is upper bounded under switches among observable and unobservable modes, as long as the section remains observable for a minimum residence time after switching to an observable mode from an unobservable one. Numerical experiments illustrate the effect of the DLKCF on reducing the overall estimation error, compared to the LKF, as well as promoting agreement among different agents. The numerical results also show a considerable reduction on the runtime of the DLKCF compared to a central KF.

In order to apply the DLKCF in the field, extension of the observability results to freeway networks with merging and diverging junctions is necessary but straightforward. Moreover, the development and incorporation of robustness results on the detection of sensing/computing outliers and model mismatches in each local agent can further improve estimation accuracy.

\section*{Acknowledgement}
The authors thank Prof. C. Canudas de Wit and the reviewers for constructive suggestions that improved this manuscript. This material is based upon work supported by the National Science Foundation under Grant No. CMMI-1351717.

\newpage
\appendix
\renewcommand{\theequation}{A.\arabic{equation}}

\subsection{Destabilizing Effect of the Cross-correlation Terms under Non-observability}\label{ap:crosscov_unstb}
Let $Q^j_{i,k}=\mathbb{E}[w_{i,k}w^T_{j,k}]$ and $R^j_{i,k}=\mathbb{E}[v_{i,k}v^T_{j,k}]$ be the cross-covariance of the model and measurement noise between agent $i$ and $j$, and denote as $\Gamma^j_{i,k|k}=\mathbb{E}[\eta_{i,k|k}\eta^T_{j,k|k}|\mathcal{Z}_k]$ and $\Gamma^j_{i,k|k-1}=\mathbb{E}[\eta_{i,k|k-1}\eta^T_{j,k|k-1}|\mathcal{Z}_{k-1}]$. The optimal DLKCF is given by
\begin{flalign}
&\textrm{Prediction:}\notag\\
&\left\{ \begin{array}{lc}
\rho_{i, k|k-1}=A_{i,k-1}\rho_{i,k-1|k-1}\\
\Gamma^j_{i,k|k-1}=A_{i,k-1}\Gamma^j_{i,k-1|k-1}A_{j,k-1}^{T}+Q^j_{i,k-1},
\end{array}\right.&\label{eq:opDLKCFforecast}
\end{flalign}
\begin{flalign}
&\textrm{Correction:}\notag\\
&\left\{ \begin{array}{lc}
\rho_{i,k|k}=\rho_{i,k|k-1}+K_{i,k}\left(z_{i,k}-H_{i,k}\rho_{i,k|k-1}\right)+\sum_{j\in \mathcal{N}_i}C^{j}_{i,k}\left(\hat{I}_{j,i}\rho_{j,k|k-1}-\hat{I}_{i,j}\rho_{i,k|k-1}\right)\\
\Gamma^{j}_{i,k|k}=F_{i,k}\Gamma^{j}_{i,k|k-1}F^T_{j,k}+K_{i,k}R^j_{i,k}K^T_{j,k}+F_{i,k}\sum_{q\in\mathcal{N}_j}\left(\Gamma^{q}_{i,k|k-1}\hat{I}^T_{q,j}-\Gamma^j_{i,k|k-1}\hat{I}^T_{j,q}\right)\left(C^q_{j,k}\right)^T\\
\quad\quad\quad\quad+\sum_{p\in\mathcal{N}_i}C^p_{i,k}\left(\hat{I}_{p,i}\Gamma^{j}_{p,k|k-1}-\hat{I}_{i,p}\Gamma^{j}_{i,k|k-1}\right)F^T_{j,k}+\Omega^j_{i,k}\\
K_{i,k}=\left(\Gamma_{i,k|k-1}+C^p_{i,k}\sum_{p\in\mathcal{N}_i}\left(\hat{I}_{p,i}\Gamma^{i}_{p,k|k-1}-\hat{I}_{i,p}\Gamma_{i,k|k-1}\right)\right)H_{i,k}^T\left(R_{i,k}+H_{i,k}\Gamma_{i,k|k-1}H_{i,k}^T\right)^{-1},\\
\end{array}\right.&\label{eq:opDLKCFAnalysis}
\end{flalign}
where $F_{i,k}=I-K_{i,k}H_{i,k}$ and
\begin{align}
\Omega^j_{i,k}=\sum_{p\in\mathcal{N}_i,q\in\mathcal{N}_j}C^p_{i,k}\left(\hat{I}_{p,i}\Gamma^{q}_{p,k|k-1}\hat{I_{q,j}^T}-\hat{I}_{p,i}\Gamma^{j}_{p,k|k-1}\hat{I}_{j,q}^T-\hat{I}_{i,p}\Gamma^{q}_{i,k|k-1}\hat{I}_{q,j}^T+\hat{I}_{i,p}\Gamma^{j}_{i,k|k-1}\hat{I}_{j,q}^T\right)\left(C^q_{j,k}\right)^T.\notag
\end{align}
The next example shows that the estimation error covariance of a section which is always observable can end up being unbounded when adding the cross-correlation terms.
\begin{example}
Consider a freeway segment partitioned into two sections with same dimension $n$. It is assumed that $q_1I<Q_{i,k}<q_2I$, $r_1I<R_{i,k}<r_2I$, and $\left\|Q^j_{i,k}\right\|<q_2$ for all $k>0$ and $i,j\in\{1,2\}$. Suppose section 2 always stays in the CC mode, i.e., $\sigma_2(k)=\text{CC}$ for all $k>0$. However, section 1 is unobservable before switching to the CC mode at time step $k_0$, i.e., $\sigma_1(k)\in\{\text{FC1, FC2}\}$ for $k\in[0,k_0)$ and $\sigma_1(k)=\text{CC}$ for $k\ge k_0$. For all $\bar{\varpi}>0$, there exists $k_0>0$ such that $\left\|\Gamma_{2,k_0|k_0}\right\|>\bar{\varpi}$, independent of the initial error covariance $\Gamma_{2,0|0}$.
\end{example}
\begin{proof}
As provided in \eqref{eq:opDLKCFAnalysis}, when section 1 switches to the CC mode at time step $k_0$, the error covariance update of section 2 is given by
\begin{equation}\label{eq:errCov}
\begin{split}
\Gamma_{2,k_0|k_0}=&F_{2,k_0}\Gamma_{2,k_0|k_0-1}F^T_{2,k_0}+K_{2,k_0}R_{2,k_0}K^T_{2,k_0}+F_{2,k_0}\left(\Gamma^1_{2,k_0|k_0-1}\hat{I}^T_{1,2}-\Gamma_{2,k_0|k_0-1}\hat{I}^T_{2,1}\right)\left(C^1_{2,k_0}\right)^T\\
&+C^1_{2,k_0}\left(\hat{I}_{1,2}\Gamma^{2}_{1,k_0|k_0-1}-\hat{I}_{2,1}\Gamma_{2,k_0|k_0-1}\right)F^T_{2,k_0}+\Omega^2_{2,k_0},\\
\end{split}
\end{equation}
where
\begin{equation}
\begin{split}\label{eq:f_k0}
K_{2,k_0}&=\left(\Gamma_{2,k_0|k_0-1}+C^1_{2,k_0}\left(\hat{I}_{1,2}\Gamma^{2}_{1,k_0|k_0-1}-\hat{I}_{2,1}\Gamma_{2,k_0|k_0-1}\right)\right)H_{2,k_0}^T\left(R_{2,k_0}+H_{2,k_0}\Gamma_{2,k_0|k_0-1}H_{2,k_0}^T\right)^{-1},\\
F_{2,k_0}&=I-K_{2,k_0}H_{2,k_0},
\end{split}
\end{equation}
and
\begin{equation}
\begin{split}\label{eq:omega}
\Omega^2_{2,k_0}=&C^1_{2,k_0}\left(\hat{I}_{1,2}\Gamma_{1,k_0|k_0-1}\hat{I}_{1,2}^T\right)\left(C^1_{2,k_0}\right)^T\\
&-C^1_{2,k_0}\left(\hat{I}_{1,2}\Gamma^{2}_{1,k_0|k_0-1}\hat{I}_{2,1}^T+\hat{I}_{2,1}\Gamma^{1}_{2,k_0|k_0-1}\hat{I}_{1,2}^T-\hat{I}_{2,1}\Gamma_{2,k_0|k_0-1}\hat{I}_{2,1}^T\right)\left(C^1_{2,k_0}\right)^T.
\end{split}
\end{equation}
\noindent \textbf{Step 1}: In this step we show that $\Gamma_{2,k_0|k_0-1}$ and $\Gamma^1_{2,k_0|k_0-1}=\left(\Gamma^2_{1,k_0|k_0-1}\right)^T$ are bounded for all $k_0>0$.

As shown in Lemma \ref{prop:GammaBound}, $\Gamma_{2,k|k}$ is bounded for all $k<k_0$, i.e., there exists $\nu_{\Gamma}>0$ such that $\left\|\Gamma_{2,k|k}\right\|<\nu_{\Gamma}$ for all $k<k_0$, and this $\nu_{\Gamma}$ is independent of $k_0$. This indicates that no matter when section 1 switches to the CC mode, the error covariance of section 2 always satisfies $\left\|\Gamma_{2,k|k}\right\|<\nu_{\Gamma}$ before section 1 becomes observable. The above explanation applies to all the scenarios in this proof when we state that a matrix is bounded for all $k<k_0$. The prior error covariance of section 2 at time $k_0$ is given by
\begin{align*}
\Gamma_{2,k_0|k_0-1}=A_{2,k_0-1}\Gamma_{2,k_0-1|k_0-1}A_{2,k_0-1}^{T}+Q_{2,k_0-1},
\end{align*}
where $A_{2,k_0-1}=A_{\text{CC}}$. Hence, the prior error covariance $\Gamma_{2,k_0|k_0-1}$ is bounded for all $k_0$ due to the assumption that $\left\|Q_{2,k_0-1}\right\|<q_2$.

When $k<k_0$, the consensus gain $C_{2,k}^1=\bm{0}$ since section 1 is not observable and the consensus term is turned off. For simplicity let $R_{2,k}^1=\bm{0}$ for $k<k_0$, which means that sections 1 and 2 do not share sensor data when $k<k_0$, thus fusing measurements obtained from non-intersecting sensor sets, and the noise models for different sensors are independent. In this case, the evolution of the cross-covariance $\Gamma^1_{2,k|k-1}$ is given by
\begin{equation}\label{eq:gamma12}
\Gamma^{1}_{2,k|k}=F_{2,k}\Gamma^{1}_{2,k|k-1}F^T_{1,k}=F_{2,k}A_{2,k-1}\Gamma^{1}_{2,k-1|k-1}A_{1,k-1}^T F^T_{1,k}+F_{2,k}Q^{1}_{2,k-1}F^T_{1,k}, \quad \text{for $k<k_0$,}
\end{equation}
where
\begin{equation}\label{eq:fk12}
F_{i,k}=I-K_{i,k}H_{i,k},\quad \text{and }K_{i,k}=\Gamma_{i,k|k-1}H_{i,k}^T\left(R_{i,k}+H_{i,k}\Gamma_{i,k|k-1}H_{i,k}^T\right)^{-1}\quad \text{for $i\in\{1,2\}$.}
\end{equation}

As shown in Proposition \ref{prop_o_bound} (cf. \eqref{eq:norm_exp}), when section 2 is observable and the consensus term is set to zero we have
\begin{equation}\label{eq:ftozero}
\left\|\prod_{\kappa=k}^1 F_{2,\kappa}A_{2,\kappa-1}\right\|\propto\tilde{q}^k, \quad \text{where $\tilde{q}<1$}.
\end{equation}
For section 1, the error dynamics when $k<k_0$ is given by
\begin{equation*}
\bm{\eta}_{1,k|k}=\left(\prod_{\kappa=k}^1 F_{1,\kappa}A_{1,\kappa-1}\right)\bm{\eta}_{1,0|0}.
\end{equation*}
As shown in Proposition \ref{prop_u_bound}, the estimation error $\bm{\eta}_{1,k|k}$ is bounded for $k<k_0$, hence $\left\|\prod_{\kappa=1}^k A_{1,\kappa-1}^T F_{1,\kappa}^T\right\|=\left\|\prod_{\kappa=k}^1 F_{1,\kappa}A_{1,\kappa-1}\right\|$ is bounded for all $k<k_0$ (provided that $\bm{\eta}_{1,k|k}\neq \bm{0}$). Moreover, since $\Gamma_{2,k|k}$ is bounded for all $k<k_0$, the prior error covariance $\Gamma_{2,k|k-1}$ is also bounded for $k<k_0$. Hence, it is concluded based on \eqref{eq:fk12} that $K_{2,k}$ and thus $F_{2,k}$ are bounded for all $k<k_0$. As for the unobservable section, it is shown in Lemma \ref{lem_u_bound_k} that the Kalman gain $K_{1,k}$ is bounded for $k<k_0$, thus $F_{1,k}$ is also bounded for $k<k_0$ according to \eqref{eq:fk12}. Consequently, the second term in \eqref{eq:gamma12} is bounded for $k<k_0$, combining this with \eqref{eq:ftozero} and the boundedness of $\left\|\prod_{\kappa=1}^k A_{2,\kappa-1}^T F_{2,\kappa}^T\right\|$, it is concluded that $\Gamma^{1}_{2,k|k}$ is bounded for all $k<k_0$. The cross-covariance $\Gamma^1_{2,k_0|k_0-1}$ is expressed as
\begin{align*}
\Gamma^1_{2,k_0|k_0-1}=A_{2,k_0-1}\Gamma^1_{2,k_0-1|k_0-1}A_{1,k_0-1}^{T}+Q^1_{2,k_0-1},
\end{align*}
which yields the boundedness of $\Gamma^1_{2,k_0|k_0-1}$ for all $k_0$.

\noindent \textbf{Step 2}: In this step we show that the first four terms in \eqref{eq:errCov} is bounded for all $k_0>0$.

As stated in equation (12) of \cite{SunWorkCONES2016}, the consensus gain is defined as $C_{2,k_0}^1=\gamma_{2,k_0}^1\Gamma_{2,k_0|k_0-1}\hat{I}^T_{2,1}$. According to Proposition \ref{Prop:ObservableGUAS}, the scaling factor $\gamma_{2,k_0}^1$ is bounded for all $k_0>0$, thus $C_{2,k_0}^1$ is bounded for all $k_0>0$. Based on \eqref{eq:f_k0}, we conclude the boundedness of $K_{2,k_0}$ for all $k_0>0$ given the boundedness of $\Gamma_{2,k_0|k_0-1}$, $\Gamma^1_{2,k_0|k_0-1}=\left(\Gamma^2_{1,k_0|k_0-1}\right)^T$ and $C_{2,k_0}^1$. As a consequence, $F_{2,k_0}$ is also bounded for all $k_0>0$ due to its relationship with $K_{2,k_0}$ shown in \eqref{eq:f_k0}.

Now we have illustrated that $\Gamma_{2,k_0|k_0-1}$, $\Gamma^1_{2,k_0|k_0-1}=\left(\Gamma^2_{1,k_0|k_0-1}\right)^T$, $C_{2,k_0}^1$, $K_{2,k_0}$ and $F_{2,k_0}$ are bounded for all $k_0>0$. Combining this with \eqref{eq:errCov}, it follows that the first four terms in \eqref{eq:errCov} is bounded for all $k_0>0$. Specifically, there exists $\nu_{1}>0$ such that
\begin{equation*}
\begin{split}
&\left\|F_{2,k_0}\Gamma_{2,k_0|k_0-1}F^T_{2,k_0}+K_{2,k_0}R_{2,k_0}K^T_{2,k_0}+F_{2,k_0}\left(\Gamma^1_{2,k_0|k_0-1}\hat{I}^T_{1,2}-\Gamma_{2,k_0|k_0-1}\hat{I}^T_{2,1}\right)\left(C^1_{2,k_0}\right)^T\right.\\
&\quad \left.+C^1_{2,k_0}\left(\hat{I}_{2,1}\Gamma^{2}_{1,k_0|k_0-1}-\hat{I}_{2,1}\Gamma_{2,k_0|k_0-1}\right)F^T_{2,k_0}\right\|<\nu_{1}, \quad \text{for all $k_0>0$.}
\end{split}
\end{equation*}

\noindent \textbf{Step 3}: In this step we show that the last term in \eqref{eq:errCov} can be unbounded as $k_0\rightarrow \infty$

First notice that based on Step 1 and Step 2, the second term in \eqref{eq:omega} is bounded for all $k_0>0$. Specifically, there exists $\nu_{2}>0$ such that
\begin{equation*}
\begin{split}
\left\|C^1_{2,k_0}\left(\hat{I}_{1,2}\Gamma^{2}_{1,k_0|k_0-1}\hat{I}_{2,1}^T+\hat{I}_{2,1}\Gamma^{1}_{2,k_0|k_0-1}\hat{I}_{1,2}^T-\hat{I}_{2,1}\Gamma_{2,k_0|k_0-1}\hat{I}_{2,1}^T\right)\left(C^1_{2,k_0}\right)^T\right\|<\nu_{2},\quad \text{for all $k_0>0$.}
\end{split}
\end{equation*}
Since section 1 is unobservable for $k\in[0,k_0)$, the error covariance $\left\|\Gamma_{1,k_0|k_0}\right\|\rightarrow \infty$ as $k_0\rightarrow \infty$. Hence, for all $\varpi_{\Gamma}>0$ there exists $k_0>0$ and $l\in\{2,\cdots,n-1\}$ such that the $l^{\text{th}}$ diagonal entry of $\Gamma_{1,k_0|k_0}$ is greater than $\varpi_{\Gamma}$. Let $n-n_{1,2}<l\le n-1$, it follows that $\left\|\hat{I}_{1,2}\Gamma_{1,k_0|k_0-1}\hat{I}_{1,2}^T\right\|>\varpi_{\Gamma}$. The consensus gain of section 2 satisfies
\begin{equation*}
C_{2,k_0}^1=\gamma_{2,k_0}^1\Gamma_{2,k_0|k_0-1}\hat{I}^T_{2,1},\quad \text{where}\quad \Gamma_{2,k_0|k_0-1}>Q_{k_0-1}>q_1 I,
\end{equation*}
it follows that
\begin{align*}
\left\|C^1_{2,k_0}\left(\hat{I}_{1,2}\Gamma_{1,k_0|k_0-1}\hat{I}_{1,2}^T\right)\left(C^1_{2,k_0}\right)^T\right\|>O(\varpi_{\Gamma}).
\end{align*}
Consequently, the error covariance $\Gamma_{2,k_0|k_0}$ satisfies
\begin{align*}
\left\|\Gamma_{2,k_0|k_0}\right\|>O(\varpi_{\Gamma})-\nu_1-\nu_2.
\end{align*}
Since $\varpi_{\Gamma}$ can be any positive value, it follows that for all $\bar{\varpi}>0$, there exists $k_0>0$ such that $\left\|\Gamma_{2,k_0|k_0}\right\|>\bar{\varpi}$, which is independent of the initial error covariance $\Gamma_{2,0|0}$. This completes the proof.
\end{proof}

\newpage
For the remainder of the Appendix, the section index $i$ is dropped for notational simplicity.

\subsection{Proof of Lemma \ref{lem:GammaBound}}\label{ap:proof_lemma_GammaBound}

The proof can be done by showing the uniform complete observability and controllability of the filter under switches among the observable modes. When a freeway section switches among the observable modes of the SMM, the information matrix for time interval $k\in [k_0,k_1]$ is defined as
\begin{align*}
\mathcal{I}_{k_1,k_0}=\sum_{k=k_0}^{k_1}\Xi^T_{k,k_1}H_{k}^TR_{k}^{-1}H_{k}\Xi_{k,k_1},
\end{align*}
where
\begin{align}\label{eq:xi}
\Xi_{k,k_1}=\prod_{\kappa=k}^{k_1-1}A^{-1}_{\kappa},\quad \text{and}\quad \Xi_{k_1,k}=\Xi^{-1}_{k,k_1}=\prod_{\kappa=k_1-1}^{k}A_{\kappa},
\end{align}
with
\begin{align}\notag
A_{\kappa}\in\mathcal{A}_{\text{O}}, \quad \text{for $k_0\le \kappa<k_1$.}
\end{align}
The controllability matrix for time interval $k\in [k_0,k_1]$ is defined as
\begin{align*}
\mathcal{C}_{k_1,k_0}=\sum_{k=k_0}^{k_1-1}\Xi_{k_1,k+1}Q_{k+1}\Xi_{k_1,k+1}^T.
\end{align*}

\noindent \textbf{Step 1}: Deriving the uniform complete observability of the filter. In order to show the uniform complete observability, we first need to find a finite integer $\bar{k}>0$ such that $\text{rank}\left(\mathcal{I}_{k_1,k_0}\right)=n$ for all $k_1-k_0=\bar{k}$.
The observability grammian matrix is defined as
\begin{align*}
\mathcal{G}_{k_1,k_0}=\sum_{k=k_0}^{k_1}\Xi^T_{k,k_0}H_k^TR_k^{-1}H_k\Xi_{k,k_0}=\Xi_{k_1,k_0}^T\mathcal{I}_{k_1,k_0}\Xi_{k_1,k_0}.
\end{align*}
Hence, the observability matrix and the information matrix have the same rank, and it is sufficient to find a $\bar{k}$ such that $\text{rank}\left(\mathcal{G}_{k_1,k_0}\right)=n$ for all $k_1-k_0=\bar{k}$.
Since at least the traffic densities at the boundary cells are measured, the output matrix has the following formula:
\begin{equation}\notag
\begin{split}
H_k=&\left(
  \begin{array}{c}
    H_{\text{b}}\\
    \bar{H}_{k}
  \end{array}
\right),
\end{split}
\end{equation}
where $H_{\text{b}}$ is defined in \eqref{eq:h_b}. Hence the observability grammian matrix satisfies
\begin{equation}\label{eq:g_gtilde}
\mathcal{G}_{k_1,k_0}=\sum_{k=k_0}^{k_1}\Xi^T_{k,k_0}H_k^TR_k^{-1}H_k\Xi_{k,k_0}>r_2^{-1}\sum_{k=k_0}^{k_1}\Xi^T_{k,k_0}H_k^TH_k\Xi_{k,k_0}\ge r_2^{-1} \sum_{k=k_0}^{k_1}\Xi^T_{k,k_0}H_{\text{b}}^TH_{\text{b}}\Xi_{k,k_0}=r_2^{-1}\tilde{\mathcal{G}}_{k_1,k_0}^T\tilde{\mathcal{G}}_{k_1,k_0},
\end{equation}
where
\begin{equation}\label{eq:tilde_mathcal_G}
\begin{split}
\tilde{\mathcal{G}}_{k_1,k_0}=&\left(
  \begin{array}{c}
    H_{\text{b}}\\
    H_{\text{b}}A_{k_0}\\
    H_{\text{b}}A_{k_0+1}A_{k_0}\\
    \vdots\\
    H_{\text{b}}\prod_{k=k_1-1}^{k_0}A_{k}
  \end{array}
\right).
\end{split}
\end{equation}
It can be shown after some basic linear algebra that
\begin{align}
H_{\text{b}}\prod_{k=k_0+\iota-1}^{k_0}A_{k}=\left(
  \begin{array}{ccccccc}
    \tilde{g}^{\iota}_{1,1}&\cdots& \tilde{g}^{\iota}_{1,\iota_1}&0&0&\cdots&0\\
    0&\cdots&0& 0&\tilde{g}^{\iota}_{2,n-\iota_2+1}&\cdots& \tilde{g}^{\iota}_{2,n}
  \end{array}
\right)\in\mathbb{R}^{2\times n},\label{eq:tilde_mathcal_G_row}
\end{align}
where $\iota_1\le \iota+1$, $\iota_2\le \iota+1$ and $\iota_{1}+\iota_{2} \geq \iota+2$, the elements $\tilde{g}^{\iota}_{\cdot,\cdot}$ are functions of $A_k\in \mathcal{A}_{\text{O}}$ for $k \in [k_0,k_0+\iota)$. Recall that $M(r,c)$ is the $(r,c)$-th entry of matrix $M$. Since
\begin{align*}
\sum_{c=1}^{n} M(r,c) \in \left\{1, \frac{v_{\text{m}} \Delta t}{\Delta x}, \frac{w\Delta t}{\Delta x}\right\} \quad\textrm{for all $M\in \mathcal{A}_{\text{O}}$,}
\end{align*}
we have
\begin{align}\label{eq:tilde_g_lb}
\underline{\theta}^{\iota}\le\tilde{g}^{\iota}_{1,\kappa_1}\le \bar{\theta},\quad \underline{\theta}^{\iota} \le \tilde{g}_{2,n-\kappa_2+1}\le \bar{\theta} \quad\textrm{for $\kappa_1\in\{1,\cdots,\iota_1\}$ and $\kappa_2\in\{1,\cdots,\iota_2\}$,}
\end{align}
where $\underline{\theta}=\min\{v_{\text{m}}\frac{\Delta t}{\Delta x}, 1-v_{\text{m}}\frac{\Delta t}{\Delta x}, w\frac{\Delta t}{\Delta x}, 1-w\frac{\Delta t}{\Delta x}\}$ and $\bar{\theta}=\max_{r,c\in\{1,\cdots,n\}}\left\{M(r,c)\left|M\in\mathcal{A}_{\text{O}} \right.\right\}=1$.
Consequently, the rank of the observability grammian matrix and the rank of the information matrix satisfy
\begin{align}
\textrm{rank}\left(\tilde{\mathcal{G}}_{k_1,k_0}\right)=n=\textrm{rank}\left(\mathcal{G}_{k_1,k_0}\right)=\textrm{rank}\left(\mathcal{I}_{k_1,k_0}\right), \quad \textrm{when $k_1-k_0\geq T_1=\max\left\{1, n-2\right\} $.}\notag
\end{align}
Hence, it can be concluded that $\mathcal{I}_{k,k-T_1}>\bm{0}$ for all $k\ge T_1$. Consequently,
\begin{align}\label{eq:infor_bound_a}
\mathcal{I}_{k,k-T_1}=\sum_{\iota=k-T_1}^{k}\Xi^T_{\iota,k}H_{\iota}^TR_{\iota}^{-1}H_{\iota}\Xi_{\iota,k}> r_2^{-1}\sum_{\iota=k-T_1}^{k}\Xi^T_{\iota,k}H_{\text{b}}^TH_{\text{b}}\Xi_{\iota,k}\ge a_{\mathcal{I}}I>\bm{0} ,\quad \text{for $k \ge T_1$,}
\end{align}
where $a_{\mathcal{I}}>0$ is defined in \eqref{eq:a_b_I}.

As stated in Section \ref{subsec:DLKCF} of \cite{SunWorkCONES2016}, the sensors are spatially distributed and measure the traffic densities at their locations. Hence, the $p^{\text{th}}$ row of the output matrix $H_k\in\mathbb{R}^{m\times n}$ is given by:
\begin{equation*}
H_k(p,:)=\left(0,\cdots,0,1,0,\cdots,0\right),
\end{equation*}
where the location of the entry 1 is the same as the location where the $p^{\text{th}}$ sensor is placed, i.e., if the $p^{\text{th}}$ sensor is located at the $l^{\text{th}}$ cell of the section, than the $l^{\text{th}}$ column of $H_k(r,:)$ is 1. Consequently, the output matrix satisfies
\begin{equation}\label{eq:H_k_ub}
H_k^TH_k<I^TI=I,
\end{equation}
where $H_k=I$ corresponds to the scenario when every cell of the section is measured. Consequently, the information matrix satisfies
\begin{align}\label{eq:infor_bound_b}
\mathcal{I}_{k,k-T_1}=\sum_{\iota=k-T_1}^{k}\Xi^T_{\iota,k}H_{\iota}^TR_{\iota}^{-1}H_{\iota}\Xi_{\iota,k}< r_1^{-1}\sum_{\iota=k-T_1}^{k}\Xi^T_{\iota,k}\Xi_{\iota,k}\le b_{\mathcal{I}}I, \quad \text{for $k \ge T_1$,}
\end{align}
where $b_{\mathcal{I}}>0$ is defined in \eqref{eq:a_b_I}. Combining \eqref{eq:infor_bound_a} and \eqref{eq:infor_bound_b}, we obtain
\begin{align*}
a_{\mathcal{I}}I < \mathcal{I}_{k,k-T_1} < b_{\mathcal{I}}I, \quad \text{for all $k \ge T_1=\max\left\{1, n-2\right\}$.}
\end{align*}

\noindent \textbf{Step 2}: Deriving the uniform complete controllability of the filter. For $k \ge T_1=\max\left\{1, n-2\right\}$, the controllability matrix is given by
\begin{align*}
\mathcal{C}_{k,k-T_1}=\sum_{\iota=k-T_1}^{k-1}\Xi_{k,\iota+1}Q_{\iota+1}\Xi_{k,\iota+1}^T>\bm{0}.
\end{align*}
It follows that
\begin{align*}
\bm{0}<a_{\mathcal{C}}I < \mathcal{C}_{k,k-T_1} < b_{\mathcal{C}}I, \quad \text{for all $k \ge T_1=\max\left\{1, n-2\right\}$,}
\end{align*}
with $0<a_{\mathcal{C}}< b_{\mathcal{C}}$ defined in \eqref{eq:a_b_C}.

\noindent \textbf{Step 3}: Deriving the bounds for the inverse of the error covariance. Let $a=\min\{a_{\mathcal{I}}, a_{\mathcal{C}}\}>0$, and $b=\max\{b_{\mathcal{I}}, b_{\mathcal{C}}\}> a>0$, combining Steps 1 and 2 we obtain
\begin{align*}
a I < \mathcal{I}_{k,k-T_1} < b I, \quad a I < \mathcal{C}_{k,k-T_1} < bI, \quad \text{for all $k \ge T_1=\max\left\{1, n-2\right\}$.}
\end{align*}
Based on the assumption that $\Gamma_{0|0}>\bm{0}$, it is concluded according to Lemma 7.1 and Lemma 7.2 in \cite{Jazwinski1970} that
\begin{align*}
\bm{0}<\left(\frac{a}{1+ab}\right) I < \Gamma_{k|k} < \left(\frac{1+ab}{a}\right) I, \quad \text{for all $k \ge T_1=\max\left\{1, n-2\right\}$.}
\end{align*}
Let $c_1=\frac{a}{1+ab}>0$ and $c_2=\frac{1+ab}{a} > c_1 >0$, the inverse of the error covariance satisfies
\begin{align*}
\bm{0}<c_1 I < \Gamma_{k|k}^{-1} < c_2 I, \quad \text{for all $k \ge T_1=\max\left\{1, n-2\right\}$,}
\end{align*}
which concludes the proof.

\subsection{Proof of Corollary \ref{cor:bound_gamma_analytical}}\label{ap:cor_bound}
\noindent \textbf{Step 1}: In this step, we derive a lower bound for $a_{\mathcal{C}}$, and an upper bound for $b_{\mathcal{C}}$. Define
\begin{equation*}
\bar{e}_{\iota}=\max_{M_{\kappa} \in \mathcal{A}_{\text{O}}}\left\{\sigma_{\max}\left(\prod_{\kappa=\iota-1}^{0}M_{\kappa}\prod_{\kappa=0}^{\iota-1}M^T_{\kappa}\right)\right\}, \quad \underline{e}_{\iota}=\min_{M_{\kappa} \in \mathcal{A}_{\text{O}}}\left\{\sigma_{\min}\left(\prod_{\kappa=\iota-1}^{0}M_{\kappa}\prod_{\kappa=0}^{\iota-1}M^T_{\kappa}\right)\right\}, \quad \text{for all $\iota \in\mathbb{Z}^+$},
\end{equation*}
where $\sigma_{\max}(\cdot)$ and $\sigma_{\min}(\cdot)$ are the maximum and minimum singular values of a matrix. Hence
\begin{align}\label{eq:b_c_ub}
a_{\mathcal{C}} \ge q_1\left( 1+\sum _{\iota=1}^{T_1-1}\underline{e}_{\iota}\right), \quad b_{\mathcal{C}} \le q_2\left( 1+\sum _{\iota=1}^{T_1-1}\bar{e}_{\iota}\right),
\end{align}
where $T_1=\max\left\{1, n-2\right\}$.
Due to the facts that
\begin{align*}
\sum_{c=1}^{n}M(r,c)\le 1,\quad\textrm{for all $M\in\mathcal{A}_{\text{O}}$ and $r\in\{1,2,\cdots,n\}$,}
\end{align*}
and
\begin{align*}
M(r,c)=0,\quad\textrm{for all $M\in\mathcal{A}_{\text{O}}$ when $|r-c|>1$,}\label{eq:sumAone01}
\end{align*}
it is concluded that the $(r,c)$-th entry of $\prod_{\kappa=\iota-1}^{0}M_{\kappa}$ satisfies
\begin{equation*}
\left(\prod_{\kappa=\iota-1}^{0}M_{\kappa}\right) \left(r,c\right) \left\{ \begin{array}{ll}
=0, & \textrm{if $|r-c|>\iota+1$,}\\
\le \bar{\theta}=1, & \textrm{otherwise.}
\end{array} \right.
\end{equation*}
Hence the diagonal entries of $\prod_{\kappa=\iota-1}^{0}M_{\kappa}\prod_{\kappa=0}^{\iota-1}M^T_{\kappa}$ satisfies
\begin{equation*}
\left(\prod_{\kappa=\iota-1}^{0}M_{\kappa}\prod_{\kappa=0}^{\iota-1}M^T_{\kappa}\right) \left(r,r\right) \le \left(2\iota+1\right)\bar{\theta}^2=2\iota+1,\quad\textrm{for all $r\in\{1,2,\cdots,n\}$,}
\end{equation*}
thus
\begin{equation}\label{eq:sigma_max_e}
\sigma_{\max}\left(\prod_{\kappa=\iota-1}^{0}M_{\kappa}\prod_{\kappa=0}^{\iota-1}M^T_{\kappa}\right) < 2(2\iota+1).
\end{equation}
Consequently, it is concluded that
\begin{equation}\label{eq:bar_e_ub}
\bar{e}_{\iota}<2(2\iota+1).
\end{equation}
We express a lower bound for $\underline{e}_{\iota}$ using the result in \cite{piazza2002upper}, i.e.,
\begin{equation}\label{eq:sigma_min_lb}
\begin{split}
\sigma_{\min}\left(M\right)\ge \frac{\left|\det M\right|}{2^{\left(n-2\right)/2}\left\|M\right\|_{\text{F}}},\quad \text{for $M\in \mathbb{R}^{n \times n}$ with positive singular values,}
\end{split}
\end{equation}
where $\|\cdot\|_{\text{F}}$ is the Frobenius norm. The determinant of $\prod_{\kappa=\iota-1}^{0}M_{\kappa}\prod_{\kappa=0}^{\iota-1}M^T_{\kappa}$ satisfies
\begin{equation}\label{eq:det_e_lb}
\det \left(\prod_{\kappa=\iota-1}^{0}M_{\kappa}\prod_{\kappa=0}^{\iota-1}M^T_{\kappa}\right)=\left(\prod_{\kappa=\iota-1}^{0}\det M_{\kappa}\right)^2 \ge \tilde{\theta}^{2\iota n},
\end{equation}
and its Frobenius norm satisfies
\begin{equation}\label{eq:frob_e_lb}
\left\|\prod_{\kappa=\iota-1}^{0}M_{\kappa}\prod_{\kappa=0}^{\iota-1}M^T_{\kappa}\right\|^2_{\text{F}}\le n \sigma^2_{\max}\left(\prod_{\kappa=\iota-1}^{0}M_{\kappa}\prod_{\kappa=0}^{\iota-1}M^T_{\kappa}\right)<4n(2\iota+1)^2,
\end{equation}
where the last inequality is due to \eqref{eq:sigma_max_e}. Substituting \eqref{eq:det_e_lb} and \eqref{eq:frob_e_lb} into \eqref{eq:sigma_min_lb}, we obtain
\begin{equation}\label{eq:bar_e_lb}
\underline{e}_{\iota}>\frac{\tilde{\theta}^{2\iota n}}{\sqrt{4n\left(2\iota+1\right)^22^{n-2}}}.
\end{equation}
Substituting \eqref{eq:bar_e_lb} and \eqref{eq:bar_e_ub} into \eqref{eq:b_c_ub} yields
\begin{equation}\label{eq:a_c_lb1}
\begin{split}
&a_{\mathcal{C}} \ge q_1\left(1+\sum_{\iota=1}^{T_1-1}\underline{e}_{\iota}\right)>q_1\left(1+\sum_{\iota=1}^{T_1-1}\frac{\tilde{\theta}^{2\iota n}}{\sqrt{4n\left(2\iota+1\right)^22^{n-2}}}\right)>q_1\left(1+\sum_{\iota=1}^{T_1-1}\frac{\tilde{\theta}^{2\iota n}}{\sqrt{4n\left(2T_1-1\right)^22^{n-2}}}\right)\\
&\quad =q_1\left(1+\frac{\tilde{\theta}^{2n}\left(1-\tilde{\theta}^{2n\left(T_1-1\right)}\right)}{\left(1-\tilde{\theta}^{2n}\right)\sqrt{4n\left(2T_1-1\right)^22^{n-2}}}\right),\\
&b_{\mathcal{C}} < q_2\left(2T_1^2-1\right).
\end{split}
\end{equation}

\noindent \textbf{Step 2}: We now derive a lower bound for $a_{\mathcal{I}}$, and an upper bound for $b_{\mathcal{I}}$. Due to the definition of $a_{\mathcal{I}}$ in \eqref{eq:a_b_I},
\begin{align}\label{eq:grammian_information}
a_{\mathcal{I}}\ge \lambda_{\min}\left( r_2^{-1}\Xi^T_{k-T_1,k}\tilde{\mathcal{G}}_{k,k-T_1}^T\tilde{\mathcal{G}}_{k,k-T_1}\Xi_{k-T_1,k}\right)\ge r_2^{-1}\lambda_{\min}\left(\tilde{\mathcal{G}}_{k,k-T_1}^T\tilde{\mathcal{G}}_{k,k-T_1}\right)\Xi_{k-T_1,k}^T\Xi_{k-T_1,k},
\end{align}
hence, we study the lower bounds for $\tilde{\mathcal{G}}_{k,k-T_1}^T\tilde{\mathcal{G}}_{k,k-T_1}$ and $\Xi_{k-T_1,k}^T\Xi_{k-T_1,k}$ individually, and then combine them together.


The lower bound for $\mathcal{G}_{k,k-T_1}$ can be derived as follows. Given the structure of $\tilde{\mathcal{G}}_{\cdot,\cdot}$ shown in \eqref{eq:tilde_mathcal_G} and \eqref{eq:tilde_mathcal_G_row}, we can pick a subset of $n$ rows of $\tilde{\mathcal{G}}_{k,k-T_1}$, i.e.,  $\mathcal{R}\subset \left\{1,2,3\cdots,2\left(T_1+1\right)\right\}$ and $|\mathcal{R}|=n$, such that removing the rows of $\tilde{\mathcal{G}}_{k,k-T_1}$ not in the selected subset yields a full rank matrix
\begin{equation*}
\begin{split}
\tilde{\mathcal{G}}_{k,k-T_1}\left(\mathcal{R},:\right)=&\left(
  \begin{array}{cc}
    \tilde{\mathcal{G}}^1_{k,k-T_1}&\\
    &\tilde{\mathcal{G}}^2_{k,k-T_1}
  \end{array}
\right)\in\mathbb{R}^{n \times n},
\end{split}
\end{equation*}
where $\tilde{\mathcal{G}}^1_{k,k-T_1}$ is a lower triangular matrix, and $\tilde{\mathcal{G}}^2_{k,k-T_1}$ is an upper triangular matrix. It follows that
\begin{equation*}
\begin{split}
\tilde{\mathcal{G}}_{k,k-T_1}^T\tilde{\mathcal{G}}_{k,k-T_1}\ge \tilde{\mathcal{G}}_{k,k-T_1}^T\left(\mathcal{R},:\right)\tilde{\mathcal{G}}_{k,k-T_1}\left(\mathcal{R},:\right)\ge \lambda_{\min}\left(\tilde{\mathcal{G}}_{k,k-T_1}^T\left(\mathcal{R},:\right)\tilde{\mathcal{G}}_{k,k-T_1}\left(\mathcal{R},:\right)\right)I.
\end{split}
\end{equation*}
Since $\tilde{\mathcal{G}}_{k,k-T_1}^T\left(\mathcal{R},:\right)\tilde{\mathcal{G}}_{k,k-T_1}\left(\mathcal{R},:\right)$ is a real symmetric matrix with non-zero eigenvalues, its singular values and eigenvalues coincide. We express a lower bound for $\lambda_{\min}\left(\tilde{\mathcal{G}}_{k,k-T_1}^T\left(\mathcal{R},:\right)\tilde{\mathcal{G}}_{k,k-T_1}\left(\mathcal{R},:\right)\right)$ using \eqref{eq:sigma_min_lb}. Since $\tilde{\mathcal{G}}^1_{k,k-T_1}$ and $\tilde{\mathcal{G}}^2_{k,k-T_1}$ are triangular matrices with elements satisfying \eqref{eq:tilde_g_lb},
\begin{equation}\label{eq:tilde_G_det}
\det\left(\tilde{\mathcal{G}}_{k,k-T_1}^T\left(\mathcal{R},:\right)\tilde{\mathcal{G}}_{k,k-T_1}\left(\mathcal{R},:\right)\right)=\left(\det\left(\tilde{\mathcal{G}}_{k,k-T_1}\left(\mathcal{R},:\right)\right)\right)^2\ge \left(\prod_{\iota=1}^{T_1}\underline{\theta}^{\iota}\right)^2=\underline{\theta}^{T_1\left(T_1+1\right)}.
\end{equation}
Now we express the upper bound for the Frobenius norm of $\tilde{\mathcal{G}}_{k,k-T_1}^T\left(\mathcal{R},:\right)\tilde{\mathcal{G}}_{k,k-T_1}\left(\mathcal{R},:\right)$. One may note that
\begin{equation*}
\lambda_{\max}\left(\tilde{\mathcal{G}}_{k,k-T_1}^T\left(\mathcal{R},:\right)\tilde{\mathcal{G}}_{k,k-T_1}\left(\mathcal{R},:\right)\right)\le \lambda_{\max}\left(\tilde{\mathcal{G}}_{k,k-T_1}^T\tilde{\mathcal{G}}_{k,k-T_1}\right) \le 2\left(T_1\bar{\theta}^2+1\right)=2\left(T_1+1\right),
\end{equation*}
where the last inequality is due to the structure of $\tilde{\mathcal{G}}_{k,k-T_1}$ shown in \eqref{eq:tilde_mathcal_G} and \eqref{eq:tilde_mathcal_G_row}, and the values of its entries discussed in \eqref{eq:tilde_g_lb}. It follows that
\begin{equation}\label{eq:tilde_G_F}
\left\|\tilde{\mathcal{G}}_{k,k-T_1}^T\left(\mathcal{R},:\right)\tilde{\mathcal{G}}_{k,k-T_1}\left(\mathcal{R},:\right)\right\|^2_{\text{F}}\le n \lambda_{\max}^2\left(\tilde{\mathcal{G}}_{k,k-T_1}^T\left(\mathcal{R},:\right)\tilde{\mathcal{G}}_{k,k-T_1}\left(\mathcal{R},:\right)\right)\le 4n(T_1+1)^2.
\end{equation}
Substituting \eqref{eq:tilde_G_det} and \eqref{eq:tilde_G_F} into \eqref{eq:sigma_min_lb}, we obtain
\begin{equation}\label{eq:lambda_min_tilde_G}
\begin{split}
\lambda_{\min}\left(\tilde{\mathcal{G}}_{k,k-T_1}^T\tilde{\mathcal{G}}_{k,k-T_1}\right)&\ge \lambda_{\min}\left(\tilde{\mathcal{G}}_{k,k-T_1}^T\left(\mathcal{R},:\right)\tilde{\mathcal{G}}_{k,k-T_1}\left(\mathcal{R},:\right)\right)\\
&=\sigma_{\min}\left(\tilde{\mathcal{G}}_{k,k-T_1}^T\left(\mathcal{R},:\right)\tilde{\mathcal{G}}_{k,k-T_1}\left(\mathcal{R},:\right)\right)\ge  \frac{\underline{\theta}^{T_1\left(T_1+1\right)}}{\sqrt{4n\left(T_1+1\right)^22^{n-2}}}.
\end{split}
\end{equation}
Given \eqref{eq:bar_e_ub}, it holds that
\begin{align}\label{eq:xi_lb}
\frac{1}{2\left(2T_1+1\right)}I< \bar{e}_{T_1}^{-1}I\le\left(\Xi_{k,k-T_1}\Xi_{k,k-T_1}^T\right)^{-1}=\Xi_{k-T_1,k}^T\Xi_{k-T_1,k}.
\end{align}
Substituting \eqref{eq:lambda_min_tilde_G} and \eqref{eq:xi_lb} into \eqref{eq:grammian_information}, we obtain
\begin{align*}
a_{\mathcal{I}}>\frac{r_2^{-1}\underline{\theta}^{T_1\left(T_1+1\right)}}{2\left(2T_1+1\right)\sqrt{4n\left(T_1+1\right)^22^{n-2}}}.
\end{align*}

Now we derive an upper bound for $b_{\mathcal{I}}$. One may note that
\begin{align*}
\left(\prod_{\kappa=\iota}^{T_1}M_{\kappa}^{-1}\right)^T\left(\prod_{\kappa=\iota}^{T_1}M_{\kappa}^{-1}\right)=\left(\prod_{\kappa=T_1}^{\iota}M_{\kappa}\prod_{\kappa=\iota}^{T_1}M_{\kappa}^{T}\right)^{-1}<\underline{e}_{T_1-\iota+1}^{-1}, \quad{\text{for $l\in\{1,2,\cdots, T_1\}$,}}
\end{align*}
it follows that
\begin{equation*}
\begin{split}
b_{\mathcal{I}}\le r_1^{-1}\left(1+\sum_{\iota=1}^{T_1}\underline{e}_{\iota}^{-1}\right)&<r_1^{-1}\left(1+\sum_{\iota=1}^{T_1} \frac{\sqrt{4n\left(2\iota+1\right)^22^{n-2}}}{\tilde{\theta}^{2\iota n}}\right)\\
&<r_1^{-1}\left(1+\sum_{\iota=1}^{T_1} \frac{\left(2\iota+1\right)\sqrt{4n2^{n-2}}}{\tilde{\theta}^{2T_1n}}\right)=r_1^{-1}\left(1+ \frac{T_1\left(T_1+2\right)\sqrt{4n2^{n-2}}}{\tilde{\theta}^{2T_1n}}\right),
\end{split}
\end{equation*}
where the second inequality is due to \eqref{eq:bar_e_lb}. Hence, the lower bound for $a_{\mathcal{I}}$ and the upper bound for $b_{\mathcal{I}}$ read:
\begin{align}\label{eq:a_I_lb1}
a_{\mathcal{I}}> \frac{r_2^{-1}\underline{\theta}^{T_1\left(T_1+1\right)}}{2\left(2T_1+1\right)\sqrt{4n\left(T_1+1\right)^22^{n-2}}}, \quad b_{\mathcal{I}}<r_1^{-1}\left(1+ \frac{T_1\left(T_1+2\right)\sqrt{4n2^{n-2}}}{\tilde{\theta}^{2T_1n}}\right).
\end{align}

\noindent \textbf{Step 3}: Based on \eqref{eq:a_c_lb1}, \eqref{eq:a_I_lb1} and recall that $a=\min\{a_{\mathcal{I}},a_{\mathcal{C}}\}$, $b=\max\{b_{\mathcal{I}},b_{\mathcal{C}}\}$, we have
\begin{align*}
&a>\underline{a}=\min\left\{q_1\left(1+\frac{\tilde{\theta}^{2n}\left(1-\tilde{\theta}^{2n\left(T_1-1\right)}\right)}{\left(1-\tilde{\theta}^{2n}\right)\sqrt{4n\left(2T_1-1\right)^22^{n-2}}}\right),\quad \frac{r_2^{-1}\underline{\theta}^{T_1\left(T_1+1\right)}}{2\left(2T_1+1\right)\sqrt{4n\left(T_1+1\right)^22^{n-2}}}\right\},\\
&b<\bar{b}=\max\left\{q_2\left(2T_1^2-1\right),\quad r_1^{-1}\left(1+ \frac{T_1\left(T_1+2\right)\sqrt{4n2^{n-2}}}{\tilde{\theta}^{2T_1n}}\right)\right\}.
\end{align*}
Consequently,
\begin{equation*}
\frac{\underline{a}}{1+\underline{a}\bar{b}}I<\frac{a}{1+ab}I=c_1 I < \Gamma_{k|k}^{-1} < c_2 I=\frac{1+ab}{a}I < \frac{1+\underline{a}\bar{b}}{\underline{a}}I, \quad \text{for all $k \ge T_1=\max\left\{1, n-2\right\}$,}
\end{equation*}
which completes the proof.

\subsection{Proof of Lemma \ref{lem:GammaBound1}}\label{ap:GammaBound1_proof}

\noindent \textbf{Step 1}: In this step, we derive an upper bound of $\Gamma_{k|k}$ for $k\in\left[0,\max\left\{1,n-2\right\}\right)$. When $0\le k < \max\left\{1,n-2\right\}$, the error covariance matrix satisfies
\begin{equation*}
\begin{split}
\Gamma_{k|k}& \le \Gamma_{k|k-1}=A_{k-1}\Gamma_{k-1|k-1}A_{k-1}^T+Q_{k-1}\\
& \le \left(\prod_{\kappa=k-1}^{0}A_{\kappa}\right)\Gamma_{0|0} \left(\prod_{\kappa=0}^{k-1}A_{\kappa}^T\right)+Q_{k-1}+\sum_{\iota=1}^{k-1}\left(\prod_{\kappa=k-1}^{\iota}A_{\kappa}\right)Q_{\iota-1} \left(\prod_{\kappa=\iota}^{k-1}A_{\kappa}^T\right)\\
&=\Xi_{k,0}\Gamma_{0|0}\Xi_{k,0}^T+Q_{k-1}+\sum_{\iota=1}^{k-1}\Xi_{k,\iota}Q_{\iota-1}\Xi_{k,\iota}^T\\
&< \left(\bar{e}_{k}\left\|\Gamma_{0|0}\right\|+q_2+\sum_{\iota=1}^{k-1}\bar{e}_{k-\iota}q_2\right)I\\
&< \left(2\left(2k+1\right)\left\|\Gamma_{0|0}\right\|+q_2+2\sum_{\iota=1}^{k-1} \left(2\left(k-\iota\right)+1\right)q_2\right)I\quad \text{(due to \eqref{eq:bar_e_ub})}\\
&= \left(2\left(2k+1\right)\left\|\Gamma_{0|0}\right\|+\left(2k^2-1\right)q_2\right)I.
\end{split}
\end{equation*}
where $\Xi_{\cdot,\cdot}$ is defined in \eqref{eq:xi}. Hence, the error covariance satisfies
\begin{equation*}
\Gamma_{k|k}\le \left\{ \begin{array}{ll}
\left(2\left(2n-5\right)\left\|\Gamma_{0|0}\right\|+\left(2\left(n-3\right)^2-1\right)q_2\right)I & \textrm{if $n\ge 4$}\\
\left\|\Gamma_{0|0}\right\|I & \textrm{if $2\le n<4$.}
\end{array} \right.
\end{equation*}
\textbf{Step 2}: In this step, we derive a lower bound of $\Gamma_{k|k}$ for $k\in\left[0,\max\left\{1,n-2\right\}\right)$. Consider a matrix sequence computed as follows
\begin{align*}
\breve{\Gamma}_{k+1|k}=\breve{A}_k\left(\breve{\Gamma}_{k|k-1}-\breve{\Gamma}_{k|k-1}\breve{H}_k^{T}(\breve{H}_k\breve{\Gamma}_{k|k-1}\breve{H}_k^{T}+\breve{R}_k)^{-1}\breve{H}_k\breve{\Gamma}_{k|k-1}\right)\breve{A}_{k}^{T}+\check{Q}_k,
\end{align*}
where
\begin{align*}
\breve{A}_k=A_k,\quad \breve{H}_k=H_k,\quad \breve{R}_k=R_k, \quad \breve{Q}_k=Q_k, \quad \text{for all $k\ge 0$,}\quad \text{and $\breve{\Gamma}_{0|-1}=\bm{0}$.}
\end{align*}
Due to Lemma 6.2 in \cite{ChuiChenKF}, it holds that $\Gamma_{k+1|k}>\breve{\Gamma}_{k+1|k}$ if $\Gamma_{k|k-1}>\breve{\Gamma}_{k|k-1}$. By definition we have $\Gamma_{0|-1}\ge \Gamma_{0|0}>\breve{\Gamma}_{0|-1}=\bm{0}$, which implies that $\breve{\Gamma}_{k+1|k}<\Gamma_{k+1|k}$ for all $k\ge 0$. Also due to Section 4.4 in \cite{Anderson1979}, we have $\breve{\Gamma}_{k+1|k}>\breve{\Gamma}_{k|k-1}$ for all $k\ge 0$ since $\breve{\Gamma}_{0|-1}=\bm{0}$. This yields
\begin{align*}
\Gamma_{k|k-1}>\breve{\Gamma}_{1|0}=\breve{A}_0\breve{\Gamma}_{0|0}\breve{A}_0^T+\breve{Q}_0>q_1I,\quad \text{for all $k\ge 1$,}
\end{align*}
where the second inequality is due to the fact that $\breve{\Gamma}_{0|0}=\bm{0}$ given $\breve{\Gamma}_{0|-1}=\bm{0}$, and $\breve{Q}_{0}>q_1I$.
It follows that
\begin{equation*}
\begin{split}
\Gamma_{k|k}=\left(\Gamma^{-1}_{k|k-1}+H_k^TR_k^{-1}H_k\right)^{-1}&>\left(\breve{\Gamma}^{-1}_{1|0}+H_k^TR_k^{-1}H_k\right)^{-1} \\
&>\left(q_1^{-1}I+r_1^{-1}H_k^TH_k\right)^{-1}>\left(q_1^{-1}+r_1^{-1}\right)^{-1}I,\quad \text{for all $k \ge 1$,}
\end{split}
\end{equation*}
where the last inequality is due to \eqref{eq:H_k_ub}.
It follows that
\begin{equation*}
\Gamma_{k|k}\ge \min\left\{\lambda_{\min}\left(\Gamma_{0|0}\right),\left(q_1^{-1}+r_1^{-1}\right)^{-1}\right\}I, \quad \text{for all $0\le k < \max\left\{1,n-2\right\}$.}
\end{equation*}
\textbf{Step 3}: Combining Steps 1 and 2, we obtain that when $0\le k < \max\left\{1,n-2\right\}$, the inverse of the error covariance satisfies
\begin{equation*}
\Gamma_{k|k}^{-1}\le \max\left\{\lambda^{-1}_{\min}\left(\Gamma_{0|0}\right),q_1^{-1}+r_1^{-1}\right\}I,
\end{equation*}
and
\begin{equation*}
\Gamma^{-1}_{k|k}\ge \left\{ \begin{array}{ll}
\left(2\left(2n-5\right)\left\|\Gamma_{0|0}\right\|+\left(2\left(n-3\right)^2-1\right)q_2\right)^{-1}I  & \textrm{if $n\ge 4$}\\
\left\|\Gamma_{0|0}\right\|^{-1}I & \textrm{if $2\le n<4$.}
\end{array} \right.
\end{equation*}
Hence, it is concluded that
\begin{equation*}
\bm{0}<\tilde{\mathfrak{c}}_1\left(\Gamma_{0|0}\right)I\le \Gamma^{-1}_{k|k} \le \tilde{\mathfrak{c}}_2\left(\Gamma_{0|0}\right)I,\quad \text{for all $0\le k < \max\left\{1,n-2\right\}$,}
\end{equation*}
which concludes the proof.

\subsection{Convergence rate of the common Lyapunov function with respect to neighbor disagreements}\label{ap:deltaV}
Due to the proof of Proposition 1, the one-step change of common Lyapunov function satisfies
\begin{align}
\Delta V_k<-2\sum_{\hat{i}=1}^{N-1}\gamma_{\hat{i},k}^{\hat{i}+1}\bm{\hat{\eta}}^T_{\hat{i},k|k-1}\hat{L}_{\hat{i}}\bm{\hat{\eta}}_{\hat{i},k|k-1},\label{eq:deltaV}
\end{align}
where $\hat{i}$ is the index for the overlapping regions. Define
\begin{align}
\bm{u}_{\hat{i},k}=\left(\left(\bm{u}^{\hat{i}+1}_{\hat{i},k}\right)^T\quad \left(\bm{u}^{\hat{i}}_{\hat{i}+1,k}\right)^T\right)^T=\left(\left(\bm{u}^{\hat{i}+1}_{\hat{i},k}\right)^T\quad -\left(\bm{u}^{\hat{i}+1}_{\hat{i},k}\right)^T\right)^T,\notag
\end{align}
also note that
\begin{align}
\hat{I}_{j,i}\bm{\eta}_{j,k|k-1}-\frac{1}{2}\left(\hat{I}_{j,i}\bm{\eta}_{j,k|k-1}+\hat{I}_{i,j}\bm{\eta}_{i,k|k-1}\right)=\frac{1}{2}\left(\hat{I}_{j,i}\bm{\eta}_{j,k|k-1}-\hat{I}_{i,j}\bm{\eta}_{i,k|k-1}\right)=\frac{1}{2}\bm{u}^j_{i,k},\notag
\end{align}
thus \eqref{eq:deltaV} becomes
\begin{align}
\Delta V_k<-\frac{1}{2}\sum_{\hat{i}=1}^{N-1}\gamma_{\hat{i},k}^{\hat{i}+1}\bm{u}_{\hat{i},k}^T\hat{L}_{\hat{i}}\bm{u}_{\hat{i},k}\leq -\sum_{\hat{i}=1}^{N-1}\gamma_{\hat{i},k}^{\hat{i}+1}\left\|\bm{u}_{\hat{i},k}\right\|^2=-\sqrt{2}\sum_{\hat{i}=1}^{N-1}\gamma_{\hat{i},k}^{\hat{i}+1}\left\|\bm{u}^{\hat{i}+1}_{\hat{i},k}\right\|^2, \quad\textrm{when $\bm{u}_{\hat{i},k}\neq 0$.}\notag
\end{align}
This indicates that $V_k$ strictly decreases at the rate proportional to $\gamma_{\hat{i},k}^{\hat{i}+1}$ and the 2-norm of the neighbor disagreement $\bm{u}^{\hat{i}+1}_{\hat{i},k}$ until the neighboring disagreements on all the overlapping regions converge to zero.

\subsection{Observable and unobservable subsystems in the unobservable modes}\label{ap:subsystems}
Generally in the SMM, $A_{k}\in\mathbb{R}^{n\times n}$ takes the following form in an unobservable mode:
\begin{align}
A_{k}=\left(
  \begin{array}{ccc}
\hat{\Theta}_{d_1}&\bm{0}_{d_1+1,1}&\bm{0}_{d_1+1,\bar{d_2+1}}\\
\left(
  \begin{array}{cc}
  \bm{0}_{1,d_1}&\frac{v_{\text{m}}\Delta t}{\Delta x}
  \end{array}
\right)&1&\left(
  \begin{array}{cc}
  \frac{w\Delta t}{\Delta x}&\bm{0}_{1,d_2}
  \end{array}
\right)\\
\bm{0}_{d_2+1,d_1+1}&\bm{0}_{d_2+1,1}&\hat{\Delta}_{d_2}
\end{array}
\right),  \notag
\end{align}
where $d_1+d_2+3=n$. We transform the state vector as follows:
\begin{align*}
\check{\rho}_k=U\rho_k,
\end{align*}
where the $(r,c)^{\textrm{th}}$ entry of $U$ is defined as
\begin{align}\label{eq:def_U}
U(r,c)= \left\{ \begin{array}{ll}
1& \textrm{if $r=r+1$ and $i\geq 3$}\\
1& \textrm{if $r=c=1$}\\
1& \textrm{if $r=2$ and $c=n$}\\
0& \textrm{otherwise.}
\end{array} \right.
\end{align}
Basically, the transformation $U$ makes
\begin{align}
\check{\rho}_{k}^2=\rho_{k}^n, \quad \textrm{and $\check{\rho}_{k}^i=\rho_{k}^{i-1}$ for $i\in\{3,4,\cdots,n\}$.}\notag
\end{align}
Hence, the state vector is transformed according to the observable and unobservable subsystems, i.e.,  
\begin{align*}
\check{\rho}_k=U\rho_{k}=\left(
  \begin{array}{c}
\check{\rho}_k^{(1)}\\
\check{\rho}_k^{(2)}
  \end{array}
\right),
\end{align*}
where the observable subsystem $\check{\rho}_k^{(1)}$ consists of the densities of the first and last cells in the freeway section, and the unobservable subsystem $\check{\rho}_k^{(2)}$ is formed by the densities of the interior cells in the section. 
Meanwhile, $A_{k}$ is transformed to $\check{A}_{k}$, which reads
\begin{align}
\check{A}_{k}=UA_{k}U^{-1}=\left(
  \begin{array}{cc}
    \check{A}^{(1)}&\bm{0}\\
    \check{A}^{(21)}&\check{A}^{(2)}_k
  \end{array}
\right),\notag
\end{align}
where
\begin{align}
&\check{A}^{(1)}=\left(
  \begin{array}{cc}
    1&\\
    &1
  \end{array}
\right)\in\mathbb{R}^{2\times 2},\notag\\
&\check{A}^{(21)}=\left(
  \begin{array}{cc}
    \frac{v_{\text{m}}\triangle t}{\triangle x}&0\\
    0&0\\
    \vdots&\vdots\\
    0&0\\
    0&\frac{w\triangle t}{\triangle x}
  \end{array}
\right)\in\mathbb{R}^{(n-2)\times 2},\notag
\end{align}
and
\begin{align}
\check{A}^{(2)}_{k}&=\left(
  \begin{array}{ccc}
\Theta_{d_1}&\bm{0}_{d_1,1}&\bm{0}_{d_1,\bar{d_2}}\\
\left(
  \begin{array}{cc}
  \bm{0}_{1,d_1-1}&\frac{v_{\text{m}}\Delta t}{\Delta x}
  \end{array}
\right)&1&\left(
  \begin{array}{cc}
  \frac{w\Delta t}{\Delta x}&\bm{0}_{1,d_2-1}
  \end{array}
\right)\\
\bm{0}_{d_2,d_1}&\bm{0}_{d_2,1}&\Delta_{d_2}
\end{array}
\right)\in\mathbb{R}^{(n-2)\times(n-2)}.\notag
\end{align}


When the densities of the boundary cells are measured, the output matrix is given by
\begin{align}
H=H_{\text{b}}=\left(
  \begin{array}{ccccc}
    1&0&\cdots&0&0\\
    0&0&\cdots&0&1
  \end{array}
\right),\notag
\end{align}
and the transformed observation matrix is given by
\begin{align}
\check{H}=\left(\check{H}^{(1)}\quad\bm{0}\right),\notag
\end{align}
where $\check{H}^{(1)}=I_2$.

Divide the prior estimation error covariance matrix based on the observable and unobservable subsystems as follows:
\begin{align}
\check{\Gamma}_{k|k-1}=\left(
  \begin{array}{cc}
    \check{\Gamma}^{(1)}_{k|k-1}&\left(\check{\Gamma}^{(12)}_{k|k-1}\right)^T\\
    \check{\Gamma}^{(21)}_{k|k-1}&\check{\Gamma}^{(2)}_{k|k-1}
  \end{array}
\right),\notag
\end{align}
where $\check{\Gamma}^{(1)}_{k|k-1}$ is of dimension 2, and $\check{\Gamma}^{(2)}_{k|k-1}$ is of dimension $n-2=d_1+d_2+1$.

In the DLKCF, the prior error covariance matrix is computed recursively by the Riccati equation
\begin{align}
\check{\Gamma}_{k+1|k}=&\check{A}_k\left(\check{\Gamma}_{k|k-1}-\check{\Gamma}_{k|k-1}\check{H}^{T}\left(\check{H}\check{\Gamma}_{k|k-1}\check{H}^{T}+R_k\right)^{-1}\check{H}\check{\Gamma}_{k|k-1}\right)\check{A}_{k}^{T}+\check{Q}_k,\notag
\end{align}
Define (recall that $\check{A}^{(1)}=I_2$)
\begin{align}
\check{\Upsilon}^{(1)}_{k}&=\check{A}^{(1)}-\check{A}^{(1)}\check{\Gamma}^{(1)}_{k|k-1}\left(\check{H}^{(1)}\right)^T\left(\check{H}\check{\Gamma}_{k|k-1}\check{H}^{T}+R_k\right)^{-1}\check{H}^{(1)}\notag\\
&=\check{A}^{(1)}-\check{A}^{(1)}\check{K}^{(1)}_{k}\check{H}^{(1)},\notag
\end{align}
and apply partition into observable and unobservable subsystems, we obtain the following two blocks of equations:
\begin{align}
\check{\Gamma}^{(1)}_{k+1|k}=\check{\Upsilon}^{(1)}_{k}\check{\Gamma}^{(1)}_{k|k-1}\left(\check{A}^{(1)}\right)^T+\check{Q}_k^{(1)},\label{eq:18}
\end{align}
\begin{equation}\label{eq:19}
\begin{split}
\check{\Gamma}^{(12)}_{k+1|k}=&\check{\Upsilon}^{(1)}_{k}\check{\Gamma}^{(12)}_{k|k-1}\left(\check{A}^{(2)}_k\right)^T+\check{\Upsilon}^{(1)}_{k}\check{\Gamma}^{(1)}_{k|k-1}\left(\check{A}^{(21)}\right)^T+\check{Q}_k^{(12)}.
\end{split}
\end{equation}


\subsection{Proof of Lemma \ref{lem_u_bound_k}}\label{ap:lem_u_bound_k}
As detailed in Appendix \ref{ap:subsystems}, the observable subsystem in an unobservable freeway section consists of the two boundary cells, and the remaining state variables form the unobservable subsystem. We transform the state vector according to observable and unobservable subsystems, i.e.,
\begin{align*}
\check{\rho}_k=U\rho_{k}=\left(
  \begin{array}{c}
\check{\rho}_k^{(1)}\\
\check{\rho}_k^{(2)}
  \end{array}
\right),
\end{align*}
where $\check{\rho}_k^{(1)}$ consists of the densities of the first and last cells in the freeway section, and $\check{\rho}_k^{(2)}$ is formed by the densities of the interior cells in the section. The transformation matrix $U$ is defined in \eqref{eq:def_U}.

The transformed Kalman gain is given by
\begin{align*}
\check{K}_k=UK_{k}=\left(
  \begin{array}{c}
\check{K}_k^{(1)}\\
\check{K}_k^{(21)}
  \end{array}
\right),
\end{align*}
where
\begin{align*}
\check{K}^{(1)}_k=\left(
  \begin{array}{cc}
K_{k}(1,1) & K_{k}(1,2)\\
K_{k}(n,1) & K_{k}(n,2)
  \end{array}
\right),\quad \text{and}\quad\check{K}^{(21)}_k=\left(
  \begin{array}{cc}
K_{k}(2,1) & K_{k}(2,2)\\
\vdots &\vdots\\
K_{k}(n-1,1) & K_{k}(n-1,2)
  \end{array}
\right).
\end{align*}
The proof consists of the following five steps. Step 1 derives an upper bound for $\left\|K_{\underline{k}_{\text{U}}+1}\right\|_{\infty}$. Step 2 derives an upper bound of $\left\|\check{K}^{(1)}_k\right\|_{\infty}$ for $k\in(\underline{k}_{\text{U}}+1,\bar{k}_{\text{U}}]$. In Step 3, we study the convergence rate of the error dynamics of the observable subsystem, which is also related to the boundedness of $\check{K}^{(21)}_k$. Based on the convergence rate obtained in Step 3, Step 4 derives an upper bound of $\left\|\check{K}^{(21)}_k\right\|_{\infty}$ for $k\in(\underline{k}_{\text{U}}+1,\bar{k}_{\text{U}}]$. Step 5 combines the above steps together and concludes the proof.

\noindent \textbf{Step 1}: At time step $\underline{k}_{\text{U}}+1$, the Kalman gain is computed as follows:
\begin{align*}
K_{\underline{k}_{\text{U}}+1}=\Gamma_{\underline{k}_{\text{U}}+1|\underline{k}_{\text{U}}}H_{\underline{k}_{\text{U}}+1}^T\left(R_{\underline{k}_{\text{U}}+1}+H_{\underline{k}_{\text{U}}+1}\Gamma_{\underline{k}_{\text{U}}+1|\underline{k}_{\text{U}}}H_{\underline{k}_{\text{U}}+1}^T\right)^{-1}, \quad \text{where $\Gamma_{\underline{k}_{\text{U}}+1|\underline{k}_{\text{U}}}=A_{\underline{k}_{\text{U}}}\Gamma_{\underline{k}_{\text{U}}|\underline{k}_{\text{U}}}A_{\underline{k}_{\text{U}}}^T+Q_{\underline{k}_{\text{U}}}$.}
\end{align*}
Given that $\left\|\Gamma_{\underline{k}_{\text{U}}|\underline{k}_{\text{U}}}\right\|_{\infty}\le\sqrt{n}\left\|\Gamma_{\underline{k}_{\text{U}}|\underline{k}_{\text{U}}}\right\|$, $\left\|Q_k\right\|_{\infty}<\sqrt{n}q_2$,  $\left\|A_k\right\|_{\infty}=1+\frac{\Delta t}{\Delta x}\left(w+v_{\text{m}}\right)$ and $\left\|A_k^T\right\|_{\infty}=1$ for $A_k\in\mathcal{A}_{\text{U}}$, the prior error covariance at time $\underline{k}_{\text{U}}+1$ satisfies
\begin{align*}
\left\|\Gamma_{\underline{k}_{\text{U}}+1|\underline{k}_{\text{U}}}\right\|_{\infty}\le \left\|A_{\underline{k}_{\text{U}}}\right\|_{\infty}\left\|\Gamma_{\underline{k}_{\text{U}}|\underline{k}_{\text{U}}}\right\|_{\infty}\left\|A_{\underline{k}_{\text{U}}}^T\right\|_{\infty}+\left\|Q_{\underline{k}_{\text{U}}}\right\|_{\infty}<\sqrt{n}\left\|\Gamma_{\underline{k}_{\text{U}}|\underline{k}_{\text{U}}}\right\|\left(1+\frac{\Delta t}{\Delta x}\left(w+v_{\text{m}}\right)\right)+\sqrt{n}q_2.
\end{align*}
Moreover, since
\begin{equation*}
\begin{split}
\left\|\left(R_{\underline{k}_{\text{U}}+1}+H_{\underline{k}_{\text{U}}+1}\Gamma_{\underline{k}_{\text{U}}+1|\underline{k}_{\text{U}}}H_{\underline{k}_{\text{U}}+1}^T\right)^{-1}\right\|_{\infty}&\le \sqrt{2} \left\|\left(R_{\underline{k}_{\text{U}}+1}+H_{\underline{k}_{\text{U}}+1}\Gamma_{\underline{k}_{\text{U}}+1|\underline{k}_{\text{U}}}H_{\underline{k}_{\text{U}}+1}^T\right)^{-1}\right\|\\
&= \sqrt{2} \left(\sigma_{\min}\left(R_{\underline{k}_{\text{U}}+1}+H_{\underline{k}_{\text{U}}+1}\Gamma_{\underline{k}_{\text{U}}+1|\underline{k}_{\text{U}}}H_{\underline{k}_{\text{U}}+1}^T\right)\right)^{-1}\\
& \le \sqrt{2} \left(\sigma_{\min}\left(R_{\underline{k}_{\text{U}}+1}\right)\right)^{-1}<\sqrt{2} r_1^{-1},
\end{split}
\end{equation*}
it follows that
\begin{equation*}
\begin{split}
\left\|K_{\underline{k}_{\text{U}}+1}\right\|_{\infty}& \le \left\|\Gamma_{\underline{k}_{\text{U}}+1|\underline{k}_{\text{U}}}\right\|_{\infty}\left\|H_{\underline{k}_{\text{U}}+1}^T\right\|_{\infty}\left\|\left(R_{\underline{k}_{\text{U}}+1}+H_{\underline{k}_{\text{U}}+1}\Gamma_{\underline{k}_{\text{U}}+1|\underline{k}_{\text{U}}}H_{\underline{k}_{\text{U}}+1}^T\right)^{-1}\right\|_{\infty}\\
&<\sqrt{2n}r_1^{-1}\left(\left\|\Gamma_{\underline{k}_{\text{U}}|\underline{k}_{\text{U}}}\right\|\left(1+\frac{\Delta t}{\Delta x}\left(w+v_{\text{m}}\right)\right)+q_2\right).
\end{split}
\end{equation*}

\noindent \textbf{Step 2}: For the error covariance of the observable subsystem (defined as $\check{\Gamma}^{(1)}_{k|k}$), one may note that
\begin{equation*}
\left(\check{\Gamma}_{k|k}^{(1)}\right)^{-1}=\left(\check{\Gamma}_{k|k-1}^{(1)}\right)^{-1}+\left(\hat{H}^{(1)}\right)^TR_k^{-1}\hat{H}^{(1)},
\end{equation*}
where $\hat{H}^{(1)}=I$ is the output matrix of the boundary measurements modeled in the observable subsystem, it follows that
\begin{equation}\label{eq:gamma_1_ub}
\check{\Gamma}_{k|k}^{(1)}<R_k<r_2I,\quad \text{for $k\in(\underline{k}_{\text{U}},\bar{k}_{\text{U}}]$.}
\end{equation}
For all $k\in(\underline{k}_{\text{U}}+1, \bar{k}_{\text{U}}]$, the Kalman gain associated with the observable subsystem is given by
\begin{align*}
\check{K}^{(1)}_{k}=\check{\Gamma}^{(1)}_{k|k-1}\left(\check{H}^{(1)}\right)^T\left(R_{k}+\check{H}^{(1)}\check{\Gamma}^{(1)}_{k|k-1}\left(\check{H}^{(1)}\right)^T\right)^{-1}, \quad \text{where $\check{\Gamma}^{(1)}_{k|k-1}=\check{A}^{(1)}\check{\Gamma}^{(1)}_{k-1|k-1}\left(\check{A}^{(1)}\right)^T+\check{Q}_{k-1}^{(1)}$.}
\end{align*}
Following the similar argument as in Step 1, we obtain
\begin{equation*}
\begin{split}
\left\|\check{K}^{(1)}_{k}\right\|_{\infty}<2r_1^{-1}\left(r_2+q_2\right), \quad \text{for all $k\in(\underline{k}_{\text{U}}+1,\bar{k}_{\text{U}}]$.}
\end{split}
\end{equation*}

\noindent \textbf{Step 3}: Denote as $\check{\mathcal{I}}^{(1)}_{\cdot,\cdot}$ and $\check{\mathcal{C}}^{(1)}_{\cdot,\cdot}$ the information and controllability matrix of the observable subsystem, it follows that
\begin{align*}
2 r_2^{-1}I< \check{\mathcal{I}}^{(1)}_{k,k-1}=R_{k-1}^{-1}+R_k^{-1}< 2r_1^{-1}I,\quad \text{and $q_1 I< \check{\mathcal{C}}^{(1)}_{k,k-1}=\check{Q}_{k}^{(1)}< q_2I$,} \quad \text{for all $k\in(\underline{k}_{\text{U}}, \bar{k}_{\text{U}}]$.}
\end{align*}
Define
\begin{align*}
\check{a}=\min\left\{2r_2^{-1},q_1\right\},\quad \check{b}=\max\left\{2r_1^{-1},q_2\right\},\quad \check{c}_1=\frac{\check{a}}{1+\check{a}\check{b}},\quad \check{c}_2=\frac{1+\check{a}\check{b}}{\check{a}},
\end{align*}
the error covariance of the observable subsystem satisfies
\begin{align*}
\check{c}_1 I< \left(\check{\Gamma}_{k|k}^{(1)}\right)^{-1}
< \check{c}_2 I, \quad \text{for all $k\in(\underline{k}_{\text{U}}, \bar{k}_{\text{U}}]$,}
\end{align*}
according to Lemma 7.1 and Lemma 7.2 in \cite{Jazwinski1970}.

Define the Lyapunov function of the observable subsystem as $\check{V}_k=\left(\check{\bm{\eta}}^{(1)}_{k|k}\right)^{T}\left(\check{\Gamma}^{(1)}_{k|k}\right)^{-1}\check{\bm{\eta}}^{(1)}_{k|k}$. According to Lemma 3 in \cite{Olfati-SaberCDC2009}, the one-step change of $\check{V}_k$ is given by (see also \eqref{eq:lyap_one_step} for a more detailed derivation)
\begin{equation}
\begin{split}\label{eq:lyap_one_step_o}
\Delta \check{V}_{k+1}&=\check{V}_{k+1}-\check{V}_k\\
&=-\left(\check{\bm{\eta}}^{(1)}_{k|k}\right)^{T}\left(\check{\Gamma}^{(1)}_{k|k}+\check{\Gamma}^{(1)}_{k|k}\left(\check{A}^{(1)}\right)^T\left(\check{Q}_{k}^{(1)}+\check{\Gamma}^{(1)}_{k+1|k}\left(\check{H}^{(1)}\right)^TR_{k+1}^{-1}\check{H}^{(1)}\check{\Gamma}^{(1)}_{k+1|k}\right)^{-1}\check{A}^{(1)}\check{\Gamma}^{(1)}_{k|k}\right)^{-1}\check{\bm{\eta}}^{(1)}_{k|k}\\
&\le -\left\|\check{\Gamma}^{(1)}_{k|k}+\check{\Gamma}^{(1)}_{k|k}\left(\check{A}^{(1)}\right)^T\left(\check{Q}_{k}^{(1)}+\check{\Gamma}^{(1)}_{k+1|k}\left(\check{H}^{(1)}\right)^TR_{k+1}^{-1}\check{H}^{(1)}\check{\Gamma}^{(1)}_{k+1|k}\right)^{-1}\check{A}^{(1)}\check{\Gamma}^{(1)}_{k|k}\right\|^{-1}\left\|\check{\bm{\eta}}^{(1)}_{k|k}\right\|^2,
\end{split}
\end{equation}
where
\begin{equation}\label{eq:lyap_one_step_1_o}
\begin{split}
&\left\|\check{\Gamma}^{(1)}_{k|k}+\check{\Gamma}^{(1)}_{k|k}\left(\check{A}^{(1)}\right)^T\left(\check{Q}_{k}^{(1)}+\check{\Gamma}^{(1)}_{k+1|k}\left(\check{H}^{(1)}\right)^TR_{k+1}^{-1}\check{H}^{(1)}\check{\Gamma}^{(1)}_{k+1|k}\right)^{-1}\check{A}^{(1)}\check{\Gamma}^{(1)}_{k|k}\right\|\\
\le & \left\|\check{\Gamma}^{(1)}_{k|k}\right\|+\left\|\check{\Gamma}^{(1)}_{k|k}\left(\check{A}^{(1)}\right)^T\left(\check{Q}_{k}^{(1)}+\check{\Gamma}^{(1)}_{k+1|k}\left(\check{H}^{(1)}\right)^TR_{k+1}^{-1}\check{H}^{(1)}\check{\Gamma}^{(1)}_{k+1|k}\right)^{-1}\check{A}^{(1)}\check{\Gamma}^{(1)}_{k|k}\right\|\\
< & \check{c}_1^{-1}+q_1^{-1}\left\|\check{\Gamma}^{(1)}_{k|k}\left(\check{A}^{(1)}\right)^T\check{A}^{(1)}\check{\Gamma}^{(1)}_{k|k}\right\|\\
\le & \check{c}_1^{-1}+q_1^{-1}\left\|\check{\Gamma}^{(1)}_{k|k}\right\|\left\|\left(\check{A}^{(1)}\right)^T\check{A}^{(1)}\right\|\left\|\check{\Gamma}^{(1)}_{k|k}\right\|\\
< & \check{c}_1^{-1}+q_1^{-1}\check{c}_1^{-2}\sigma_{\max}\left(\check{A}^{(1)}\right),
\end{split}
\end{equation}
It follows that the Lyapunov function $\check{V}_k$ satisfies
\begin{align*}
\check{c}_1\left\|\check{\bm{\eta}}^{(1)}_{k|k}\right\|^2<\check{V}_k<\check{c}_2\left\|\check{\bm{\eta}}^{(1)}_{k|k}\right\|^2,\quad \text{and }\check{V}_{k+1}-\check{V}_k<-\check{c}_3\left\|\check{\bm{\eta}}^{(1)}_{k|k}\right\|^2, \quad \text{for all $k\in(\underline{k}_{\text{U}}, \bar{k}_{\text{U}}]$,}
\end{align*}
where
\begin{align*}
\check{c}_3=\check{c}_1^{-1}+q_1^{-1}\check{c}_1^{-2}\sigma_{\max}\left(\check{A}^{(1)}\right)=\check{c}_1^{-1}+q_1^{-1}\check{c}_1^{-2}.
\end{align*}
Hence, the 2-norm of the mean estimation error of the observable subsystem satisfies
\begin{equation}\label{eq:normerror_o}
\begin{split}
\left\|\check{\bm{\eta}}^{(1)}_{k|k}\right\|<\left(\frac{\check{V}_k}{\check{c}_1}\right)^{\frac{1}{2}}&<\left(\frac{\check{V}_{\underline{k}_{\text{U}}+1}\left(1-\check{c}_3\check{c}_2^{-1}\right)^{k-\underline{k}_{\text{U}}-1}}{\check{c}_1}\right)^{\frac{1}{2}}< \left(\frac{\check{c}_2\left\|\check{\bm{\eta}}^{(1)}_{\underline{k}_{\text{U}}+1|\underline{k}_{\text{U}}+1}\right\|^2\left(1-\check{c}_3\check{c}_2^{-1}\right)^{k-\underline{k}_{\text{U}}-1}}{\check{c}_1}\right)^{\frac{1}{2}}\\
&=\left(\frac{\check{c}_2}{\check{c}_1}\right)^{\frac{1}{2}}\left\|\check{\bm{\eta}}^{(1)}_{\underline{k}_{\text{U}}+1|\underline{k}_{\text{U}}+1}\right\|\left(\left(1-\check{c}_3\check{c}_2^{-1}\right)^{\frac{1}{2}}\right)^{k-\underline{k}_{\text{U}}-1},\quad \text{for all $k\in(\underline{k}_{\text{U}}+1,\bar{k}_{\text{U}}]$}.
\end{split}
\end{equation}
Moreover, the mean estimation error of the observable subsystem is given as follows:
\begin{align}\label{eq:error_evolve_o}
\check{\bm{\eta}}^{(1)}_{k|k}=\prod_{\kappa=k}^{\underline{k}_{\text{U}}+2}\check{\Upsilon}^{(1)}_{\kappa}\check{\bm{\eta}}^{(1)}_{\underline{k}_{\text{U}}+1|\underline{k}_{\text{U}}+1},\quad \text{for all $k\in(\underline{k}_{\text{U}}+1,\bar{k}_{\text{U}}]$},
\end{align}
where $\check{\Upsilon}^{(1)}_{\kappa}=\check{\Gamma}^{(1)}_{\kappa|\kappa}\left(\check{\Gamma}_{\kappa|\kappa-1}^{(1)}\right)^{-1}\check{A}^{(1)}=\check{\Gamma}^{(1)}_{\kappa|\kappa}\left(\check{\Gamma}_{\kappa|\kappa-1}^{(1)}\right)^{-1}$. Combining \eqref{eq:normerror_o} and \eqref{eq:error_evolve_o}, it is concluded based on the definition of matrix induced norm that
\begin{align}\label{eq:error_evolve_rate}
\left\|\prod_{\kappa=k}^{\underline{k}_{\text{U}}+2}\check{\Upsilon}^{(1)}_{\kappa}\right\|\le \left(\frac{\check{c}_2}{\check{c}_1}\right)^{\frac{1}{2}}\left(\left(1-\check{c}_3\check{c}_2^{-1}\right)^{\frac{1}{2}}\right)^{k-\underline{k}_{\text{U}}-1}, \quad \text{for $k\in(\underline{k}_{\text{U}}+1,\bar{k}_{\text{U}}]$}.
\end{align}

\noindent \textbf{Step 4}:
Vectorizing both sides of \eqref{eq:19} yields
\begin{align*}
\textrm{vec}\left\{\check{\Gamma}^{(12)}_{k+1|k}\right\}=\left(\check{A}^{(2)}_k\otimes\check{\Upsilon}^{(1)}_{k}\right)\textrm{vec}\left\{\check{\Gamma}^{(12)}_{k|k-1}\right\}+\textrm{vec}\left\{\check{\Upsilon}^{(1)}_{k}\check{\Gamma}^{(1)}_{k|k-1}\left(\check{A}^{(21)}_k\right)^T\right\}+\textrm{vec}\left\{\check{Q}_k^{(12)}\right\},\text{for $k\in(\underline{k}_{\text{U}},\bar{k}_{\text{U}}]$,}
\end{align*}
which implies that
\begin{align}
\textrm{vec}\left\{\check{\Gamma}^{(12)}_{k+1|k}\right\}=\left(\prod_{\kappa=k}^{\underline{k}_{\text{U}}+2}\left(\check{A}^{(2)}_{\kappa}\otimes\check{\Upsilon}^{(1)}_{\kappa}\right)\right)\textrm{vec}\left\{\check{\Gamma}^{(12)}_{\underline{k}_{\text{U}}+2|\underline{k}_{\text{U}}+1}\right\}+\Phi_k,\quad\text{for $k\in(\underline{k}_{\text{U}}+1,\bar{k}_{\text{U}}]$,}\label{eq:Gamma12}
\end{align}
where
\begin{align}
\Phi_k=&\textrm{vec}\left\{\check{\Upsilon}^{(1)}_{k}\check{\Gamma}^{(1)}_{k|k-1}\left(\check{A}^{(21)}_k\right)^T+\check{Q}_k^{(12)}\right\}+\left(\check{A}^{(2)}_k\otimes\check{\Upsilon}^{(1)}_{k}\right)\textrm{vec}\left\{\check{\Upsilon}^{(1)}_{k-1}\check{\Gamma}^{(1)}_{k-1|k-2}\left(\check{A}^{(21)}_{k-1}\right)^T+\check{Q}_{k-1}^{(12)}\right\}\notag\\
&+\left(\check{A}^{(2)}_k\otimes\check{\Upsilon}^{(1)}_{k}\right)\left(\check{A}^{(2)}_{k-1}\otimes\check{\Upsilon}^{(1)}_{k-1}\right)\textrm{vec}\left\{\check{\Upsilon}^{(1)}_{k-2}\check{\Gamma}^{(1)}_{k-2|k-3}\left(\check{A}^{(21)}_{k-2}\right)^T+\check{Q}_{k-2}^{(12)}\right\}\notag\\
&+\cdots+\prod_{\kappa=k}^{\underline{k}_{\text{U}}+3}\left(\check{A}^{(2)}_{\kappa}\otimes\check{\Upsilon}^{(1)}_{\kappa}\right)\textrm{vec}\left\{\check{\Upsilon}^{(1)}_{\underline{k}_{\text{U}}+2}\check{\Gamma}^{(1)}_{\underline{k}_{\text{U}}+2|\underline{k}_{\text{U}}+1}\left(\check{A}^{(21)}_{\underline{k}_{\text{U}}+2}\right)^T+\check{Q}_{\underline{k}_{\text{U}}+2}^{(12)}\right\}.\notag
\end{align}
The explicit form of $\check{A}^{(2)}_{k}\otimes\check{\Upsilon}^{(1)}_{k}$ reads
\begin{align*}
\check{A}^{(2)}_{k}\otimes\check{\Upsilon}^{(1)}_{k}=\left(
  \begin{array}{ccc}
    \check{A}^{(2)}_{k}(1,1)\check{\Upsilon}^{(1)}_{k}&\cdots&\check{A}^{(2)}_{k}(1,n-2)\check{\Upsilon}^{(1)}_{k}\\
    \vdots&\ddots&\vdots\\
    \check{A}^{(2)}_{k}(n-2,1)\check{\Upsilon}^{(1)}_{k}&\cdots&\check{A}^{(2)}_{k}(n-2,n-2)\check{\Upsilon}^{(1)}_{k}\\
  \end{array}
\right),
\end{align*}
hence
\begin{align*}
\prod_{\kappa=k}^{\underline{k}_{\text{U}}+2}\left(\check{A}^{(2)}_{\kappa}\otimes\check{\Upsilon}^{(1)}_{\kappa}\right)=\left(
  \begin{array}{ccc}
    \vartheta_{k}(1,1)\prod_{\kappa=k}^{\underline{k}_{\text{U}}+2}\check{\Upsilon}^{(1)}_{\kappa}&\cdots&\vartheta_{k}(1,n-2)\prod_{\kappa=k}^{\underline{k}_{\text{U}}+2}\check{\Upsilon}^{(1)}_{\kappa}\\
    \vdots&\ddots&\vdots\\
    \vartheta_{k}(n-2,1)\prod_{\kappa=k}^{\underline{k}_{\text{U}}+2}\check{\Upsilon}^{(1)}_{\kappa}&\cdots&\vartheta_{k}(n-2,n-2)\prod_{\kappa=k}^{\underline{k}_{\text{U}}+2}\check{\Upsilon}^{(1)}_{\kappa}\\
  \end{array}
\right),
\end{align*}
where $\vartheta_{k}(r,c)$ is the $(r,c)^{\textrm{th}}$ element of $\prod_{\kappa=k}^{\underline{k}_{\text{U}}+2}\check{A}^{(2)}_{\kappa}$. Also since
\begin{align*}
\sum_{r=1}^{n-2}\check{A}^{(2)}_{k}(r,c)=1,\quad\textrm{for all $k\in(\underline{k}_{\text{U}}, \bar{k}_{\text{U}}]$ and $c\in\{1,2,\cdots,n-2\}$,}
\end{align*}
and
\begin{align*}
0\le \left|\check{A}^{(2)}_{k}(r,c)\right|\le 1,\quad\textrm{for all $k\in(\underline{k}_{\text{U}}, \bar{k}_{\text{U}}]$ and $r, c\in\{1,2,\cdots,n-2\}$,}
\end{align*}
it follows that $0\le \vartheta_{k}(r,c) \le 1$ for all $r$, $c$ and $k\in(\underline{k}_{\text{U}}, \bar{k}_{\text{U}}]$. Consequently
\begin{align}\label{eq:a_upsilon_infty_bound}
\left\|\prod_{\kappa=k}^{\underline{k}_{\text{U}}+2}\left(\check{A}^{(2)}_{\kappa}\otimes\check{\Upsilon}^{(1)}_{\kappa}\right)\right\|_{\infty}\le \left(n-2\right)\left\|\prod_{\kappa=k}^{\underline{k}_{\text{U}}+2}\check{\Upsilon}^{(1)}_{\kappa}\right\|_{\infty}\le \sqrt{2}\left(n-2\right) \left(\frac{\check{c}_2}{\check{c}_1}\right)^{\frac{1}{2}}\left(\left(1-\check{c}_3\check{c}_2^{-1}\right)^{\frac{1}{2}}\right)^{k-\underline{k}_{\text{U}}-1}=\bar{t}\bar{q}^{k-\underline{k}_{\text{U}}-1},
\end{align}
where the last inequality is due to \eqref{eq:error_evolve_rate}. Recall from \eqref{eq:gamma_1_ub} that $\check{\Gamma}_{k|k}^{(1)}<R_k<r_2I$ for $k\in(\underline{k}_{\text{U}},\bar{k}_{\text{U}}]$. Since
\begin{equation*}
\check{\Gamma}_{k|k-1}^{(1)}=\check{A}^{(1)}\check{\Gamma}^{(1)}_{k-1|k-1}\left(\check{A}^{(1)}\right)^T+\check{Q}^{(1)}_{k-1},
\end{equation*}
the prior error covariance of the observable subsystem satisfies
\begin{equation*}
q_1 I<\check{\Gamma}_{k|k-1}^{(1)}<\left(r_2+q_2\right) I,\quad \text{for $k\in(\underline{k}_{\text{U}},\bar{k}_{\text{U}}]$.}
\end{equation*}
As a consequence,
\begin{equation}\label{eq:bound_a_upsilon}
\begin{split}
\left\|\check{\Gamma}^{(1)}_{k|k}\left(\check{\Gamma}_{k|k-1}^{(1)}\right)^{-1}\right\|<r_2q_1^{-1},\quad \left\|\check{\Upsilon}^{(1)}_{k}\right\|_{\infty}\le\sqrt{2}\left\|\check{\Gamma}^{(1)}_{k|k}\left(\check{\Gamma}_{k|k-1}^{(1)}\right)^{-1}\right\|<\sqrt{2}r_2q_1^{-1},\quad \text{for $k\in(\underline{k}_{\text{U}},\bar{k}_{\text{U}}]$.}
\end{split}
\end{equation}
It follows that\footnote{Recall that for matrix $M\in\mathbb{R}^{p\times q}$, $\left\|M\right\|_{\max}\le \left\|M\right\|_{2}=\max_{1\le r\le p, 1 \le c\le q}\left|M(r,c)\right|$.}
\begin{equation}\label{eq:bound_additive}
\begin{split}
\left\|\textrm{vec}\left\{\check{\Upsilon}^{(1)}_{k}\check{\Gamma}^{(1)}_{k|k-1}\left(\check{A}^{(21)}_k\right)^T\right\}+\textrm{vec}\left\{\check{Q}_k^{(12)}\right\}\right\|_{\infty}&\le \left\|\check{\Upsilon}^{(1)}_{k}\check{\Gamma}^{(1)}_{k|k-1}\left(\check{A}^{(21)}_k\right)^T\right\|_{\infty}+\left\|\check{Q}_k^{(12)}\right\|_{\max}\\
& <\sqrt{2}r_2q_1^{-1}\frac{\Delta t}{\Delta x}\max\{v_{\text{m}},w\}\left\|\check{\Gamma}^{(1)}_{k|k-1}\right\|_{\infty}+q_2\\
& < 2r_2q_1^{-1}\frac{\Delta t}{\Delta x}\max\{v_{\text{m}},w\}\left(r_2+q_2\right)+q_2\\
& =\bar{p},\quad \text{for $k\in(\underline{k}_{\text{U}},\bar{k}_{\text{U}}]$}.
\end{split}
\end{equation}
Substituting \eqref{eq:a_upsilon_infty_bound} and \eqref{eq:bound_additive} into \eqref{eq:Gamma12}, we obtain
\begin{align*}
\left\|\textrm{vec}\left\{\check{\Gamma}^{(12)}_{k+1|k}\right\}\right\|_{\infty}\le \mathfrak{b}\left(k\right)\triangleq\bar{t}\bar{q}^{k-\underline{k}_{\text{U}}-1}\left\|\textrm{vec}\left\{\check{\Gamma}^{(12)}_{\underline{k}_{\text{U}}+2|\underline{k}_{\text{U}}+1}\right\}\right\|_{\infty}+\bar{p}+\bar{t}\bar{p}\sum_{\iota=1}^{k-\underline{k}_{\text{U}}-2}\bar{q}^{\iota},\quad \text{for $k\in(\underline{k}_{\text{U}}+1,\bar{k}_{\text{U}}]$,}
\end{align*}
where $\mathfrak{b}(k)$ is either a non-increasing or a non-decreasing function of $k$.
Hence, we obtain that for $k\in(\underline{k}_{\text{U}},\bar{k}_{\text{U}}]$,
\begin{align*}
\left\|\textrm{vec}\left\{\check{\Gamma}^{(12)}_{k+1|k}\right\}\right\|_{\infty}&\le\max\left\{\left\|\textrm{vec}\left\{\check{\Gamma}^{(12)}_{\underline{k}_{\text{U}}+2|\underline{k}_{\text{U}}+1}\right\}\right\|_{\infty},\quad \mathfrak{b}(\underline{k}_{\text{U}}+2),\quad \lim_{k\rightarrow \infty}\mathfrak{b}(k)\right\}\\
&\le\max\left\{\left\|\textrm{vec}\left\{\check{\Gamma}^{(12)}_{\underline{k}_{\text{U}}+2|\underline{k}_{\text{U}}+1}\right\}\right\|_{\infty},\quad \bar{t}\bar{q}\left\|\textrm{vec}\left\{\check{\Gamma}^{(12)}_{\underline{k}_{\text{U}}+2|\underline{k}_{\text{U}}+1}\right\}\right\|_{\infty}+\bar{p},\quad \frac{\bar{t}\bar{p}\bar{q}}{1-\bar{q}}+\bar{p}\right\},
\end{align*}
where
\begin{equation*}
\begin{split}
\left\|\textrm{vec}\left\{\check{\Gamma}^{(12)}_{\underline{k}_{\text{U}}+2|\underline{k}_{\text{U}}+1}\right\}\right\|_{\infty}&\le \left\|\check{\Gamma}^{(12)}_{\underline{k}_{\text{U}}+2|\underline{k}_{\text{U}}+1}\right\|_{\infty}<\sqrt{n}\left\|\Gamma_{\underline{k}_{\text{U}}+1|\underline{k}_{\text{U}}+1}\right\|\left(1+\frac{\Delta t}{\Delta x}\left(w+v_{\text{m}}\right)\right)+\sqrt{n}q_2\\
& \le \sqrt{n}\left\|\Gamma_{\underline{k}_{\text{U}}+1|\underline{k}_{\text{U}}}\right\|\left(1+\frac{\Delta t}{\Delta x}\left(w+v_{\text{m}}\right)\right)+\sqrt{n}q_2\\
&< n\sqrt{n}\left\|\Gamma_{\underline{k}_{\text{U}}|\underline{k}_{\text{U}}}\right\|\left(1+\frac{\Delta t}{\Delta x}\left(w+v_{\text{m}}\right)\right)^2+n\sqrt{n}q_2\left(1+\frac{\Delta t}{\Delta x}\left(w+v_{\text{m}}\right)\right)+\sqrt{n}q_2\\
& =\bar{\gamma}.
\end{split}
\end{equation*}
Also since
\begin{align*}
\check{K}_k^{(21)}=\left(\check{\Gamma}_{k|k-1}^{(12)}\right)^T\left(\check{H}^{(1)}\right)^T\left(R_k+\check{H}^{(1)}\check{\Gamma}_{k|k-1}^{(1)}\left(\check{H}^{(1)}\right)^T\right)^{-1},
\end{align*}
it follows that
\begin{align*}
\left\|\check{K}_k^{(21)}\right\|_{\infty}&\le \sqrt{2}r_1^{-1}\left\|\left(\check{\Gamma}_{k|k-1}^{(12)}\right)^T\right\|_{\infty}\le 2\sqrt{2}r_1^{-1}\max\left\{\bar{\gamma},\quad \bar{t}\bar{q}\bar{\gamma}+\bar{p},\quad \frac{\bar{t}\bar{p}\bar{q}}{1-\bar{q}}+\bar{p}\right\}, \quad\text{for $k\in(\underline{k}_{\text{U}}+1,\bar{k}_{\text{U}}]$.}
\end{align*}

\noindent \textbf{Step 5}: Combining Steps 1, 2 and 4, it can be concluded that for $k\in (\underline{k}_{\text{U}},\bar{k}_{\text{U}}]$
\begin{align*}
&\left\|K_k\right\|_{\infty}=\left\|U^{-1}\check{K}_k\right\|_{\infty}=\left\|\check{K}_k\right\|_{\infty}\\
\le&\sqrt{2}r_1^{-1}\max\left\{\sqrt{n}\left(\left\|\Gamma_{\underline{k}_{\text{U}}|\underline{k}_{\text{U}}}\right\|\left(1+\frac{\Delta t \left(w+v_{\text{m}}\right)}{\Delta x}\right)+q_2\right), \sqrt{2}\left(r_2+q_2\right),  2\bar{\gamma}, 2\left(\bar{t}\bar{q}\bar{\gamma}+\bar{p}\right), 2\left(\frac{\bar{t}\bar{p}\bar{q}}{1-\bar{q}}+\bar{p}\right)\right\},
\end{align*}
which completes the proof.

\subsection{Proof of Proposition \ref{Prop:UltimateBoundedness}}\label{ap:prop_ub_proof}

The proof is by induction. For all $\epsilon>0$, since the upstream cell is in the observable subsystem, we have $\bm{\rho}^1_{k|k}\rightarrow\rho^1_{k}$, where $\rho^1_{k}\geq0$. Hence a finite time $T_1(\epsilon)$ exists such that $\bm{\rho}^1_{k|k}>-\frac{\epsilon}{n}$ for all $k>T_1(\epsilon)$.

Suppose $\bm{\rho}^{l-1}_{k|k}>-\frac{(l-1)\epsilon}{n}$. For all $l\in\{2,\cdots,n\}$, if $\bm{\rho}^{l}_{k|k}<-\frac{(l-1)\epsilon}{n}$, we obtain from \eqref{eq:qflow} that
\begin{align}
&\mathfrak{f}\left(\bm{\rho}^{l-1}_{k|k},\bm{\rho}^{l}_{k|k}\right)=v_{\text{m}}\bm{\rho}^{l-1}_{k|k}>-v_{\text{m}}\frac{(l-1)\epsilon}{n},\label{UBProofqIn}\\
&\mathfrak{f}\left(\bm{\rho}^{l}_{k|k},\bm{\rho}^{l+1}_{k|k}\right)\leq v_{\text{m}}\bm{\rho}^{l}_{k|k}.\label{UBProofqOut}
\end{align}
Combining \eqref{UBProofqIn} and \eqref{UBProofqOut} with \eqref{eq:DiscretizedLWR}, and adding an information update term from the analysis step yields
\begin{align}
\bm{\rho}^{l}_{k+1|k+1}>\bm{\rho}^{l}_{k|k}+\frac{v_{\text{m}}\Delta t}{\Delta x}\left|\bm{\rho}^{l}_{k|k}+\frac{(l-1)\epsilon}{n}\right|-\bar{c}\left\|\check{\bm{\eta}}^{(1)}_{k|k}\right\|_{\infty},\notag
\end{align}
where $\bar{c}>0$ is a finite scalar whose existence is guaranteed by the boundedness of Kalman gain, and we denote $\check{\bm{\eta}}^{(1)}_{k|k}=\left(\bm{\eta}^{1}_{k|k},\bm{\eta}^{n}_{k|k}\right)^T$ as the posterior estimation error of the upstream and downstream cells, which form an observable subsystem, hence $\left\|\check{\bm{\eta}}^{(1)}_{k|k}\right\|_{\infty}\rightarrow0$ as $k\rightarrow\infty$. Thus a class $\mathcal{K}$ function $\mathfrak{w}_0(\cdot)$ and a continuous positive definite function $\mathfrak{w}\left(|\cdot|\right)$ on $\mathbb{R}$ exist such that $\bm{\rho}^{l}_{k+1|k+1}-\bm{\rho}^{l}_{k|k}>\mathfrak{w}\left(\left|\bm{\rho}^{l}_{k|k}+\frac{(l-1)\epsilon}{n}\right|\right)$ for all $\left|\bm{\rho}^{l}_{k|k}+\frac{(l-1)\epsilon}{n}\right|\geq\mathfrak{w}_0\left(\left\|\check{\bm{\eta}}^{(1)}_{k|k}\right\|_{\infty}\right)$. This indicates that the one-step change of the estimates is always positive, and large enough so that a finite time $T_l(\epsilon)$ exists such that $\bm{\rho}^{l}_{k|k}>-\frac{l\epsilon}{n}$ for all $k>T_l(\epsilon)$ \cite{Khalil}.

By induction we conclude that if $\bm{\rho}^{n-1}_{k|k}>-\frac{(n-1)\epsilon}{n}$, a finite time $T_n(\epsilon)$ exists such that $\bm{\rho}^{n}_{k|k}>-\epsilon$ for all $k>T_n(\epsilon)$. Letting $T(\epsilon)=\max_{l}\{T_l(\epsilon)\}=T_n(\epsilon)$, we obtain $\bm{\rho}^{l}_{k|k}>-\epsilon$ for all $k>T(\epsilon)$ and $l\in\{
1,2,\cdots,n\}$. This proves the ultimate lower bound of the estimates. The proof for an ultimate upper bound is similar, with a variation that the induction is conducted from $n$ to 1.

\subsection{Proof of Lemma \ref{lem:Lyapunov}}\label{ap:Proof_Lyap}
Consider the following linear system:
\begin{align}
\zeta_{k}=F_{k}A_{k-1}\zeta_{k-1}, \quad\text{for $\underline{k}_{\text{O}}< k \le \bar{k}_{\text{O}}$},\label{eq:sys_zeta0}
\end{align}
it follows that
\begin{align}
\zeta_{k}=\left(\prod_{\kappa=k-1}^{\underline{k}_{\text{O}}}F_{\kappa+1}A_{\kappa}\right)\zeta_{\underline{k}_{\text{O}}}, \quad\text{for $\underline{k}_{\text{O}}< k \le \bar{k}_{\text{O}}$}.\label{eq:sys_zeta0_1}
\end{align}
Let $V_{\zeta,k}=\zeta_{k}^T\Gamma_{k|k}^{-1}\zeta_{k}$ be the Lyapunov function candidate of system \eqref{eq:sys_zeta0}. According to Lemma 3 in \cite{Olfati-SaberCDC2009}, the one-step change of $V_{\zeta,k}$ is given by
\begin{equation}
\begin{split}\label{eq:lyap_one_step}
\Delta V_{\zeta,k+1}&=V_{\zeta,k+1}-V_{\zeta,k}=-\zeta_{k}^T\left(A_{k}^TF_{k+1}^T\Gamma_{k+1|k+1}^{-1}F_{k+1}A_{k}-\Gamma_{k|k}^{-1}\right)\zeta_{k}\\
&=-\zeta_{k}^T\left(\Gamma_{k|k}^{-1}-A_{k}^T\left(A_{k}\Gamma_{k|k}A_{k}^T+Q_{k}+\Gamma_{k+1|k}H_{k+1}^TR_{k+1}^{-1}H_{k+1}\Gamma_{k+1|k}\right)^{-1}A_{k}\right)\zeta_{k}\\
&=-\zeta_{k}^T\Gamma_{k|k}^{-1}\left(\Gamma_{k|k}^{-1}+A_{k}^T\left(Q_{k}+\Gamma_{k+1|k}H_{k+1}^TR_{k+1}^{-1}H_{k+1}\Gamma_{k+1|k}\right)^{-1}A_{k}\right)^{-1}\Gamma_{k|k}^{-1}\zeta_{k}\\
&=-\zeta_{k}^T\left(\Gamma_{k|k}+\Gamma_{k|k}A_{k}^T\left(Q_{k}+\Gamma_{k+1|k}H_{k+1}^TR_{k+1}^{-1}H_{k+1}\Gamma_{k+1|k}\right)^{-1}A_{k}\Gamma_{k|k}\right)^{-1}\zeta_{k}\\
&\le -\left\|\Gamma_{k|k}+\Gamma_{k|k}A_{k}^T\left(Q_{k}+\Gamma_{k+1|k}H_{k+1}^TR_{k+1}^{-1}H_{k+1}\Gamma_{k+1|k}\right)^{-1}A_{k}\Gamma_{k|k}\right\|^{-1}\left\|\zeta_{k}\right\|^2,
\end{split}
\end{equation}
where
\begin{equation}\label{eq:lyap_one_step_1}
\begin{split}
&\left\|\Gamma_{k|k}+\Gamma_{k|k}A_{k}^T\left(Q_{k}+\Gamma_{k+1|k}H_{k+1}^TR_{k+1}^{-1}H_{k+1}\Gamma_{k+1|k}\right)^{-1}A_{k}\Gamma_{k|k}\right\|\\
\le & \left\|\Gamma_{k|k}\right\|+\left\|\Gamma_{k|k}A_{k}^T\left(Q_{k}+\Gamma_{k+1|k}H_{k+1}^TR_{k+1}^{-1}H_{k+1}\Gamma_{k+1|k}\right)^{-1}A_{k}\Gamma_{k|k}\right\|\\
< & d_1^{-1}+q_1^{-1}\left\|\Gamma_{k|k}A_{k}^TA_{k}\Gamma_{k|k}\right\|\\
\le & d_1^{-1}+q_1^{-1}\left\|\Gamma_{k|k}\right\|\left\|A_{k}^TA_{k}\right\|\left\|\Gamma_{k|k}\right\|\\
\le & d_1^{-1}+q_1^{-1}d_1^{-2}
\max_{M\in\mathcal{A}_{\text{O}}}\sigma^2_{\max}\left(M\right),
\end{split}
\end{equation}
Combining \eqref{eq:lyap_one_step} and \eqref{eq:lyap_one_step_1}, we obtain
\begin{equation*}
\begin{split}
\Delta V_{k+1}< -\left(d_1^{-1}+q_1^{-1}d_1^{-2}
\max_{M\in\mathcal{A}_{\text{O}}}\sigma^2_{\max}\left(M\right)\right)^{-1}\left\|\zeta_{k}\right\|^2=-\mathfrak{d}\left(d_1,d_2\right)\left\|\zeta_{k}\right\|^2,  \quad \text{for all $\underline{k}_{\text{O}} < k\le \bar{k}_{\text{O}}$.}
\end{split}
\end{equation*}
Consequently, the 2-norm of $\zeta_{k}$ satisfies
\begin{equation}\label{eq:normerror_zeta0}
\begin{split}
\left\|\zeta_{k}\right\|&\le \left(\frac{V_{\zeta,k}}{d_1}\right)^{\frac{1}{2}}<\left(\frac{V_{\zeta,\underline{k}_{\text{O}}}\left(1-\mathfrak{d}\left(d_1,d_2\right)d_2^{-1}\right)^{k-\underline{k}_{\text{O}}}}{d_1}\right)^{\frac{1}{2}}\\
&\le\left(\frac{d_2\left\|\zeta_{\underline{k}_{\text{O}}}\right\|^2\left(1-\mathfrak{d}\left(d_1,d_2\right)d_2^{-1}\right)^{k-\underline{k}_{\text{O}}}}{d_1}\right)^{\frac{1}{2}}\\
&=\left(d_2d_1^{-1}\right)^{\frac{1}{2}}\left(\left(1-\mathfrak{d}\left(d_1,d_2\right)d_2^{-1}\right)^{\frac{1}{2}}\right)^{k-\underline{k}_{\text{O}}}\left\|\zeta_{\underline{k}_{\text{O}}}\right\|, \quad\text{for $\underline{k}_{\text{O}}< k \le \bar{k}_{\text{O}}$}.
\end{split}
\end{equation}
Combining \eqref{eq:sys_zeta0_1} and \eqref{eq:normerror_zeta0}, we obtain
\begin{equation*}
\begin{split}
\left\|\prod_{\kappa=k-1}^{\underline{k}_{\text{O}}}F_{\kappa+1}A_{\kappa}\right\|\le \hat{a}\hat{q}^{k-\underline{k}_{\text{O}}},\quad\text{for $\underline{k}_{\text{O}}< k \le \bar{k}_{\text{O}}$}.
\end{split}
\end{equation*}

Moreover, since for all $M\in\mathcal{A}_{\text{O}}$, the diagonal element of $M^TM$ is no greater than $1+\left(\frac{\Delta t}{\Delta x}\max\left\{v_{\text{m}},w\right\}\right)^2$, thus $\left\|M^TM\right\|<2\left(1+\left(\frac{\Delta t}{\Delta x}\max\left\{v_{\text{m}},w\right\}\right)^2\right)$ since $M^TM>\bm{0}$. It follows that
\begin{align*}
\mathfrak{d}\left(d_1,d_2\right)> \left(d_1^{-1}+2q_1^{-1}d_1^{-2}\left(1+\left(\frac{\Delta t}{\Delta x}\max\left\{v_{\text{m}},w\right\}\right)^2\right)\right)^{-1}.
\end{align*}

\subsection{Proof of Proposition \ref{prop:switch}}\label{ap:prop_switch}

The proof can be done by a straightforward application of the results in Propositions \ref{prop_u_bound} and \ref{prop_o_bound}. Note that when the section is unobservable at time $0$ (i.e., $\underline{k}_{\text{U}}^{1}=0$), we have $\underline{k}_{\text{U}}^{r+1}=\bar{k}_{\text{O}}^{r}$ and $\bar{k}_{\text{U}}^{r}=\underline{k}_{\text{O}}^{r}$ for all $r \in \mathbb{Z}^{+}$. When the section is observable at time $0$ (i.e., $\underline{k}_{\text{O}}^{1}=0$), we have $\underline{k}_{\text{O}}^{r+1}=\bar{k}_{\text{U}}^{r}$ and $\bar{k}_{\text{O}}^{r}=\underline{k}_{\text{U}}^{r}$ for all $r \in \mathbb{Z}^{+}$.

\noindent \textbf{Step 1}: When $r\ge 2$.

\noindent{\textit{(a) For the $r^{\text{th}}$ unobservable time interval $k\in(\underline{k}_{\emph{\text{U}}}^{r},\bar{k}_{\emph{\text{U}}}^{r}]$}:} When the observable time interval right before $(\underline{k}_{\text{U}}^{r},\bar{k}_{\text{U}}^{r}]$ is sufficiently long such that condition \eqref{eq:time_condition} is satisfied, the estimation error at time $\underline{k}_{\text{U}}^{r}$ satisfies (note that the observable time interval right before $(\underline{k}_{\text{U}}^{r},\bar{k}_{\text{U}}^{r}]$ can also be written as $(\bar{k}_{\text{U}}^{r-1},\underline{k}_{\text{U}}^{r}]$)
\begin{align*}
\left\|\bm{\eta}_{\underline{k}_{\text{U}}^{r}|\underline{k}_{\text{U}}^{r}}\right\|\le \delta +\hat{c}+\frac{\hat{c}\mathfrak{a}\left(\Gamma_{\bar{k}_{\text{U}}^{r-1}|\bar{k}_{\text{U}}^{r-1}}\right)\mathfrak{q}\left(\Gamma_{\bar{k}_{\text{U}}^{r-1}|\bar{k}_{\text{U}}^{r-1}}\right)}{1-\mathfrak{q}\left(\Gamma_{\bar{k}_{\text{U}}^{r-1}|\bar{k}_{\text{U}}^{r-1}}\right)}
\end{align*}
based on Proposition \ref{prop_o_bound}. As a consequence, Proposition \ref{prop_u_bound} gives
\begin{equation*}
\begin{split}
\bm{\rho}_{k|k}^{l}&>-\left(\delta +\hat{c}+\frac{\hat{c}\mathfrak{a}\left(\Gamma_{\bar{k}_{\text{U}}^{r-1}|\bar{k}_{\text{U}}^{r-1}}\right)\mathfrak{q}\left(\Gamma_{\bar{k}_{\text{U}}^{r-1}|\bar{k}_{\text{U}}^{r-1}}\right)}{1-\mathfrak{q}\left(\Gamma_{\bar{k}_{\text{U}}^{r-1}|\bar{k}_{\text{U}}^{r-1}}\right)}\right)\left(c_0+\left(n-2\right)\mathfrak{c}\left(\Gamma_{\underline{k}_{\text{U}}^{r}|\underline{k}_{\text{U}}^{r}}\right)\right),\\
\bm{\rho}_{k|k}^{l}&<\varrho_{\text{m}}+\left(\delta +\hat{c}+\frac{\hat{c}\mathfrak{a}\left(\Gamma_{\bar{k}_{\text{U}}^{r-1}|\bar{k}_{\text{U}}^{r-1}}\right)\mathfrak{q}\left(\Gamma_{\bar{k}_{\text{U}}^{r-1}|\bar{k}_{\text{U}}^{r-1}}\right)}{1-\mathfrak{q}\left(\Gamma_{\bar{k}_{\text{U}}^{r-1}|\bar{k}_{\text{U}}^{r-1}}\right)}\right)\left(c_0+\left(n-2\right)\mathfrak{c}\left(\Gamma_{\underline{k}_{\text{U}}^{r}|\underline{k}_{\text{U}}^{r}}\right)\right),
\end{split}
\end{equation*}
for all $l\in\{1,\cdots,n\}$ and $k\in(\underline{k}_{\text{U}}^{r},\bar{k}_{\text{U}}^{r}]$. Consequently, the estimation error satisfies $\left\|\bm{\eta}_{k|k}\right\|\le\mathfrak{e}\left(\delta,\Gamma_{\bar{k}_{\text{U}}^{r-1}|\bar{k}_{\text{U}}^{r-1}},\Gamma_{\underline{k}_{\text{U}}^{r}|\underline{k}_{\text{U}}^{r}}\right)$ for all $k\in(\underline{k}_{\text{U}}^{r},\bar{k}_{\text{U}}^{r}]$.

\noindent{\textit{(b) For the $r^{\text{th}}$ observable time interval $k\in(\underline{k}_{\emph{\text{O}}}^{r},\bar{k}_{\emph{\text{O}}}^{r}]$}:} Note that the unobservable time interval right before $(\underline{k}_{{\text{O}}}^{r},\bar{k}_{{\text{O}}}^{r}]$ is written as $(\underline{k}_{{\text{U}}}^{r},\bar{k}_{{\text{U}}}^{r}]=(\bar{k}_{{\text{O}}}^{r-1},\underline{k}_{\text{O}}^{r}]$ when the section is unobservable at time 0, and is written as $(\underline{k}_{{\text{U}}}^{r-1},\bar{k}_{{\text{U}}}^{r-1}]=(\bar{k}_{{\text{O}}}^{r-1},\underline{k}_{\text{O}}^{r}]$ when the section is observable at time $0$. Similar to Case (\textit{a}) in Step 1, when $\bar{k}_{{\text{O}}}^{r-1}-\underline{k}_{{\text{O}}}^{r-1}$ satisfies condition \eqref{eq:time_condition}, the estimation error at time $\underline{k}_{\text{O}}^{r}$ satisfies
\begin{align*}
\left\|\bm{\eta}_{\underline{k}_{\text{O}}^{r}|\underline{k}_{\text{O}}^{r}}\right\|\le\mathfrak{e}\left(\delta,\Gamma_{\underline{k}_{\text{O}}^{r-1}|\underline{k}_{\text{O}}^{r-1}},\Gamma_{\bar{k}_{\text{O}}^{r-1}|\bar{k}_{\text{O}}^{r-1}}\right).
\end{align*}
Applying Proposition \ref{prop_o_bound}, it is concluded that for $k\in(\underline{k}_{\text{O}}^{r},\bar{k}_{\text{O}}^{r}]$,
\begin{align*}
\left\|\bm{\eta}_{k|k}\right\|\le\max\left\{\hat{c}+\mathfrak{a}\left(\Gamma_{\underline{k}_{\text{O}}^{r}|\underline{k}_{\text{O}}^{r}}\right)\mathfrak{q}\left(\Gamma_{\underline{k}_{\text{O}}^{r}|\underline{k}_{\text{O}}^{r}}\right) \mathfrak{e}\left(\delta,\Gamma_{\underline{k}_{\text{O}}^{r-1}|\underline{k}_{\text{O}}^{r-1}},\Gamma_{\bar{k}_{\text{O}}^{r-1}|\bar{k}_{\text{O}}^{r-1}}\right), \hat{c}+\frac{\hat{c}\mathfrak{a}\left(\Gamma_{\underline{k}_{\text{O}}^{r}|\underline{k}_{\text{O}}^{r}}\right)\mathfrak{q}\left(\Gamma_{\underline{k}_{\text{O}}^{r}|\underline{k}_{\text{O}}^{r}}\right)}{1-\mathfrak{q}\left(\Gamma_{\underline{k}_{\text{O}}^{r}|\underline{k}_{\text{O}}^{r}}\right)} \right\}.
\end{align*}

\noindent \textbf{Step 2}: When $r=1$ and $\underline{k}_{{\text{U}}}^1=0$.

\noindent{\textit{(a) For the $1^{\text{st}}$ unobservable time interval $k\in(\underline{k}_{\emph{\text{U}}}^{1},\bar{k}_{\emph{\text{U}}}^{1}]$}:} In this case, the section is unobservable at time $0$. Since $\bm{\rho}^{l}_{0|0}\in [0,\varrho_{\text{m}}]$, the initial estimation error satisfies $\left\|\bm{\eta}_{0|0}\right\|\le \sqrt{n}\varrho_{\text{m}}$. According to Proposition \ref{prop_u_bound}, we have
\begin{equation*}
\begin{split}
\bm{\rho}_{k|k}^{l}&>-\sqrt{n}\varrho_{\text{m}}\left(c_0+\left(n-2\right)\mathfrak{c}\left(\Gamma_{0|0}\right)\right)\\
\bm{\rho}_{k|k}^{l}&<\varrho_{\text{m}}+\sqrt{n}\varrho_{\text{m}}\left(c_0+\left(n-2\right)\mathfrak{c}\left(\Gamma_{0|0}\right)\right),
\end{split}
\end{equation*}
for all $l\in\{1,\cdots,n\}$ and $k\in(\underline{k}_{\text{U}}^{1},\bar{k}_{\text{U}}^{1}]$. It follows that
\begin{align*}
\left\|\bm{\eta}_{k|k}\right\|\le \sqrt{n}\left(\sqrt{n}\varrho_{\text{m}}\left(c_0+\left(n-2\right)\mathfrak{c}\left(\Gamma_{0|0}\right)\right)+\varrho_{\text{m}} \right)=\mathfrak{e}_0\left(\Gamma_{0|0}\right), \quad \text{for $k\in(\underline{k}_{\text{U}}^{1},\bar{k}_{\text{U}}^{1}]$.}
\end{align*}

\noindent{\textit{(b) For the $1^{\text{st}}$ observable time interval $k\in(\underline{k}_{\emph{\text{O}}}^{1},\bar{k}_{\emph{\text{O}}}^{1}]$}:} When the section switches from an unobservable mode at time $\underline{k}_{\text{O}}^{1}$ to an observable mode at time $\underline{k}_{\text{O}}^{1}+1$, it is shown in Case (\textit{a}) of Step 2 that the mean error is upper bounded by
\begin{align*}
\left\|\eta_{\underline{k}_{\text{O}}^{1}|\underline{k}_{\text{O}}^{1}}\right\|\le \mathfrak{e}_0\left(\Gamma_{0|0}\right).
\end{align*}
Applying Proposition \ref{prop_o_bound}, it follows that
\begin{align*}
\left\|\bm{\eta}_{k|k}\right\|
\le\max\left\{\hat{c}+\mathfrak{a}\left(\Gamma_{\underline{k}_{\text{O}}^{1}|\underline{k}_{\text{O}}^{1}}\right)\mathfrak{q}\left(\Gamma_{\underline{k}_{\text{O}}^{1}|\underline{k}_{\text{O}}^{1}}\right) \mathfrak{e}_0\left(\Gamma_{0|0}\right), \hat{c}+\frac{\hat{c}\mathfrak{a}\left(\Gamma_{\underline{k}_{\text{O}}^{1}|\underline{k}_{\text{O}}^{1}}\right)\mathfrak{q}\left(\Gamma_{\underline{k}_{\text{O}}^{1}|\underline{k}_{\text{O}}^{1}}\right)}{1-\mathfrak{q}\left(\Gamma_{\underline{k}_{\text{O}}^{1}|\underline{k}_{\text{O}}^{1}}\right)} \right\},\quad \text{for $k\in(\underline{k}_{\text{O}}^{1},\bar{k}_{\text{O}}^{1}]$.}
\end{align*}

\noindent \textbf{Step 3}: When $r=1$ and $\underline{k}_{{\text{O}}}^1=0$.

\noindent{\textit{(a) For the $1^{\text{st}}$ unobservable time interval $k\in(\underline{k}_{{\emph{\text{U}}}}^{1},\bar{k}_{\emph{\text{U}}}^{1}]$}:} In this case, the section is observable at time $0$. When $\bar{k}_{{\text{O}}}^{1}-\underline{k}_{{\text{O}}}^{1}$ is larger than the third residence time listed in \eqref{eq:time_condition}, the estimation error at time $\bar{k}_{\text{O}}^{1}$ satisfies
\begin{align*}
\left\|\bm{\eta}_{\bar{k}_{\text{O}}^{1}|\bar{k}_{\text{O}}^{1}}\right\|\le \delta +\hat{c}+\frac{\hat{c}\mathfrak{a}\left(\Gamma_{0|0}\right)\mathfrak{q}\left(\Gamma_{0|0}\right)}{1-\mathfrak{q}\left(\Gamma_{0|0}\right)}
\end{align*}
based on Proposition \ref{prop_o_bound}. As a consequence, Proposition \ref{prop_u_bound} gives
\begin{equation*}
\begin{split}
\bm{\rho}_{k|k}^{l}&>-\left(\delta +\hat{c}+\frac{\hat{c}\mathfrak{a}\left(\Gamma_{0|0}\right)\mathfrak{q}\left(\Gamma_{0|0}\right)}{1-\mathfrak{q}\left(\Gamma_{0|0}\right)}\right)\left(c_0+\left(n-2\right)\mathfrak{c}\left(\Gamma_{\underline{k}_{\text{U}}^{1}|\underline{k}_{\text{U}}^{1}}\right)\right),\\
\bm{\rho}_{k|k}^{l}&<\varrho_{\text{m}}+\left(\delta +\hat{c}+\frac{\hat{c}\mathfrak{a}\left(\Gamma_{0|0}\right)\mathfrak{q}\left(\Gamma_{0|0}\right)}{1-\mathfrak{q}\left(\Gamma_{0|0}\right)}\right)\left(c_0+\left(n-2\right)\mathfrak{c}\left(\Gamma_{\underline{k}_{\text{U}}^{1}|\underline{k}_{\text{U}}^{1}}\right)\right),
\end{split}
\end{equation*}
for all $l\in\{1,\cdots,n\}$ and $k\in(\underline{k}_{\text{U}}^{1},\bar{k}_{\text{U}}^{1}]$. Consequently, the estimation error satisfies $\left\|\bm{\eta}_{k|k}\right\|\le\mathfrak{e}\left(\delta,\Gamma_{0|0},\Gamma_{\underline{k}_{\text{U}}^{1}|\underline{k}_{\text{U}}^{1}}\right)$ for all $k\in(\underline{k}_{\text{U}}^{1},\bar{k}_{\text{U}}^{1}]$.

\noindent{\textit{(b) For the $1^{\text{st}}$ observable time interval $k\in(\underline{k}_{\emph{\text{O}}}^{1},\bar{k}_{\emph{\text{O}}}^{1}]$}:} Since the section is observable at time $0$, it holds that $\underline{k}_{{\text{O}}}^{1}=0$. In this case, we have $\Gamma_{\underline{k}_{{\text{O}}}^{1}|\underline{k}_{{\text{O}}}^{1}}=\Gamma_{0|0}$ and $\left\|\bm{\eta}_{\underline{k}_{{\text{O}}}^{1}|\underline{k}_{{\text{O}}}^{1}}\right\|=\left\|\bm{\eta}_{0|0}\right\|\le\sqrt{n}\varrho_{\text{m}}$. Then we can directly apply Proposition \ref{prop_o_bound} and conclude that
\begin{align*}
\left\|\bm{\eta}_{k|k}\right\|
\le\max\left\{\hat{c}+\mathfrak{a}\left(\Gamma_{0|0}\right)\mathfrak{q}\left(\Gamma_{0|0}\right) \sqrt{n}\varrho_{\text{m}},\hat{c}+\frac{\hat{c}\mathfrak{a}\left(\Gamma_{0|0}\right)\mathfrak{q}\left(\Gamma_{0|0}\right)}{1-\mathfrak{q}\left(\Gamma_{0|0}\right)} \right\},
\end{align*}
for $k\in(\underline{k}_{\text{O}}^{1},\bar{k}_{\text{O}}^{1}]$.

We conclude the proof by combining the above three steps.


%

\ifCLASSOPTIONcaptionsoff
  \newpage
\fi



%

\bibliographystyle{IEEEtran}
\bibliography{KCF_proof}

%
\vspace{-0.4in}







\end{document}